\documentclass[onecolumn,a4paper,superscriptaddress,11pt,accepted=2017-08-12]{quantumarticle}
\pdfoutput=1
\usepackage[utf8]{inputenc}
\usepackage[english]{babel}
\usepackage{hyperref}
\usepackage[margin=1in]{geometry}
\usepackage{amsmath,amsfonts,amssymb,amsthm}
\usepackage{microtype}
\usepackage{graphicx}
\usepackage{varioref}
\usepackage{verbatim} 
\usepackage{bm}
\usepackage{multicol}
\usepackage{enumerate}
\usepackage[normalem]{ulem}
\usepackage{caption}
\usepackage[T1]{fontenc}
\usepackage{tikz}
\usetikzlibrary{patterns}
\usetikzlibrary{arrows}

\usepackage{mathrsfs}

\usepackage{url}
\usepackage{xspace}
\usepackage{thm-restate}

\usepackage{hyperref}
\hypersetup{
    bookmarksnumbered=true, % If Acrobat bookmarks are requested, include section numbers
    unicode=false, % non-Latin characters in Acrobat bookmarks
    pdfstartview={}, % fits the width of the page to the window
    pdftitle={Quantum Algorithms for Graph Connectivity
and Formula Evaluation}, % title
    pdfauthor={}, % author
    pdfsubject={}, % subject of the document
    pdfcreator={}, % creator of the document
    pdfproducer={}, % producer of the document
    pdfkeywords={}, % list of keywords
    pdfnewwindow=true, % links in new window
    colorlinks=true, % false: boxed links; true: colored links
    linkcolor=blue, % color of internal links
    citecolor=blue, % color of links to bibliography
    filecolor=blue, % color of file links
    urlcolor=blue % color of external links
}

\newtheorem{theorem}{Theorem}
\newtheorem{definition}[theorem]{Definition}
\newtheorem{lemma}[theorem]{Lemma}

\newtheorem{corollary}[theorem]{Corollary}

\newtheorem{claim}[theorem]{Claim}

\newcommand{\ket}[1]{|#1\rangle}
\newcommand{\bra}[1]{\langle#1|}

\DeclareMathAlphabet{\matheu}{U}{eus}{m}{n}

\newcommand{\braket}[2]{\langle{#1}|{#2}\rangle}
\newcommand{\ketbra}[2]{|{#1}\rangle\!\langle{#2}|}

\newcommand{\eps}{\varepsilon}

\newcommand{\depth}{d}
\newcommand{\subf}{l}
\newcommand{\edgeL}{\lambda}

\newcommand{\Ohm}{c}
\newcommand{\norm}[1]{\left\| #1 \right\|}

\def\tO{\widetilde{\mathrm{O}}}

\title{Quantum Algorithms for Graph Connectivity\\ and Formula Evaluation}
\date{\today}
\author{Stacey Jeffery}
\email{jeffery@cwi.nl}
\thanks{SJ completed parts of this work while at the Institute for Quantum Information and Matter (IQIM), Caltech}
\affiliation{QuSoft and CWI, Amsterdam, the Netherlands}
\author{Shelby Kimmel}
\email{shelby.kimmel@gmail.com}
\thanks{SK completed parts of this work while at the Joint Center for Quantum Information and Computer Science (QuICS),
University of Maryland}
\affiliation{Middlebury College, Middlebury, VT, USA}

\usepackage[numbers,sort&compress]{natbib}

\begin{document}

\maketitle

%-Title

\begin{abstract}
We give a new upper bound on the quantum query complexity of deciding $st$-connectivity on certain classes of planar graphs, and show the bound is sometimes  exponentially better than previous results. We then show Boolean formula evaluation reduces to deciding connectivity on just such a class of graphs. Applying the algorithm for $st$-connectivity to Boolean formula evaluation problems, we match the $O(\sqrt{N})$ bound on the quantum query complexity of evaluating formulas on $N$ variables, give a quadratic speed-up over the classical query complexity of a certain class of promise Boolean formulas, and show this approach can yield superpolynomial quantum/classical separations. These results indicate that this $st$-connectivity-based approach may be the ``right'' way of looking at quantum algorithms for formula evaluation.

\end{abstract}

%\tableofcontents

\section{Introduction}

Deciding whether two points are connected in a network is a problem of
significant practical importance. In this work, we argue that this problem, $st$-connectivity, is also important as a quantum algorithmic
primitive. 

D\"urr, Heiligman, H{\o}yer, and Mhalla designed a quantum algorithm for deciding
$st$-connectivity that requires $O(|V|^{3/2})$ queries to the adjacency matrix of a
graph on vertex set $V$ \cite{durr2006quantum}. Belovs and Reichardt later discovered
an especially elegant span-program-based quantum algorithm for this problem, which is time-efficient and requires only logarithmic space \cite{BR12}. Belovs and
Reichardt's algorithm improves on the query complexity of D\"urr et al.'s
algorithm when the connecting path is promised to be short (if it exists).

Belovs and Reichardt's $st$-connectivity algorithm has already been adapted or been used as a subroutine for deciding other graph problems, such as detecting certain subgraphs \cite{BR12}, 
deciding whether a graph is a forest \cite{cade2016time}, and deciding whether a graph is bipartite \cite{cade2016time}. 

In this work, we modify the span program algorithm used in \cite{BR12},
inheriting its space and time efficiency, and we restrict to deciding $st$-connectivity on a class of planar graphs. If the effective resistances of the set
of graphs in question (and their planar duals) are small,
then we find the quantum algorithm requires far fewer queries than suggested by the analysis
in \cite{BR12}. In fact, we obtain a  polynomial to constant improvement in query
complexity for some classes of graphs.

In addition to improving our understanding of the quantum query complexity of
$st$-connectivity problems, we show that Boolean formula evaluation reduces (extremely naturally) to
$st$-connectivity problems of the kind for which our improved analysis holds. Therefore, finding good algorithms for $st$-connectivity can lead to
 good algorithms for Boolean formula evaluation. 
While one might not expect that such a reduction would produce good algorithms, we find the reduction gives optimal performance for certain classes of Boolean formulas.

Boolean formula evaluation is a fundamental class of problems with wide-reaching implications in algorithms and complexity theory. Quantum speed-ups
for evaluating formulas like \textsc{or} \cite{grover1996fast} and the
\textsc{nand}-tree \cite{FGG07} spurred interest in better understanding the
performance of quantum algorithms for Boolean formulas. This research
culminated in the development of span program algorithms
\cite{RS12,R01}, which can have optimal quantum query complexity for any problem \cite{lee2011quantum}. Using span program algorithms, it was shown that $O(\sqrt{N})$
queries are sufficient for any read-once formula with $N$ inputs
\cite{reichardt2010span,lee2011quantum}. Classically, the query complexity of
evaluating \textsc{nand}-trees is $\Theta(N^{.753})$ \cite{SW86} and the query
complexity of evaluating arbitrary read-once formulas is $\Omega(N^{.51})$
\cite{heiman1991randomized}.

While there are simple bounds on the quantum query
complexity of total formula evaluation problems, promise versions are still not fully understood.
Kimmel \cite{K11} showed that for a certain promise
version of \textsc{nand}-trees, called \emph{$k$-fault trees}, the quantum
query complexity is $O(2^k)$, while Zhan, Kimmel, and Hassidim \cite{ZKH12} showed the classical query complexity is
$\Omega((\log\frac{\log N}{k})^k)$, giving a \emph{superpolynomial} quantum
speed-up for a range of values of $k$. More general treatment of when promises
on the inputs give superpolynomial query speed-ups can be found in
\cite{aaronson2015sculpting}.

Since our analysis of $st$-connectivity shows that graphs with small effective resistance can be decided efficiently, this in turn means that Boolean formula evaluation problems with the promise that their inputs correspond to low resistance graphs can also be evaluated efficiently. This result gives us new insight into the structure of quantum speed-ups for promise Boolean formulas.

\vspace{.25cm}
\noindent{{\bf{Contributions.}} 
We summarize the main results in this paper as follows:
\begin{itemize}
\item Improved quantum query algorithm for deciding $st$-connectivity when the input is a subgraph of some graph $G$ such that $G\cup\{\{s,t\}\}$ --- $G$ with an additional $st$-edge --- is planar.
\begin{itemize}
\item The analysis involves the effective resistance of the original graph and its planar dual.
\item We find families of graphs for which this analysis gives  exponential and polynomial improvements, respectively, over the previous quantum analysis in \cite{BR12}.
\end{itemize}
\item Algorithm for Boolean formula evaluation via reduction to $st$-connectivity.
\begin{itemize}
\item Using this reduction, we provide a simple proof of the fact that read-once Boolean formulas with $N$ input variables can be evaluated using $O(\sqrt{N})$ queries.
\item We show both a quadratic and a superpolynomial quantum-to-classical speed-up using this reduction, for certain classes of promise Boolean formula evaluation problems.
\end{itemize}
\end{itemize}

\vspace{.25cm}
\noindent{\bf{Open Problems.}}
We would like to have better bounds on the classical
query complexity of evaluating $st$-connectivity problems, as this would provide a
new approach to finding separations between classical and quantum query
complexity. Additionally, our reduction from Boolean formula evaluation to $st$-connectivity could be helpful in the design of new classical algorithms for formulas.

Another open problem concerns span programs in general:
 when can we view span programs as solving $st$-connectivity
problems? This could be useful for understanding
when span programs are time-efficient, since the time-complexity
analysis of $st$-connectivity span programs is straightforward (see 
Appendix \ref{app:time},
\cite[Section 5.3]{BR12}, \cite[Appendix B]{IJ15}).

An important class of $st$-connectivity-related span programs are those arising
from the learning graph framework, which provides a means of designing quantum
algorithms that is much simpler and more intuitive than designing a general
span program \cite{Bel11}. A limitation of this framework is its one-sidedness
with respect to 1-certificates: whereas a learning graph algorithm is designed
to detect 1-certificates, a framework capable of giving optimal quantum query
algorithms for any decision problem would likely treat 0- and 1-inputs
symmetrically. In our analysis of $st$-connectivity, 1-inputs and 0-inputs are on equal footing. This duality between 1- and 0-inputs in $st$-connectivity
problems could give insights into how to extend the learning graph framework to
a more powerful framework, without losing its intuition and relative
simplicity.

\vspace{.25cm}
\noindent{\bf{Organization:}} Section \ref{sec:Preliminaries} provides background information. In Section \ref{sec:stconn}, we describe our improved analysis of the span program algorithm for $st$-connectivity for subgraphs of graphs $G$ such that $G\cup\{\{s,t\}\}$ is planar. In Section \ref{sec:nandgraphs}, we show that every formula evaluation problem is equivalent to an $st$-connectivity problem. In Section \ref{sec:nand}, we apply these results to promise \textsc{nand}-trees, for which we are able to prove the most significant classical/quantum separation using our approach. Also in Section \ref{sec:nand}, we use these ideas to create an improved algorithm for playing the two-player game associated with a \textsc{nand}-tree.

\section{Preliminaries}\label{sec:Preliminaries}

\subsection{Graph Theory}
 For an undirected weighted multigraph $G$, let
$V(G)$ and $E(G)$ denote the vertices and edges of $G$ respectively. In this work, we will only consider undirected multigraphs, which we will henceforth often refer to as {\it graphs}. To refer to an edge in a multigraph, we will specify the endpoints, as well as a label $\edgeL$, so that an edge is written $(\{u,v\},\edgeL)$. Although the label $\edgeL$ will be assumed to uniquely specify the edge, we include the endpoints for convenience. 
Let $\overrightarrow{E}(G)=\{(u,v,\edgeL ):(\{u,v\},\edgeL )\in E(G)\}$ denote the set of directed edges of $G$. 
For a planar graph $G$ (with an implicit planar embedding) let $F(G)$ denote the faces of $G$. We call the infinite face of a planar graph the \emph{external face}.

For any graph $G$ with connected vertices $s$ and $t$, we can imagine a fluid flowing into $G$ at $s$, and traveling through the graph along its edges, until it all finally exits at $t$. The fluid will spread out along some number of the possible $st$-paths in $G$. Such a linear combination of $st$-paths is called an \emph{$st$-flow}.
More precisely:
\begin{definition}[Unit $st$-flow]
Let $G$ be an undirected weighted graph with $s,t\in V(G)$, and $s$ and $t$ connected. Then a \emph{unit $st$-flow} on $G$ is a function $\theta:\overrightarrow{E}(G)\rightarrow\mathbb{R}$ such that:
\begin{enumerate}
\item For all $(u,v,\lambda)\in \overrightarrow{E}(G)$, $\theta(u,v,\edgeL )=-\theta(v,u,\edgeL)$;
\item $\sum_{v,\lambda:(s,v,\lambda)\in \overrightarrow{E}}\theta(s,v,\lambda)=\sum_{v,\lambda:(v,t,\lambda)\in \overrightarrow{E}}\theta(v,t,\lambda)=1$; and 
\item for all $u\in V(G)\setminus\{s,t\}$, $\sum_{v,\lambda:(u,v,\lambda )\in \overrightarrow{E}}\theta(u,v,\lambda )=0$. 
\end{enumerate}
\end{definition}

\begin{definition}[Unit Flow Energy]
Given a unit $st$-flow $\theta$ on a graph $G$, the \emph{unit flow energy} is 
\begin{align}
J(\theta)=\sum_{(\{u,v\},\edgeL )\in {E(G)}}{\theta(u,v,\edgeL )^2}.
\end{align}

\end{definition}
\begin{definition}[Effective resistance] Let $G$ be a graph with $s,t\in V(G)$.
If $s$ and $t$ are connected in $G$, the \emph{effective resistance} is $R_{s,t}(G) = \min_\theta J(\theta)$, where $\theta$ runs over all unit $st$-flows. If $s$ and $t$ are not connected, $R_{s,t}(G)=\infty.$ 
\end{definition}

Intuitively, $R_{s,t}(G)$ characterizes ``how connected'' the vertices $s$ and $t$ are. The more, shorter paths connecting $s$ and $t$, the smaller the effective resistance. 
 
The effective resistance has many applications. 
In a random walk on $G$, $R_{s,t}(G)|E(G)|$ is equal to the \emph{commute time} between $s$ and $t$, or the expected time a random walker starting from $s$ takes to reach $t$ and then return to $s$ \cite{CRRST96,aldous2002reversible}. 
If $G$ models an electrical network in which each edge $e$ of $G$ is a unit resistor and a potential difference is applied between $s$ and $t$, then $R_{s,t}(G)$ corresponds to the resistance of the network, which determines the ratio of current to voltage in the circuit (see \cite{DS84}). We can extend these connections further by considering weighted edges. 
A \emph{network} consists of a graph $G$ combined with a positive real-valued \emph{weight} function $\Ohm:E(G)\rightarrow\mathbb{R}^+$. 
\begin{definition}[Effective Resistance with weights]
Let ${\cal N}=(G,\Ohm)$ be a network with $s,t\in V(G)$. The \emph{effective resistance of $\cal N$} is $R_{s,t}({\cal N})=\min_{\theta}\sum_{(\{u,v\},\lambda)\in E(G)}\frac{\theta(u,v,\lambda)^2}{c(\{u,v\},\lambda)}$, where $\theta$ runs over all unit $st$-flows. 
\end{definition}
In a random walk on a network, which models any reversible Markov chain, a walker at vertex $u$ traverses edge $(\{u,v\},\lambda)$ with probability proportional to $\Ohm(\{u,v\},\lambda)$. Then the commute time between $s$ and $t$ is $R_{s,t}({\cal N})\sum_{e\in E(G)}\Ohm(e)$. When $\cal N$ models an electrical network in which each edge $e$ represents a resistor with resistance $1/c(e)$, then $R_{s,t}({\cal N})$ corresponds to the resistance of the network. 

 When $G$ is a single edge $e=(\{s,t\},\lambda)$ with weight $c(e)$, then the resistance $R_{s,t}(G)=1/c(e)$. 
When calculating effective resistance, $R_{s,t}$, we use the rule that for edges in series (i.e., a path), or more generally, graphs connected in series, resistances add. Edges in parallel, or more generally, graphs connected in parallel, follow the rule that conductances in parallel add, where the \emph{conductance} of a graph is given by one over the resistance. (The conductance of an edge $e$ is equal to $c(e),$ the weight of the edge.) More precisely, it is easy to verify the following:

\begin{claim}\label{claim:parallel_series}
Let two networks ${\cal N}_1=(G_1,c_1)$ and ${\cal N}_2=(G_2,c_2)$ each have nodes $s$ and $t$. If we create a new graph $G$ by identifying the $s$ nodes and the $t$ nodes (i.e. connecting the graphs in parallel) and define $c:E(G)\rightarrow\mathbb{R}^+$ by $c(e)=c_1(e)$ if $e\in E(G_1)$ and $c(e)=c_1(e)$ if $e\in E(G_2)$, then
\begin{align}
R_{s,t}(G,c)=\left(\frac{1}{R_{s,t}(G_1,c_1)}+\frac{1}{R_{s,t}(G_2,c_2)}\right)^{-1}.\label{eq:parallel}
\end{align}
However, if we create a new graph $G$ by identifying the $t$ node of $G_1$ with the $s$ node of $G_2$, relabeling this node $v\not\in\{s,t\}$ (i.e. connecting the graphs in series) and define $c$ as before, then
\begin{align}
R_{s,t}(G,c)=R_{s,t}(G_1,c_1)+R_{s,t}(G_2,c_2).\label{eq:series}
\end{align}
\end{claim}
\noindent As a bit of foreshadowing, if we let $R_{s,t}(G_1,c_1)$ and $R_{s,t}(G_2,c_2)$ take values $0$, representing $\textsc{false}$, or $\infty$,  representing $\textsc{true}$, then clearly \eqref{eq:series} computes the function $\textsc{or}$, since $0+0=0$, and $0+\infty=\infty+0=\infty+\infty=\infty$. We also have that \eqref{eq:parallel} computes the \textsc{and} function, if we use $\frac{1}{0}=\infty$ and $\frac{1}{\infty}=0$. 

\begin{definition}[$st$-cut]
Given a graph $G$ with $s,t\in V(G)$, if $s$ and $t$ are not connected, an \emph{$st$-cut} is a function 
$\kappa:V(G)\rightarrow\{0,1\}$ such that $\kappa(s)=1$, $\kappa(t)=0$, 
and $\kappa(v)- \kappa(u)=0$ whenever $\{u,v\}\in E(G)$.
\end{definition}
\noindent In other words, 
$\kappa$ defines a subset $S\subset V(G)$ such that $s\in S$, $t\not\in S$, and there is no edge of $G$
with one endpoint in $S$, and one endpoint in $\overline{S}$. An $st$-cut is a witness that $s$ and $t$ are in different components of $G$, so no path exists between $s$ and $t.$

Finally, we consider dual graphs:
\begin{definition}[Dual Graph]
Let $G$ be a planar graph (with an implicit embedding). The \emph{dual} graph, $G^\dagger$, is
defined as follows. For every face $f\in F(G)$, $G^\dagger$ has a
vertex $v_f$, and any two vertices are adjacent if their corresponding
faces share an edge, $e$. We call the edge between two such vertices the
\emph{dual edge} to $e$, $e^\dagger$. By convention, $e$ and $e^\dagger$ will always have the same label, so that if $e=(\{u,v\},\lambda)$, then $e^\dagger=(\{v_f,v_{f'}\},\lambda)$ for $f$ and $f'$ the faces of $G$ on either side of the edge $e$. 
\end{definition}

\subsection{Span Programs and Quantum Query Algorithms} 
Span programs \cite{KW93} were first introduced to the study of
quantum algorithms by Reichardt and \v{S}palek \cite{RS12}. They have
since proven to be immensely important for designing quantum
algorithms in the query model.

\begin{definition}[Span Program]\label{def:span}
A span program $P=(H,U,\tau,A)$ on $\{0,1\}^N$ is made up of 
{\bf{(I)}} finite-dimensional inner product spaces $H=H_1\oplus \dots \oplus H_N$, and $\{H_{j,b}\subseteq H_j\}_{j\in [N],b\in \{0,1\}}$ such that $H_{j,0}+H_{j,1}=H_j$, {\bf{(II)}} a vector space $U$, {\bf{(III)}} a non-zero \emph{target vector} $\tau\in U$, and {\bf{(IV)}} a linear operator $A:H\rightarrow U$.
For every string $x\in \{0,1\}^N$, we associate the subspace $H(x):=H_{1,x_1}\oplus \dots\oplus H_{N,x_N}$, and an operator $A(x):=A\Pi_{H(x)}$, where $\Pi_{H(x)}$ is the orthogonal projector onto~$H(x)$. 
\end{definition}

%%%%%%%%%%%%%%%%%%%%%%%%%%%%%%

\begin{definition}[Positive and Negative Witness]\label{def:posNegWit}
Let $P$ be a span program on $\{0,1\}^N$ and let $x$ be a string $x\in \{0,1\}^N$.
Then we call $\ket{w}$ a \emph{positive witness for $x$ in $P$} if
$\ket{w}\in H(x)$, and $A\ket{w}=\tau$. We define the \emph{positive
witness size of $x$} as:
\begin{equation}
w_+(x,P)=w_+(x)=\min\{\norm{\ket{w}}^2: \ket{w}\in H(x),A\ket{w}=\tau\},
\end{equation}
if there exists a positive witness for $x$, and $w_+(x)=\infty$ otherwise.
Let $\mathcal{L}(U,\mathbb{R})$ denote the set of linear maps from $U$ to $\mathbb{R}.$ We call a linear map $\omega\in\mathcal{L}(U,\mathbb{R})$ a \emph{negative
witness for $x$ in $P$} if $\omega A\Pi_{H(x)} = 0$ and $\omega\tau =
1$. We define the \emph{negative witness size of $x$} as:
\begin{equation}
w_-(x,P)=w_-(x)=\min\{\norm{\omega A}^2:{\omega\in \mathcal{L}(U,\mathbb{R}), 
\omega A\Pi_{H(x)}=0, \omega\tau=1}\},
\end{equation}
if there exists a negative witness, and $w_-(x)=\infty$ otherwise.
If $w_+(x)$ is finite, we say that $x$ is \emph{positive} (wrt.\ $P$), 
and if $w_-(x)$ is finite, we say that $x$ is \emph{negative}. We let 
$P_1$ denote the set of positive inputs, and $P_0$ the set of negative 
inputs for $P$.  In this way, the span program defines a partition $(P_0,P_1)$ of $[N]$.
\end{definition}

For a function $f:X\rightarrow \{0,1\}$, with $X\subseteq\{0,1\}^N$, we say
 $P$ \emph{decides} $f$ if $f^{-1}(0)\subseteq P_0$ and
$f^{-1}(1)\subseteq P_1$. We can use $P$ to design a
quantum query algorithm that decides $f$, 
given access to the input $x\in X$ via queries of the form
$\mathcal{O}_x:\ket{i,b}\mapsto \ket{i,b\oplus x_i}$.

\begin{theorem}[\cite{Rei09}]\label{thm:span-decision}
Fix $X\subseteq\{0,1\}^N$ and $f:X\rightarrow\{0,1\}$, and let $P$ be a span program on $\{0,1\}^N$ that decides $f$. 
Let $W_+(f,P)=\max_{x\in f^{-1}(1)}w_+(x,P)$ 
and $W_-(f,P)=\max_{x\in f^{-1}(0)}w_-(x,P)$. 
Then there is a bounded error quantum algorithm that decides $f$ with quantum query complexity $O(\sqrt{W_+(f,P)W_-(f,P)})$.
\end{theorem}

\subsection{Boolean Formulas}\label{sec:PreliminariesFormulas} A read-once Boolean formula can be expressed as a rooted tree in which the leaves are uniquely labeled by variables, $x_1,\dots,x_N$, and the internal nodes are labeled by gates from the set $\{\wedge,\vee,\neg\}$. Specifically, a node of degree 2 must be labeled by $\neg$ (\textsc{not}), whereas higher degree nodes are labeled by $\wedge$ (\textsc{and}) or $\vee$ (\textsc{or}), with the \emph{fan-in} of the gate being defined as the number of children. The \emph{depth} of a Boolean formula is the largest distance from the root to a leaf. We define an \textsc{and}-\textsc{or} \emph{formula} (also called a \emph{monotone formula}) as a read-once Boolean formula for which every internal node is labeled by $\wedge$ or $\vee$. Restricting to \textsc{and}-\textsc{or} formulas does not lose much generality, since for any formula, there is an equivalent formula in which all $\textsc{not}$-gates are at distance one from a leaf, and such $\textsc{not}$ gates
do not affect the query complexity of the formula. 
Moreover, although we only consider read-once formulas here, our techniques can be applied to more general formulas in which a single variable may label multiple leaves, since this is equivalent to a larger read-once formula with a promise on the input. Hereafter, when we refer to a \emph{formula}, we will mean an \emph{\textsc{and}-\textsc{or} read-once formula}.

In a slight abuse of notation, at times $x_i$ will denote a Boolean variable, and at times, it will denote a bit instantiating that variable. If $x\in\{0,1\}^N$ is an instantiation of all variables labeling the leaves of a formula $\phi$, then $\phi(x)$ is the value of $\phi$ on that input, defined as follows. If $\phi=x_i$ has depth 0, then $\phi(x)=x_i$. If $\phi$ has depth greater than 0, we can express $\phi$ recursively in terms of subformulas $\phi_1,\dots,\phi_\subf$, as $\phi=\phi_1\wedge\dots\wedge\phi_\subf$, if the root is labeled by $\wedge$, or $\phi=\phi_1\vee\dots\vee\phi_\subf$, if the root is labeled by $\vee$. In the former case, we define $\phi(x)=\phi_1(x)\wedge\dots\wedge\phi_\subf(x)$, and in the latter case, we define $\phi(x)=\phi_1(x)\vee \dots\vee\phi_\subf(x)$. 
A family of formulas $\phi=\phi_N$ on $N$ variables gives rise to an evaluation problem, $\textsc{Eval}_\phi$, in which the input is a string $x\in \{0,1\}^N$, and the output is $\phi_N(x)$. If $\phi(x)=0$, we say $x$ is a $0$-instance, and if $\phi(x)=1$, $x$ is a 1-instance. 
By $\phi_1\circ \phi_2$, we mean $\phi_1$ composed with $\phi_2$. That is, if $\phi_1:\{0,1\}^{N_1}\rightarrow\{0,1\}$ and $\phi_2:\{0,1\}^{N_2}\rightarrow\{0,1\}$, then $\phi_1\circ\phi_2:\{0,1\}^{N_1N_2}\rightarrow\{0,1\}$ evaluates as $\phi_1\circ\phi_2(x)=\phi_1(\phi_2(x^1),\dots,\phi_2(x^{N_1}))$, where $x=(x^1,\dots,x^{N_1})$ for $x^i\in \{0,1\}^{N_2}$.

An important formula evaluation problem is \textsc{nand}-tree evaluation. A \textsc{nand}-tree is a full binary tree
of arbitrary depth $\depth$ --- that is, every internal node has two children, and every leaf node is at distance $\depth$ from the root ---
in which an internal node is labeled by  $\vee$ if it is at even distance from the
leaves, or $\wedge$ if it is at odd
distance from the leaves. We use $\textsc{nand}_\depth$ to denote a \textsc{nand}-tree
of depth $\depth$.
While $\textsc{nand}_\depth$ is sometimes defined as a Boolean formula of \textsc{nand}-gates composed to
depth $\depth$, we will instead think of the formula as alternating
\textsc{and}-gates and \textsc{or}-gates --- when $\depth$ is even, these two characterizations are identical. An instance of \textsc{nand}$_d$ is a binary
string $x\in \{0,1\}^{N}$, where $N=2^{\depth}$.  For example, the formula $\textsc{nand}_2(x_1,x_2,x_3,x_4)
=(x_1\wedge x_2)\vee (x_3\wedge x_4)$ is a \textsc{nand}-tree of depth $2$.
$\textsc{nand}_0$ denotes the single-bit identity function. 

A \textsc{nand}$_d$ instance $x\in \{0,1\}^{2^d}$ can be associated with a two-player game on the rooted binary tree that represents \textsc{nand}$_d$,
where the leaves take the values $x_i$, as in Figure \ref{fig:nandtree}. The game starts at the root
node, which we call the current node. In each round of the game, as
long as the current node is not a leaf, if the current node is at
even (respectively odd) distance from the leaves, Player $A$ (resp.\ Player $B$) chooses one
of the current node's children to
become the current node. When the current node is a leaf, if the leaf
has value $1$, then Player $A$ wins, and if the leaf has value $0$, then
Player $B$ wins. The sequence of moves by the two players determines a path
from the root to a leaf.

A simple inductive argument shows that if $x$ is a 1-instance of \textsc{nand}-tree, then there exists a strategy by which Player
$A$ can always win, no matter what strategy $B$ employs; and if $x$ is a $0$-instance, there exists a strategy by which Player $B$ can always win. 
We say an input $x$ is $A$-winnable if it has value 1 and
$B$-winnable if it has value $0$.

\section{Improved Analysis of st-connectivity Algorithm}\label{sec:stconn}

In this section, we give an improved bound on the runtime of a quantum algorithm for $st$-connectivity on subgraphs of $G$, where $G\cup\{\{s,t\}\}$ is planar. 

Let $st$-\textsc{conn}$_{G,D}$ be a problem parameterized by a family of multigraphs $G$, which takes as input a string $x\in D$ where $ D\subseteq \{0,1\}^{E(G)}$. An input $x$ defines a subgraph $G(x)$ of $G$ by including the edge $e$ if and only if $x_e=1$. For all $x\in D$, $st$-\textsc{conn}$_{G,D}(x)=1$ if and only if there exists a path connecting $s$ and $t$ in $G(x)$. We write $st$-\textsc{conn}$_G$ when $D=\{0,1\}^{E(G)}$. A quantum algorithm for $st$-\textsc{conn}$_{G,D}$ accesses the input via queries to a standard quantum oracle $O_x$, defined $O_x\ket{e}\ket{b}=\ket{e}\ket{b\oplus x_e}$.

The authors of~\cite{BR12} present a quantum query algorithm for $st$-\textsc{conn}$_G$ when $G$ is a complete graph, which is easily extended to any multigraph $G$. We further generalize their algorithm to depend on some weight function $c:E(G)\rightarrow\mathbb{R}^+$ (a similar construction is also implicit in \cite{Bel11}). We call the following span program $P_{G,c}$:
\begin{align}
\forall e\in \overrightarrow{E}(G):\; H_{e,0}=\{0\},\quad H_{e,1}=\mathrm{span}\{\ket{e}\},\qquad\qquad\qquad H=\mathrm{span}\{\ket{e}:e\in \overrightarrow{E}(G)\}\label{eq:P}\qquad\nonumber\\
U= \mathrm{span}\{\ket{u}:u\in V(G)\},\; \tau=\ket{s}-\ket{t}, \; A=\sum_{(u,v,\edgeL )\in \overrightarrow{E}(G)}\sqrt{c(\{u,v\},\lambda)}(\ket{u}-\ket{v})\bra{u,v,\edgeL }.
\end{align}
For any choice of weight function $\Ohm$, this span program decides $st$-\textsc{conn}$_{G}$, but as we will soon see, the choice of $c$ may impact the complexity of the resulting algorithm. 

Using $P_{G,c}$ with $c(\{u,v\},\lambda)=1$ for all $(\{u,v\},\lambda)\in E(G)$, the authors of Ref.~\cite{BR12} show that the query complexity of evaluating $st$-\textsc{conn}$_{G,D}$ is
\begin{align}\label{eq:theirPreBound}
O\left(\sqrt{\max_{x\in D: s,t \textrm{ are connected}}R_{s,t}(G(x))\times|E(G)|}\right).
\end{align}
Their analysis was for the case where $G$ is a complete graph, but it is easily seen to apply to more general multigraphs $G$. In fact, it is straightforward to show that this bound can be improved to 
\begin{align}\label{eq:theirbound}
O\left(\sqrt{\max_{x\in D: s,t \textrm{ are connected}}R_{s,t}(G(x))\times\max_{x\in D: s,t \textrm{ are not connected}}\left(C_{s,t}(G(x))\right)}\right).
\end{align}
where
\begin{align}
C_{s,t}(G(x)) = \begin{cases}
\displaystyle\min_{\kappa:\kappa\textrm{ is an }st\textrm{-cut of }G(x)}\sum_{(\{u,v\},\lambda)\in E(G)}|\kappa(u)-\kappa(v)| &\textrm{ if $s$ and $t$ not connected}\\
\infty &\textrm{ otherwise. }
\end{cases}
\end{align}

In particular, when $G$ is a complete graph on vertex set $V$, with the promise that if an $st$-path exists, it is of length at most $k$, 
Eq.\ \eqref{eq:theirPreBound} gives a bound of $O\left(\sqrt{k}|V|\right).$
In the worst case, when $k=|V|$, the analysis of \cite{BR12} does not improve on the previous quantum algorithm of \cite{durr2006quantum}, which gives a bound of $O(|V|^{3/2}$).

In this paper, we consider in particular multigraphs that are planar even when an additional $st$-edge is added (equivalently, there exists a planar embedding in which $s$ and $t$ are on the same face), as in graph $G$ in Figure \ref{fig:dual2}. (In the case of
Figure \ref{fig:dual2}, $s$ and $t$ are both on the external face.) Given such
a graph $G$, we define three other related graphs, which we denote by
$\overline{G}$, $\overline{G}^\dagger$, and $G'$. 

We first define the graph $\overline{G}$, which is the same as $G$, but with an
extra edge labeled by $\emptyset$ connecting $s$ and $t.$ 
We then denote by $\overline{G}^\dagger$ the planar dual of $\overline{G}$.
Because every planar dual has one edge crossing each edge of the original
graph, there exists an edge that is dual to $(\{s,t\},\emptyset)$, also labeled by
$\emptyset$. We denote by $s'$ and $t'$ the two vertices at the
endpoints of $(\{s,t\},\emptyset)^\dagger = (\{s',t'\},\emptyset)$.
Finally, we denote by $G'$ the graph $\overline{G}^\dagger$ except with the edge $(\{s',t'\},\emptyset)$ removed. 

\begin{figure}[ht]
\centering
\begin{tikzpicture}[scale = .95]
\node at (0,0) {\begin{tikzpicture}[scale=.8]
\filldraw (0,1) circle (.1);
\filldraw (0,-1) circle (.1);
\filldraw (-.7,0) circle (.1);

\draw (0,1)--(0,-1);
\draw plot [smooth] coordinates{(0,1) (-.7,0) (0,-1)};
\draw plot [smooth] coordinates{(0,1) (.7,0) (0,-1)};

\node at (.4,1) {$s$};
\node at (.4,-1) {$t$};
\node at (-.7,.6) {$1$};
\node at (-.7,-.6) {$2$};
\node at (.3,0) {$3$};
\node at (1,0) {$4$};

\node at (0,-1.5) {$G$};
\end{tikzpicture}};

\node at (4,0) {\begin{tikzpicture}[scale=.8]
\filldraw (0,1) circle (.1);
\filldraw (0,-1) circle (.1);
\filldraw (-.7,0) circle (.1);

\draw (0,1)--(0,-1);
\draw plot [smooth] coordinates{(0,1) (-.7,0) (0,-1)};
\draw plot [smooth] coordinates{(0,1) (.7,0) (0,-1)};

\node at (.4,1) {$s$};
\node at (.4,-1) {$t$};
\node at (-.7,.6) {$1$};
\node at (-.7,-.6) {$2$};
\node at (.3,0) {$3$};
\node at (1,0) {$4$};

\draw[dashed] plot[smooth] coordinates {(0,1) (-1,1) (-1.5,0) (-1,-1) (0,-1)};

\node at (-1.7,0) {$\emptyset$};

\node at (0,-1.5) {$\overline{G}$};
\end{tikzpicture}};

\node at (8.5,0) {\begin{tikzpicture}[scale=.85]
\filldraw[gray] (0,1) circle (.1);
\filldraw[gray] (0,-.9) circle (.1);
\filldraw[gray] (-.7,0) circle (.1);

\draw[gray] (0,1)--(0,-.9);
\draw[gray] plot [smooth] coordinates{(0,1) (-.7,0) (0,-.9)};
\draw[gray] plot [smooth] coordinates{(0,1) (.7,0) (0,-.9)};

\node[white] at (0,1.25) {$s$};

\draw[gray] plot[smooth] coordinates {(0,1) (-1,1) (-1.5,0) (-1,-1) (0,-.9)};

\filldraw (-.3,0) circle (.1);
\filldraw (-1.1,0) circle (.1);
\filldraw (.3,0) circle (.1);
\filldraw (1,0) circle (.1);

\draw plot[smooth] coordinates {(-1.1,0) (-.7,.25) (-.3,0)};
\draw plot[smooth] coordinates {(-1.1,0) (-.7,-.25) (-.3,0)};
\draw (-.3,0) -- (1,0);
\draw plot[smooth] coordinates {(-1.1,0) (0,1.25) (1,0)};

\node at (-1.1,-.35) {$s'$};
\node at (1,-.35) {$t'$};

\node at (-1.7,0) {\color{gray}$\emptyset$};
\node at (.8,.9) {$\emptyset$};

\node at (0,-1.3) {$\overline{G}^\dagger$};
\end{tikzpicture}};

\node at (12,0) {\begin{tikzpicture}[scale=.8]
\node[white] at (0,1.25) {$s$};

\filldraw (-.3,0) circle (.1);
\filldraw (-1.1,0) circle (.1);
\filldraw (.3,0) circle (.1);
\filldraw (1,0) circle (.1);

\draw plot[smooth] coordinates {(-1.1,0) (-.7,.25) (-.3,0)};
\draw plot[smooth] coordinates {(-1.1,0) (-.7,-.25) (-.3,0)};
\draw (-.3,0) -- (1,0);

\node at (-1.35,0) {$s'$};
\node at (1.35,0) {$t'$};
\node at (-.7, .6) {$1$};
\node at (-.7,-.6) {$2$};
\node at (0,.25) {$3$};
\node at (.7,.25) {$4$};

\node at (0,-1.3) {${G}'$};
\end{tikzpicture}};

\end{tikzpicture}
\caption{Example of how to derive $\overline{G},$ $\overline{G}^\dagger,$ and $G'$ from a planar graph $G$ where $s$ and $t$ are on the same face. $\overline{G}$ is obtained from $G$ by adding an edge $(\{s,t\},\emptyset)$. $\overline{G}^\dagger$ is the planar dual of $\overline{G}$. (In the diagram labeled by $\overline{G}^\dagger$, $\overline{G}$ is the gray graph, while $\overline{G}^\dagger$ is black). $G'$ is obtained from $\overline{G}^\dagger$ by removing the edge $(\{s,t\},\emptyset)^\dagger$. Note that dual edges inherit their labels (in this case $1,2,3,4,\emptyset$) from the primal edge.} \label{fig:dual2}
\end{figure}
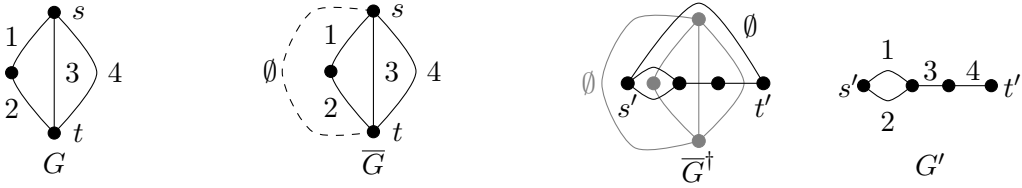

By construction, $G'$ will always have the same number of edges as $G$. Then as $x$ defines a subgraph $G(x)$ of $G$ by including the edge $e$ if and only if $x_e=1$, we let $G'(x)$ be the subgraph of $G'$ where we include the edge $e^\dagger$ if and only if $x_e=0$.

If there is no path from $s$ to $t$ in $G(x)$, there must be a cut between $s$ and $t.$ Note 
that for any $e\in E(G)$, $e\in E(G(x))$ if and only if 
$e^\dagger\not\in E(G'(x))$. 
Looking at Figure \ref{fig:dual2}, 
one can convince oneself that any $s't'$-path in $G'(x)$ defines an $st$-cut in $G(x)$: simply define $\kappa(v)=1$ for vertices \emph{above} the path, and $\kappa(v)=0$ for vertices \emph{below} the path. 

Let $c$ be a weight function on $E(G)$. Then we define a weight function $c'$ on $E(G')$ as $c'(e^\dagger) = 1/c(e).$  Then for every $x$ there will be a path either from $s$ to $t$ in $G(x)$ (and hence $R_{s,t}(G(x),c)$) will be finite), or a path from $s'$ to $t'$ in $G'(x)$ (in which case $R_{s',t'}(G'(x),c')$ will be finite).

We can now state our main lemma:

\begin{restatable}{lemma}{poswit}\label{lemma:both_witnesses}
Let $G$ be a planar multigraph with $s,t\in V(G)$ such that $G\cup\{\{s,t\}\}$ is also planar, and let $c$ be a weight function on $E(G)$. Let $x\in\{0,1\}^{E(G)}$.
Then $w_+(x,P_{G,c})=\frac{1}{2}R_{s,t}(G(x),c)$ and $w_-(x, P_{G,c})={2}R_{s',t'}(G'(x),c')$. {}
\end{restatable}

Using Lemma \ref{lemma:both_witnesses} and Theorem \ref{thm:span-decision}, we immediately have the following:
\begin{theorem}\label{thm:stconn}
Let $G$ be a planar multigraph with $s,t\in V(G)$ such that $G\cup\{\{s,t\}\}$ is also planar. Then the bounded error quantum query complexity of evaluating $st$-\textsc{conn}$_{G,D}$ is
\begin{align}\label{eq:ourBound}
O\left(\min_c\sqrt{\max_{x\in D: st\textsc{-conn}_G(x)=1}R_{s,t}(G(x), c)\times\max_{x\in D: st\textsc{-conn}_G(x)=0}R_{s',t'}(G'(x),c')}\right)
\end{align}
where the minimization is over all positive real-valued functions $c$ on $E(G)$.
\end{theorem}

While it might be difficult in general to find the optimal edge weighting $c$, any choice of $c$ will at least give an upper bound on the query complexity. However, as we will see, sometimes the structure of the graph will allow us to efficiently find good weight functions.

The proof of Lemma \ref{lemma:both_witnesses} is in Appendix \ref{app:st-proofs}. The positive witness result follows from generalizing the proof in \cite{BR12} to weighted multigraphs. The idea is that an $st$-path witnesses that $s$ and $t$ are connected, as does any linear combination of such paths --- i.e. an $st$-flow. The effective resistance $R_{s,t}(G(x),c)$ characterizes the size of the smallest possible $st$-flow. 

Just as
a positive witness is some linear combination of $st$-paths, similarly, a negative witness turns out to be a linear combination of $st$-cuts in $G(x)$. But as we've argued, every $st$-cut corresponds to an $s't'$-path in $G'(x)$. Using the correspondence between cuts and paths, we have that a negative witness is a linear combination of $s't'$-paths in $G'(x)$. This allows us 
to show a correspondence between complexity-optimal negative witnesses 
and minimal $s't'$-flows, connecting $w_-(x,P_{G,c})$ to 
$R_{s',t'}(G'(x),c')$. 

In Appendix \ref{app:time}, we show that if a quantum walk step on the network $(G,c)$ can be implemented time efficiently, then this algorithm is not only query efficient, but also time efficient, with only $\frac{1}{\sqrt{\delta}}$ multiplicative overhead, where $\delta$ is the spectral gap of the symmetric normalized Laplacian. For example, if $G$ is a complete graph with unit weights, $\delta$ is constant. Let
\begin{align}
U_{G,c}:\ket{u}\ket{0}\mapsto \frac{1}{\sqrt{\sum_{v,\lambda:(u,v,\lambda)\in \overrightarrow{E}(G)}c(\{u,v\},\lambda)}}\sum_{v,\lambda:(u,v,\lambda)\in \overrightarrow{E}(G)}\sqrt{c(\{u,v\},\lambda)}\ket{u}\ket{u,v,\lambda}.
\end{align}
Then we show the following.\footnote{An earlier version of this work was missing the $1/\sqrt{\delta}$ term in the complexity, due to an error in the proof. We thank Arjan Cornelissen and Alvaro Piedrafita for finding this error and bringing it to our attention.}

\begin{restatable}{theorem}{timeComp}\label{thm:timeComp}
Let $P_{G,\Ohm}=(H,U,A,\tau)$ be defined as in \eqref{eq:P}. Let $S_{G,\Ohm}$ be an upper bound on the time complexity of implementing $U_{G,\Ohm}$, and $\delta$ the spectral gap of the symmetric normalized Laplacian of $(G,\Ohm)$. 
If $G$ has the property that $G\cup \{\{s,t\}\}$ is planar, then the time complexity of deciding $st$-\textsc{conn}$_{G,D}$ is at most
\begin{align}
\widetilde{O}\left(\min_\Ohm \frac{S_{G,\Ohm}}{\sqrt{\delta}}\sqrt{\max_{x\in D:s,t \textrm{ are connected}}R_{s,t}(G(x),\Ohm)\times \max_{x\in D:s,t\textrm{ are not connected}}R_{s',t'}(G'(x),\Ohm')}\right).
\end{align}
\end{restatable}

\noindent In Appendix \ref{app:time}, we also show that if the space complexity of implementing $U_{G,c}$ in time $S_{G,c}$ is $S_{G,c}'$, the algorithm referred to in Theorem \ref{thm:timeComp} has space complexity at most $O(\max\{\log|E(G)|,\log|V(G)|\}+S_{G,c}'+\log(1/\delta))$.

%%%%%%%%%%%%%%%%%%%%%%%%%%%%%%%%%%%%%%%%%%%%%%%%%%%%%%%%%%%%%%%%%%%%%%%%%%%%%%%%%%%%%%%%%%%%%%%%%%%%%%%%%%%%%%%%%%%%%%%%%%%%%%%%%%%%%%%%%%%%%%%%%%%%%%%%
%%%%%%%%%%%%%%%%%%%%%%%%%%%%%%%%%%%%%%%%%%%%%%%%%%%%%%%%%%%%%%%%%%%%%%%%%%%%
%%%%%%%%%%%%%%%%%%%%%%%%%%%%%%%%%%%%%%%%%%%%%%%%%%%%%%%%%%%%%%%%%%%%%%%%%%%%
%%%%%%%%%%%%%%%%%%%%%%%%%%%%%%%%%%%%%%%%%%%%%%%%%%%%%%%%%%%%%%%%%%%%%%%%%%%%
\subsection{Comparison to Previous Quantum Algorithm}

When $G\cup\{\{s,t\}\}$ is planar, our algorithm always matches or improves on the algorithm in \cite{BR12}. To see this, we compare Eqs. \eqref{eq:ourBound} and \eqref{eq:theirbound}, and  choose $\Ohm$ to have value $1$ on all edges of $G$. Then the first terms are the same in both bounds, so we only need to analyze the second term. However, using the duality between paths and cuts, we have
\begin{align}
C_{s,t}(G(x))=\left(\textrm{shortest path length from $s'$ to $t'$ in }G'(x)\right) \geq R_{s't'}(G'(x)).\label{eq:Cst}
\end{align} 
To obtain the inequality in Eq.\ \eqref{eq:Cst}, we create an $s't'$-flow on $G'(x)$
that has value one on edges on the shortest path from $s'$ to $t'$ and zero on all other edges. Such a flow
has unit flow energy equal to the shortest path. However, the true effective resistance
can only be smaller than this, because it is the minimum energy over all
possible $s't'$-flows.

We now present two simple examples where our algorithm and analysis do better than that of \cite{BR12}. In the first example, we highlight how the change from $C_{s,t}(G(x))$ to $R_{s',t'}(G(x))$ in the complexity gives us an advantage for some graphs. In the second example, we show that being able to choose a non-trivial weight function $\Ohm$ can give us an advantage for some graphs.

Let $G$ be an $st$-path of length $N$: i.e., $N+1$ vertices arranged in a line so that each vertex is connected to its neighbors to the left and right by a single edge, and $s$ and $t$ are the vertices on either end of the line, as in
Figure \ref{fig:exLine}. For some $h\in\{1,\dots,N\}$, let $D=\{1^N\}\cup\{x\in \{0,1\}^N: |{x}|\leq N-h\}$, where $1^N$ is the all-one string of length $N$, and $|x|$ is the hamming weight of the string $x$.

\begin{figure}[ht]
\centering
\begin{tikzpicture}[scale = .95]
\node at (0,0) {\begin{tikzpicture}[scale=.8]
\filldraw (-1,0) circle (.1);
\filldraw (0,0) circle (.1);
\filldraw (1,0) circle (.1);
\filldraw (2,0) circle (.1);
\filldraw (3,0) circle (.1);
\node at (2.5,0) {$\dots$};
\filldraw (4,0) circle (.1);
\draw (-1,0)--(2,0);
\draw (3,0)--(4,0);
\node at (2,-1.7) {$G$};
\node at (-1.2,.5) {$s$};
\node at (4.2,.5) {$t$};

\end{tikzpicture}};

\node at (6,.2) {\begin{tikzpicture}[scale=.8]
\filldraw (0,1) circle (.1);
\filldraw (0,-1) circle (.1);
\draw plot [smooth] coordinates{(0,1) (-.3,0) (0,-1)};
\draw plot [smooth] coordinates{(0,1) (-1.2,0) (0,-1)};
\draw plot [smooth] coordinates{(0,1) (.3,0) (0,-1)};
\draw plot [smooth] coordinates{(0,1) (1.2,0) (0,-1)};

\node at (.7,0) {$\dots$};

\node at (-.5,-1.5) {$G'$};
\node at (0,1.3) {$s'$};
\node at (.3,-1.2) {$t'$};
\end{tikzpicture}};

\end{tikzpicture}
\caption{Example of graph for which our analysis does better than the analysis of \cite{BR12}, even with $\Ohm=1$ for all edges, under the promise that $G'(x)$ always contains at least $h$ edges, if $s'$ and $t'$ are connected.}\label{fig:exLine}
\end{figure}
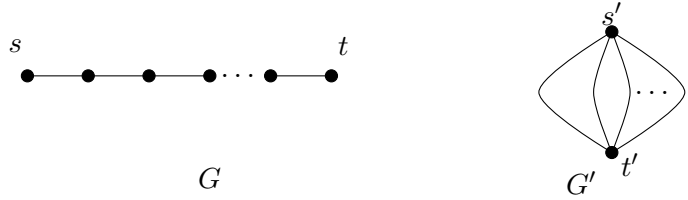

Then, choosing $\Ohm$ to have value $1$ on all edges of $G$, we have
\begin{align}
\max_{x\in D: st\textsc{-conn}_G(x)=1}R_{s,t}(G(x))= N
\end{align}
because the only $x\in D$ such that $s$ and $t$ are connected in $G(x)$ is $x=1^N$, in which case the only unit flow has value $1$ on each edge. This flow has energy $N$. However
\begin{align}
\max_{x\in D: st\textsc{-conn}_G(x)=0}R_{s',t'}(G'(x)) \leq 1/h,
\end{align}
because when $s$ and $t$ are not connected in $G(x)$, $G(x)$ has at most $N-h$ edges, so $G'(x)$ has at least $h$ edges. Thus we can define a unit flow with value $1/h$ on each of $h$ parallel edges in $G'(x)$, giving an energy of $1/h.$ On the other hand
\begin{align}
\max_{x\in D: st\textsc{-conn}_G(x)=0}C_{s,t}(G(x)) = 1.
\end{align}
In fact, since $C_{s,t}(G(x))$ counts the minimum number of edges $(\{u,v\},\lambda)$ across any cut (i.e. such that $\kappa(u)=1$ and $\kappa(v)=0$), it is always at least 1, for \emph{any} $G(x)$ in which an $st$-cut exists, whereas $R_{s',t'}(G'(x))$ can be as small as $1/N$ for some $G$. 

Choosing $h=\sqrt{N}$ in our example, and applying Eqs. (\ref{eq:theirbound}) and (\ref{eq:ourBound}), the analysis in \cite{BR12} gives a query complexity of $O(N^{1/2})$ while our analysis gives a query complexity of $O(N^{1/4})$. 
In Section \ref{sec:nandgraphs} we will show that this bound is tight. 

Now consider the graph $G$ in Figure \ref{fig:exBalloon}. It consists of $N$ edges in a line, connecting vertices $s,u_1,\dots,u_N$,  and then $N$ multi-edges between $u_N$ and $t$. We assign weights $c(e)=1$ for edges $e$ on the path from $s$ to $u_N$, and $c(e)=N^{-1}$ for all other edges.

\begin{figure}[ht]
\centering
\begin{tikzpicture}[scale = .95]
\node at (0,0) {\begin{tikzpicture}[scale=.8]
\filldraw (-1,0) circle (.1);
\filldraw (0,0) circle (.1);
\filldraw (1,0) circle (.1);
\filldraw (2,0) circle (.1);
\filldraw (3,0) circle (.1);
\node at (2.5,0) {$\dots$};
\filldraw (4,0) circle (.1);
\filldraw (5,0) circle (.1);
\draw (-1,0)--(2,0);
\draw (3,0)--(4,0);
\draw (4,0)--(5,0);
\draw plot [smooth] coordinates{(4,0) (4.5,2) (5,0)};
\draw plot [smooth] coordinates{(4,0) (4.5,-2) (5,0)};
\draw plot [smooth] coordinates{(4,0) (4.5,1) (5,0)};
\draw plot [smooth] coordinates{(4,0) (4.5,-1) (5,0)};
\node at (4.75,2.2) {$\bm{N^{-1}}$};
\node at (4.75,1.2) {$\bm{N^{-1}}$};
\node at (4.75,.2) {$\bm{N^{-1}}$};
\node at (4.75,-1.2) {$\bm{N^{-1}}$};
\node at (4.75,-2.2) {$\bm{N^{-1}}$};

\node at (-.5,.3) {$\bf{1}$};
\node at (0.5,.3) {$\bf{1}$};
\node at (1.5,.3) {$\bf{1}$};
\node at (2.5,.3) {$\bf{1}$};
\node at (3.5,.3) {$\bf{1}$};

\node at (2,-1.7) {$G$};
\node at (3.75,-.3) {$u_{N}$};
\node at (-1.2,-.3) {$s$};
\node at (0,-.3) {$u_1$};
\node at (1,-.3) {$u_2$};
\node at (5.2,-.3) {$t$};
\end{tikzpicture}};

\node at (6,.2) {\begin{tikzpicture}[scale=.8]
\filldraw (0,1) circle (.1);
\filldraw (0,-1) circle (.1);
\draw plot [smooth] coordinates{(0,1) (-.3,0) (0,-1)};
\draw plot [smooth] coordinates{(0,1) (-1.2,0) (0,-1)};
\draw plot [smooth] coordinates{(0,1) (.3,0) (0,-1)};
\draw plot [smooth] coordinates{(0,1) (1.2,0) (0,-1)};
\draw plot [smooth] coordinates{(0,1) (2,0) (0,-1)};
\filldraw (1,.6) circle (.1);
\filldraw (1,-.6) circle (.1);
\filldraw (1.6,.3) circle (.1);
\filldraw (1.6,-.3) circle (.1);

\node at (.7,0) {$\dots$};

\node at (-.5,-1.5) {$G'$};
\node at (0,1.3) {$s'$};
\node at (.3,-1.2) {$t'$};

\node at (-1.3,0) {\textbf{1}};
\node at (-.5,0) {\textbf{1}};
\node at (.1,0) {\textbf{1}};
\node at (1,0) {\textbf{1}};
\node at (.7,1.05) {\textbf{N}};
\node at (1.4,.75) {\textbf{N}};
\node at (2.3,0) {\textbf{N}};
\node at (1.4,-.75) {\textbf{N}};
\node at (.7,-1.05) {\textbf{N}};
\end{tikzpicture}};

\end{tikzpicture}
\caption{Example of graph for which our analysis does quadratically better than the analysis of \cite{BR12} by taking advantage of a non-trivial weight function $\Ohm$. The values of $\Ohm$ for each edge of $G$, and of $\Ohm'$ for each edge of $G',$ are shown in boldface.}\label{fig:exBalloon}
\end{figure}
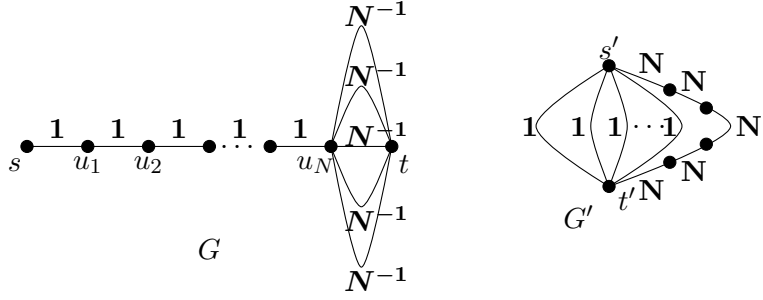

Then,
\begin{align}
\max_{x\in D:st\textsc{-conn}_G(x)=1}R_{s,t}(G(x),\Ohm)= 2N,
\end{align}
which occurs when only one of the multi-edges between $u_N$ and $t$ is present. In that case, the $N$ edges $\{s,u_1\},\{u_1,u_2\},\dots,\{u_{N-1},u_N\}$ each contribute 1 to the effective resistance, and the final edge between $u_N$ and $t$ contributes $\frac{1}{c(e)}=N$. 
Also
\begin{align}
\max_{x\in D: st\textsc{-conn}_G(x)=0}R_{s',t'}(G'(x), \Ohm') \leq 1,
\end{align}
where the maximum occurs when there is only one path from $s'$ to $t'$. (If it is the path with $N$ edges, each edge has weight $N$, and so contributes $1/N$ to the flow energy.) However
\begin{align}
\max_{x\in D: st\textsc{-conn}_G(x)=0}C_{s,t}(G(x)) = N
\end{align}
for a cut across the multi-edges between $u_N$ and $t$, and
\begin{align}
\max_{x\in D: st\textsc{-conn}_G(x)=1}R_{s,t}(G(x))= N+1,
\end{align}
which occurs when only one of the multi-edges between $u_N$ and $t$ is present.

Thus, the analysis in \cite{BR12} gives a query complexity of $O(N)$ while our analysis gives a query complexity of $O(N^{1/2})$.

In Section \ref{sec:nand} we will give an example where our analysis provides an  exponential improvement over the analysis in \cite{BR12}.

%%%%%%%%%%%%%%%%%%%%%%%%%%%%%%%%%%%%%%%%%%%%%%%%%%%%%%%%%%%%%%%%%%%%%%%%%

%%%%%%%%%%%%%%%%%%%%%%%%%%%%%%%%%%%%%%%%%%%%%%%%%%%%%%%%%%%%%%%%%%%%%%%%%
\section{AND-OR Formulas and st-Connectivity}\label{sec:nandgraphs}

In this section, we present a useful relationship between \textsc{and}-\textsc{or} formula evaluation problems and $st$-connectivity problems on certain graphs. As mentioned in Section \ref{sec:Preliminaries}, for simplicity we will restrict our analysis to read-once formulas, but the algorithm extends simply to ``read-many'' formulas. In this case, we will primarily be concerned with the query complexity: the input $x=(x_1,\dots,x_N)$ to a formula will be given via a standard quantum oracle $O_x$, defined $O_x\ket{i}\ket{b}=\ket{i}\ket{b\oplus x_i}$.

Given an \textsc{and}-\textsc{or} formula $\phi$ with $N$ variables, we will recursively construct a planar multigraph $G_\phi$, such that $G_\phi$ has two distinguished vertices labeled by $s$ and $t$ respectively, and every edge of $G_{\phi}$ is uniquely labeled by a variable $\{x_i\}_{i\in [N]}$. 
If $\phi=x_i$ is just a single variable, then $G_\phi$ is just a single edge with vertices labeled by $s$ and $t$, and edge label $x_i$. That is $E(G_\phi)=\{(\{s,t\},x_i)\}$ and $V(G_\phi)=\{s,t\}.$

Otherwise, suppose $\phi=\phi_1\wedge\dots\wedge\phi_\subf$. Then $G_\phi$ is the graph obtained from the graphs $G_{\phi_1},\dots,G_{\phi_\subf}$ by identifying the vertex labeled $t$ in $G_{\phi_i}$ with the vertex labeled $s$ in $G_{\phi_{i+1}}$, for all $i=1,\dots,\subf-1$, and labeling the vertex labeled $s$ in $G_{\phi_1}$ by $s$, and the vertex labeled $t$ in $G_{\phi_\subf}$ by $t$. That is, we connect the graphs $G_{\phi_1},\dots, G_{\phi_\subf}$ \emph{in series}, as in Figure \ref{fig:series-parallel}. (For a formal definition of $G_\phi$, see Appendix \ref{app:formula}). 

The only other possibility is that $\phi=\phi_1\vee \dots\vee\phi_\subf$. In that case, we construct $G_\phi$ by starting with $G_{\phi_1},\dots,G_{\phi_\subf}$ and  identifying all vertices labeled by $s$, and labeling the resulting vertex with $s$, and identifying all vertices labeled by $t$, and labeling the resulting vertex by $t$. That is, we connect $G_{\phi_1},\dots,G_{\phi_\subf}$ \emph{in parallel} (see Figure \ref{fig:series-parallel}). We note that graphs constructed in this way are exactly the set of \emph{series-parallel} graphs with two terminals (see e.g. \cite[Def. 3]{Valdes:1979:RSP:800135.804393}), and are equivalent to graphs without a $K_4$ minor \cite{dirac1952property,DUFFIN1965303}.

\begin{figure}[ht]
\centering
\begin{tikzpicture}
\node at (2,0) {\begin{tikzpicture}[scale = .9] %\phi_1
\draw[gray] plot [smooth] coordinates {(0,0) (.25,.5) (0,1)};
\draw[gray] plot [smooth] coordinates {(0,0) (-.25,.5) (0,1)};
\filldraw (0,0) circle (.1);
\filldraw (0,1) circle (.1);
\filldraw (.25,.5) circle (.1);

\node at (0,1.25) {$s$};
\node at (0,-.35) {$t$};
\node at (0,-1) {$G_{\phi_2}$};
\end{tikzpicture}};

\node at (0,0) {\begin{tikzpicture}[scale = .9] %\phi_2
\filldraw (0,0) circle (.1);
\filldraw (0,.5) circle (.1);
\filldraw (0,1) circle (.1);
\draw (0,0) -- (0,1);

\node at (0,1.25) {$s$};
\node at (0,-.35) {$t$};
\node at (0,-1) {$G_{\phi_1}$};
\end{tikzpicture}};

\node at (4,0) {\begin{tikzpicture}[scale = .9] %\phi_3
\draw[red] (0,0) -- (0,1);
\filldraw (0,0) circle (.1);
\filldraw (0,1) circle (.1);

\node at (0,1.25) {$s$};
\node at (0,-.35) {$t$};
\node at (0,-1) {$G_{\phi_3}$};
\end{tikzpicture}};

\node at (8,0) {\begin{tikzpicture}[scale = .9] %paralellel
\draw[gray] plot [smooth] coordinates {(0,0) (.25,.5) (0,1)};
\draw[gray] plot [smooth] coordinates {(0,0) (-.25,.5) (0,1)};
\draw plot [smooth] coordinates {(0,0) (-.5,.5) (0,1)};
\draw[red] plot [smooth] coordinates {(0,0) (.5,.5) (0,1)};
\filldraw (0,0) circle (.1);
\filldraw (0,1) circle (.1);
\filldraw (.25,.5) circle (.1);
\filldraw (-.5,.5) circle (.1);

\node at (0,1.25) {$s$};
\node at (0,-.35) {$t$};
\node at (0,-1) {$G_{\phi_1\vee \phi_2\vee \phi_3}$};
\end{tikzpicture}};

\node at (6,0) {\begin{tikzpicture}[scale = .9] %series 

\draw (0,3)--(0,2);
\draw[gray] plot [smooth] coordinates {(0,2) (-.25,1.5) (0,1)};
\draw[gray] plot [smooth] coordinates {(0,2) (.25,1.5) (0,1)};
\draw[red] (0,0)--(0,1);

\filldraw (0,3) circle (.1);
\filldraw (0,2.5) circle (.1);
\filldraw (0,2) circle (.1);
\filldraw (.25,1.5) circle (.1);
\filldraw (0,1) circle (.1);
\filldraw (0,0) circle (.1);

\node at (0,3.25) {$s$};
\node at (0,-.35) {$t$};
\node at (0,-.7) {$G_{\phi_1\wedge \phi_2\wedge \phi_3}$};

\end{tikzpicture}};

\end{tikzpicture}
	\caption{Let $\phi_1 = x_1\wedge x_2$, $\phi_2=x_3\vee (x_4\wedge x_5)$, and $\phi_3=x_6$. Then we obtain $G_{\phi_1\wedge \phi_2\wedge \phi_3}$ by connecting $G_{\phi_1}$, $G_{\phi_2}$, and $G_{\phi_3}$ in series, and $G_{\phi_1\vee \phi_2\vee\phi_3}$ by connecting them in parallel.}\label{fig:series-parallel}
\end{figure}
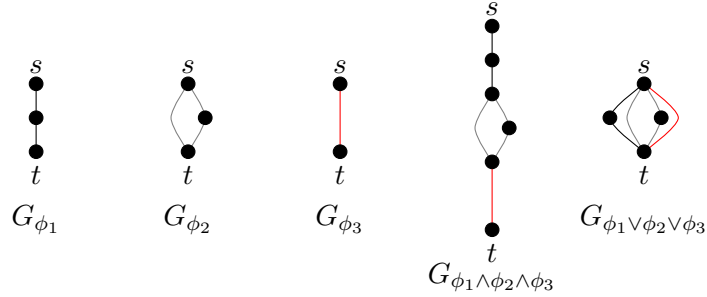

Note that for any $\phi$, $G_\phi$ is planar, and furthermore, both $s$ and $t$ are always on the same face. Thus, we can define $G_\phi'$, $G_\phi(x)$ and $G_\phi'(x)$ as in Section \ref{sec:stconn}. Then we can show the following:
\begin{restatable}{lemma}{formulast}\label{lem:formula-st}
Let $\phi$ be any \textsc{and}-\textsc{or} formula on $N$ variables. For every $x\in\{0,1\}^N$, there exists a path from $s$ to $t$ in $G_\phi(x)$ if and only if $\phi(x)=1$. Furthermore, for every $x\in\{0,1\}^N$, there exists a path from $s'$ to $t'$ in $G'_\phi(x)$ if and only if $\phi(x)=0$. 
\end{restatable}
We give a formal proof of Lemma \ref{lem:formula-st} in Appendix \ref{app:formula}, but the intuition is that an \textsc{or} of subformulas, $\phi_1\vee\dots\vee\phi_\subf$ evaluates to true if any of the subformulas evaluates to true, and likewise, if two vertices are connected by multiple subgraphs in parallel, the vertices are connected if there is a path in any of the subgraphs. An $\textsc{and}$ of subformulas $\phi_1\wedge\dots\wedge \phi_\subf$ evaluates to true only if every subformula evaluates to true, and likewise, if two vertices are connected by multiple subgraphs in series, the vertices are only connected if there is a path through every subgraph.
Thus, we can show by induction that $s$ and $t$ are connected in $G_\phi(x)$ if and only if $\phi(x)=1$. To see that $s'$ and $t'$ are connected in $G_\phi'(x)$ if and only if $\phi(x)=0$, we can use a similar argument, and make use of the fact that an $s't'$-path in $G_\phi'(x)$ is an $st$-cut in $G_\phi(x)$.

Lemma \ref{lem:formula-st} implies that we can solve a formula evaluation problem $\textsc{Eval}_\phi$ by solving the associated $st$-connectivity problem, in which the input is a subgraph of $G_\phi$. By our construction, $G_\phi$ will always be a planar graph with $s$ and $t$ on the external face, so moreover, we can apply Theorem \ref{thm:stconn} to obtain the following.
\begin{theorem}\label{thm:easy-instances}
For any family $\phi$ of \textsc{and}-\textsc{or} formulas, the bounded error quantum query complexity of $\textsc{Eval}_\phi$ when the input is promised to come from a set $D$ is 
\begin{align}
O\left(\min_{\Ohm}\sqrt{\max_{x\in D:\phi(x)=1} R_{s,t}(G_\phi(x),\Ohm) \times \max_{x\in D:\phi(x)=0}R_{s',t'}(G_\phi'(x),\Ohm')}\right),\label{eq:easy-instances}
\end{align}
where the minimization is over all positive real-valued functions $\Ohm$ on $E(G_\phi)$.
\end{theorem}
\begin{proof} 
By Lemma \ref{lem:formula-st}, the query complexity of \textsc{Eval}$_\phi$ on $D$ is at most the query complexity of $st$-\textsc{conn}$_{G_\phi,D}$. Since $G_\phi$ is planar, and has $s$ and $t$ on the same face, we can apply Theorem \ref{thm:stconn}, which immediately implies the result.
\end{proof}

\subsection{Comparison to Existing Boolean Formula Algorithms }\label{sec:nandGraphsApplication}

Reichardt proved that the quantum query complexity of evaluating any formula on $N$ variables is $O(\sqrt{N})$ \cite[Corollary 1.6]{reichardt2010span}. 
Our algorithm recovers this result:

\begin{restatable}{theorem}{rootN}\label{thm:rootN}
Let $\phi$ be a read-once  formula on $N$ variables. Then there exists a choice of $\Ohm$ on $E(G_{\phi})$ such that the quantum algorithm obtained from the span program $P_{G_\phi,\Ohm}$ computes \textsc{Eval}$_\phi$ with bounded error in $O(\sqrt{N})$ queries.
\end{restatable}

\noindent We need the following claim, which we prove in Appendix \ref{app:formula}:
\begin{restatable}{claim}{compose}\label{claim:compose}
If $\phi = \phi_1\vee\phi_2\vee\cdots\vee\phi_l$, then $G'_\phi(x)$ is formed by composing $\{G'_{\phi_i}(x)\}_i$ in series, and
if $\phi = \phi_1\wedge\phi_2\wedge\cdots\wedge\phi_l$, then $G'_\phi(x)$ is formed by composing $\{G'_{\phi_i}(x)\}_i$ in parallel.
\end{restatable}

The intuition behind Claim \ref{claim:compose} is the following. Although $G_\phi'$ is defined via the dual of $G_\phi$, which is constructed through a sequence of series and parallel compositions, $G_\phi'$ itself can also be built up through a sequence of series and parallel compositions. For any and-or formula $\phi$ on $N$ variables, we can define a formula $\phi'$ on $N$ variables by replacing all $\vee$-nodes in $\phi$ with $\wedge$-nodes, and all $\wedge$-nodes in $\phi$ with $\vee$-nodes. By de Morgan's law, for all $x\in\{0,1\}^N$, $\phi(x)=\neg \phi'(\bar{x})$, where $\bar{x}$ is the entrywise negation of $x$. A simple inductive proof shows that $G_{\phi'}=G_{\phi}'$, and for all $x$, $G_{\phi'}(\bar{x})=G_\phi'(x)$ (see Lemma \ref{claim:primeForm} in Appendix \ref{app:formula}). 

\begin{proof}[Proof of Theorem \ref{thm:rootN}]
We will make use of the following fact: for any network $(G,c)$, and any positive real number $W$:
\begin{equation}
R_{s,t}(G,c/W)=\min_\theta\sum_{e\in E(G)}\frac{\theta(e)^2}{c(e)/W}=W\min_\theta\sum_{e\in E(G)}\frac{\theta(e)^2}{c(e)}
=WR_{s,t}(G,c).\label{eq:weight-factor}
\end{equation}

We now proceed with the proof.
For any formula $\phi$ in $\{\wedge,\vee,\neg\}$, by repeated applications of de Morgan's law, we can push all \textsc{not}-gates to distance-1 from a leaf. Since $x_i$ and $\neg x_i$ can both be learned in one query, we can restrict our attention to \textsc{and}-\textsc{or} formulas. 

If $\phi$ has only $N=1$ variable, it's easy to see that $W_+(P_{G_\phi,\Ohm})W_-(P_{G_\phi,\Ohm})\leq N$ for $\Ohm$ taking value $1$ on the single edge in $G_\phi$. We will prove by induction that this is true for any $\phi$, for some choice of $\Ohm$, completing the proof, since the complexity of our algorithm obtained from $P_{G_\phi,c}$ is $O\left(\sqrt{W_+(P_{G_\phi,\Ohm})W_-(P_{G_\phi,\Ohm})}\right)$. 

Suppose $\phi=\phi_1\wedge\dots\wedge\phi_\subf$ for formulas $\phi_i$ on $N_i$
variables, so $\phi$ has $N=\sum_iN_i$ variables. For $x\in \{0,1\}^N$, we will
let $x^i\in \{0,1\}^{N_i}$ denote the $(N_1+\dots+N_{i-1}+1)$-th to
$(N_1+\dots+N_i)$-th bits of $x$. For each $G_{\phi_i}$, by the induction hypothesis, there is some
weight function $\Ohm_i$ on $E(G_{\phi_i})$ such that
$W_+(P_{G_{\phi_i},\Ohm_i})W_-(P_{G_{\phi_i},\Ohm_i})\leq N_i$. 

Using our
construction, $G_\phi$ is formed by composing $\{G_{\phi_i}\}_i$ in series. Thus
every edge $(\{u,v\},\lambda)\in E(G_{\phi})$ corresponds to an edge
$(\{u,v\},\lambda)\in E(G_{\phi_i})$ for some $i$. We create a weight function $\Ohm:E(G_\phi)\rightarrow \mathbb R^+$ such that 
$\Ohm(\{u,v\},\lambda) = \frac{\Ohm_i(\{u,v\},\lambda)}{W_{-}(P_{G_{\phi_i},\Ohm_i})}$ if $(\{u,v\},\lambda)$ is an edge
originating from the graph $G_{\phi_i}$. That is, our new weight function
is the same as combining all of the old weight functions, up to scaling
factors $\{W_{-}(P_{G_{\phi_i},\Ohm_i})\}_i$.

Using Lemma \ref{lemma:both_witnesses}, Claim \ref{claim:parallel_series}, and Eq.\ \eqref{eq:weight-factor}, for any 1-instance $x$,
\begin{align}
w_+(x,P_{G_phi,c}) &= \frac{1}{2}R_{s,t}(G_\phi(x),c)=\frac{1}{2}\sum_{i=1}^\subf R_{s,t}\left(G_{\phi_i}(x),\frac{c_i}{W_-(P_{G_{\phi_i},c_i})}\right)\nonumber\\
 &= \frac{1}{2}\sum_{i=1}^\subf W_-(P_{G_{\phi_i},c_i}) R_{s,t}\left(G_{\phi_i}(x),{c_i}{}\right) 
\leq \sum_{i=1}^\subf W_{-}(\phi_i,P_{G_{\phi_i},\Ohm_i})W_+(P_{G_{\phi_i},\Ohm_i}).
\end{align}
Thus 
\begin{equation}
W_+(P_{G_{\phi},\Ohm})\leq \sum_{i=1}^\subf W_{-}(P_{G_{\phi_i},\Ohm_i})W_+(P_{G_{\phi_i},\Ohm_i})\leq \sum_{i=1}^{\subf} N_i=N.\label{eq:Wplus}
\end{equation}
Recall that for a weight function $c$ on $G_\phi$, we define a weight function $c'$ on $G'_\phi$ by $c'(e^\dagger)=1/c(e)$. Then for an edge $e\in E(G_{\phi_i}(x^i))$, we have $c'(e^\dagger)={W_-(P_{G_{\phi_i},c_i})}/{c_i(e)}=W_-(P_{G_{\phi_i},c_i})c_i'(e^\dagger)$.
By Claim \ref{claim:compose}, $G'_\phi$ is formed by composing $\{G'_{\phi_i}\}_i$ in parallel, so by Lemma \ref{lemma:both_witnesses}, Claim \ref{claim:parallel_series}, and Eq.\ \eqref{eq:weight-factor}:
\begin{align}
w_-(x,P_{G_\phi,\Ohm})&=2R_{s',t'}(G_{\phi}'(x),\Ohm') 
 =  2\left(\sum_{i=1}^{\subf} \frac{1}{R_{s',t'}(G_{\phi_i}'(x^i),c')}\right)^{-1}\nonumber\\
 &= 2\left(\sum_{i=1}^{\subf} \frac{1}{R_{s',t'}(G_{\phi_i}'(x^i),W_-(P_{G_{\phi_i},c_i})c'_i)}\right)^{-1}
= 2\left(\sum_{i=1}^{\subf} \frac{W_-(P_{G_{\phi_i},c_i})}{R_{s',t'}(G_{\phi_i}'(x^i),c'_i)}\right)^{-1}\nonumber\\
 &= 2\left(\sum_{i=1}^{\subf} \frac{W_-(P_{G_{\phi_i},c_i})}{w_-(x^i,P_{G_{\phi_i},c_i})}\right)^{-1}.
\end{align}
Whenever $x$ is a 0-instance of $\phi$, the set $S\subseteq[\subf]$ of $i$ such that $x^i$ is a 0-instance of $\phi_i$ is non-empty. This is exactly the set of $i$ such that $w_-(x^i,P_{G_{\phi_i},c_i})<\infty$. Continuing from above, we have:
\begin{equation}
w_-(x,P_{G_\phi,\Ohm})=\left(\sum_{i\in S} \frac{W_-(P_{G_{\phi_i},c_i})}{w_-(x^i,P_{G_{\phi_i},c_i})}\right)^{-1} 
\leq\left(\sum_{i\in S} \frac{W_-(P_{G_{\phi_i},c_i})}{W_-(P_{G_{\phi_i},c_i})}\right)^{-1}
=\frac{1}{|S|}\leq 1 .
\end{equation}
Thus $W_-(P_{G_{\phi_i},c_i})\leq 1$. Combining this with Eq.\ \eqref{eq:Wplus} we have $W_+(P_{G_{\phi},c})W_-(P_{G_\phi,c})\leq N$, as desired. 

\noindent The proof for the case $\phi=\phi_1\vee\dots\vee\phi_\subf$ is similar. 
\end{proof}

\noindent An immediate corollary of Theorem \ref{thm:rootN} is the following.
\begin{corollary}
Deciding $st$-connectivity on subgraphs of two-terminal series-parallel graphs of $N$ edges
can be accomplished using $O(\sqrt{N})$ queries, if $s$ and $t$ are chosen to be the two terminal nodes.
\end{corollary} 

As with many results in this field, characterizing classical complexity seems to be more difficult than quantum complexity. However, we show we can lower bound the classical query complexity of a class of Boolean formulas in terms of the effective resistance of their corresponding graphs, achieving a quadratic quantum/classical speed-up in query complexity.

We consider \textsc{and}-\textsc{or} formulas on restricted domains. For $N,h\in \mathbb
Z^+$, let $D_{N,h}=\{x\in\{0,1\}^N:|x|=N \textrm{ or } |x|\leq
N-h\}$ and let $D_{N,h}'=\{x\in\{0,1\}^N:|x|=0 \textrm{ or }
|x|\geq h\}$. We will analyze \textsc{and}-\textsc{or} formulas such that the input to every gate in the formula comes from $D_{N,h}$ (in the case of $\textsc{and}$), which we denote $\textsc{and}|_{D_{N,h}}$ and $D_{N,h}'$ (in the case of $\textsc{or}$), which we denote $\textsc{or}|_{D_{N,h}'}$.
 These promises on the domains make it easier to evaluate both functions. For example, if $\textsc{or}$ evaluates to 1, we are promised that there will not be just one input with value $1,$ but at least $h.$

 Then using sabotage complexity \cite{BK16} to bound the classical query complexity, we have the following theorem, whose (somewhat long, but not technical) proof can be found in Appendix \ref{app:classical-lb}:

\begin{theorem}\label{thm:class}
Let $\phi=\phi_1\circ \phi_2\circ\cdots\circ \phi_l$, where for
each $i\in[l]$, $\phi_i=\textsc{or}|_{D_{N_i,h_i}'}$ or $\phi_i=\textsc{and}|_{D_{N_i,h_i}}$. Then the randomized bounded-error query complexity of evaluating $\phi$ is $\Omega\left(\prod_{i=1}^lN_i/h_i\right)$, and the bounded-error quantum query complexity of evaluating $\phi$ is $O\left(\prod_{i=1}^l\sqrt{N_i/h_i}\right)$.
\end{theorem}

Note that in the above theorem, when we compose formulas with promises on the input, we implicitly assume a promise on the input to the composed formula. More precisely, for $\phi_1$ on $D_1\subseteq \{0,1\}^{N_1}$ and $\phi_2$ on $D_2\subseteq \{0,1\}^{N_2}$, $\phi=\phi_1\circ\phi_2$ is defined on all $x=(x^1,\dots,x^{N_1})\in \{0,1\}^{N_1N_2}$ such that $x^i\in D_2$ for all $i\in [N_1]$, and $(\phi_2(x^1),\dots,\phi_2(x^{N_1}))\in D_1$. 

Theorem \ref{thm:class} is proven by showing that 
\begin{align}
\frac{\prod_{i=1}^lN_i}{\prod_{i=1}^lh_i}=\left(\max_{x\in D: \phi(x)=1}R_{s,t}(G_\phi(x))\right) \left(\max_{x\in D: \phi(x)=0}R_{s,t}(G'_\phi(x))\right),
\end{align}
and using sabotage complexity to show that this is a lower bound on the randomized query complexity of $\phi$. This gives us a quadratic separation between the randomized and quantum query complexities of this class of formulas. For details, see Appendix \ref{app:classical-lb}.

Using the composition lower bound for promise Boolean functions of \cite{K11}, and the lower bound for Grover's search with multiple marked items \cite{BBHT98}, we have that the quantum query complexity of Theorem \ref{thm:class} is tight. Additionally, in light of our reduction from Boolean formula evaluation to $st$-connectivity, we see that our example from Figure~\ref{fig:exLine} in Section \ref{sec:stconn} is equivalent to the problem of $\textsc{and}|_{D_{N,h}}$, so our query bound in that example is also tight.

Based on Theorem \ref{thm:class}, one might guess that when evaluating formulas using the $st$-connectivity reduction, one can obtain at most a quadratic speed-up over classical randomized query complexity. However, it is in fact possible to obtain a superpolynomial speed-up for certain promise problems using this approach, as we will discuss in Section \ref{sec:nandQuery}.

\section{NAND-tree Results}\label{sec:nand}

\subsection{Query Separations}\label{sec:nandQuery}

In this section, we prove two query separations that are stronger than our previous results. These query separations rely on the \textsc{nand}-tree formula with a promise on the inputs. This restriction, the $k$-fault promise, will be defined shortly. Let $F^d_k$ be the set of inputs to $\textsc{nand}_d$ that satisfy the $k$-fault condition. Then the two results are the following:
\begin{theorem}\label{thm:superPoly}
Using the $st$-connectivity approach to formula evaluation (Theorem \ref{thm:easy-instances}), one can solve $\textsc{Eval}_{\textsc{nand}_d}$ when the input is promised to be from $F^d_{\log d}$ with $O(d)$ queries, while any classical algorithm requires $\Omega(d^{\log \log (d)})$ queries. 
\end{theorem}
For a different choice of $k$, this example demonstrates the dramatic improvement our $st$-connectivity algorithm can give over the analysis of \cite{BR12} --- in this case, an exponential (or more precisely, a polynomial to constant) improvement:
\begin{theorem}\label{thm:stBig}
Consider the problem $st$-$\textsc{conn}_{G_{\textsc{nand}_d},F^d_1}$. The analysis of \cite{BR12} gives a bound of $O(N^{1/4})$ quantum queries to decide this problem (where $N=2^d$ is the number of edges in $G_{\textsc{nand}_d}$), while our analysis shows this problem can be decided with $O(1)$ quantum queries. 
\end{theorem}

We now define what we mean by $k$-fault \textsc{nand}-trees. In \cite{ZKH12}, Zhan et al.\ find a relationship between
the difficulty of playing the two-player game associated with a
\textsc{nand}-tree, and the witness size of a particular span program for $\textsc{nand}_d$. They find that trees with fewer
 \emph{faults}, or critical decisions for a player
playing the associated two-player game, are easier to evaluate on a
quantum computer. We show that our algorithm does at least as well as the algorithm of Zhan et al.\ for evaluating $k$-fault trees.
To see this, we relate fault complexity to effective resistances of $G_{\textsc{nand}_d}(x)$ or $G_{\textsc{nand}_d}'(x)$.

Consider a \textsc{nand}$_d$ instance $x\in\{0,1\}^{2^{d}}$, and recall the relationship between a \textsc{nand}-tree instance and the two-player \textsc{nand}-tree game described in Section \ref{sec:PreliminariesFormulas}. We call the sequence of nodes that Player $A$ and Player $B$ choose during the course of a game a path $p$ --- this is just a path from the root of the $\textsc{nand}$-tree to a leaf node. If $x$ is $Z$-winnable, we call ${\cal P}_Z(x)$ the set of paths where Player $Z$ wins, and 
Player $Z$ never makes a move that would allow her opponent to win. That is, a  path in ${\cal P}_A(x)$ (resp.\ ${\cal P}_B(x)$) only encounters nodes that are themselves the roots of $A$-winnable (resp.\ $B$-winnable) subtrees and never passes through a node where Player $B$ (resp.\ Player $A$) could make a decision to move to a $B$-winnable (resp.\ $A$-winnable) subtree. Whether a node in the tree is the root of an $A$-winnable or $B$-winnable subtree can be determined by evaluating the subformula corresponding to that subtree. See Figure \ref{fig:nandtree} for an example of ${\cal P}_A$.
Let $\nu_Z(p)$ be the set
of nodes along a path $p$ at which it is Player $Z$'s turn. Thus, $\nu_A(p)$ (resp.\ $\nu_B(p))$ contains those nodes in $p$ at even (resp.\ odd) distance $>0$ from the leaves.

 Zhan et al.\ call a node $v$ a \emph{fault} if one child is the root of an $A$-winnable tree, while the other child is the root of a $B$-winnable tree. Such a node constitutes a critical decision point. If we let
${f}_Z(p)$ denote the number of faults in $\nu_Z(p)$, we can define the fault
complexity ${\cal F}(x)$ of input $x$ as\footnote{We have actually
used the more refined definition of $k$-fault from \cite{K11}.} ${\cal F}(x)=\min\{{\cal F}_A(x),{\cal F}_B(x)\}$, where:
\begin{align}\label{eq:FaultDef}
{\cal F}_Z(x)=\begin{cases}2^{\max_{p\in {\cal P}_Z(x)}
f_Z(p)} &\textrm{ if $x$ is $Z$-winnable}\\
\infty &\textrm{ otherwise}.
\end{cases}
\end{align}
For $k=0,\dots,d/2$, the set of $k$-fault trees, $F^d_k$, are those instances $x\in\{0,1\}^{2^d}$ with $\log{\cal F}(x)\leq k$. In these trees, the winning player will encounter at most $k$ fault nodes on their path to a leaf. Kimmel \cite{K11} shows there exists a span program for evaluating \textsc{nand}-trees whose witness size for an instance $x$ is at most the fault complexity ${\cal F}(x).$

\begin{figure}[ht]
\centering
\begin{tikzpicture}[scale=.65]
\tikzstyle{vertex} = [circle,draw,fill=white,minimum size=.5em]
\tikzstyle{vertexFault} = [circle,draw,fill=white,double,minimum size=.5em]
\tikzstyle{operator} = [rectangle,rounded corners,draw,%fill=cyan,
fill = black!20!white,
minimum size=1.5em]
\tikzstyle{turn} = [rectangle,draw,fill=white,minimum size=1.5em]
\node[vertex] (v0000) at (-7.5,0) {$1$};
\node[operator] (p1) at (-7.5,-1) {$p_1$};
\node[vertex] (v0001) at (-6.5,0) {$1$};
\node[operator] (p2) at (-6.5,-1) {$p_2$};
\node[vertex] (v0010) at (-5.5,0) {$1$};
\node[vertex] (v0011) at (-4.5,0) {$0$};
\node[vertex] (v0100) at (-3.5,0) {$0$};
\node[vertex] (v0101) at (-2.5,0) {$0$};
\node[vertex] (v0110) at (-1.5,0) {$1$};
\node[operator] (p3) at (-1.5,-1) {$p_3$};
\node[vertex] (v0111) at (-0.5,0) {$1$};
\node[operator] (p4) at (-0.5,-1) {$p_4$};
\node[vertex] (v1000) at (0.5,0) {$0$};
\node[vertex] (v1001) at (1.5,0) {$0$};
\node[vertex] (v1010) at (2.5,0) {$0$};
\node[vertex] (v1011) at (3.5,0) {$1$};
\node[vertex] (v1100) at (4.5,0) {$1$};
\node[vertex] (v1101) at (5.5,0) {$1$};
\node[vertex] (v1110) at (6.5,0) {$0$};
\node[vertex] (v1111) at (7.5,0) {$1$};
\node[vertex] (v000) at (-7,1.5) {$\wedge$};
\node[vertexFault] (v001) at (-5,1.5) {$\wedge$};
\node[vertex] (v010) at (-3,1.5) {$\wedge$};
\node[vertex] (v011) at (-1,1.5) {$\wedge$};
\node[vertex] (v100) at (1,1.5) {$\wedge$};
\node[vertexFault] (v101) at (3,1.5) {$\wedge$};
\node[vertex] (v110) at (5,1.5) {$\wedge$};
\node[vertexFault] (v111) at (7,1.5) {$\wedge$};
\node[vertexFault] (v00) at (-6,3) {$\vee$};
\node[vertexFault] (v01) at (-2,3) {$\vee$};
\node[vertex] (v10) at (2,3) {$\vee$};
\node[vertexFault] (v11) at (6,3) {$\vee$};
\node[vertex] (v0) at (-4,4.5) {$\wedge$};
\node[vertexFault] (v1) at (4,4.5) {$\wedge$};
\node[vertexFault] (v) at (0,6) {$\vee$};
\path (v000) edge[thick] (v0001);
\path (v000) edge[thick] (v0000);
\path[dashed] (v00) edge[thick] (v001);
\path (v00) edge[thick] (v000);
\path (v0) edge[thick] (v01);
\path (v0) edge[thick] (v00);
\path[dashed] (v001) edge[thick] (v0011);
\path[dashed] (v001) edge[thick] (v0010);
\path (v01) edge[thick] (v011);
\path[dashed] (v01) edge[thick] (v010);
\path[dashed] (v1) edge[thick] (v11);
\path (v011) edge[thick] (v0111);
\path (v011) edge[thick] (v0110);
\path[dashed] (v111) edge[thick] (v1110);
\path[dashed] (v111) edge[thick] (v1111);
\path[dashed] (v010) edge[thick] (v0101);
\path[dashed] (v010) edge[thick] (v0100);
\path[dashed] (v100) edge[thick] (v1001);
\path[dashed] (v100) edge[thick] (v1000);
\path[dashed] (v101) edge[thick] (v1011);
\path[dashed] (v101) edge[thick] (v1010);
\path[dashed] (v110) edge[thick] (v1101);
\path[dashed] (v110) edge[thick] (v1100);
\path[dashed] (v10) edge[thick] (v101);
\path[dashed] (v10) edge[thick] (v100);
\path[dashed] (v1) edge[thick] (v10);
\path[dashed] (v11) edge[thick] (v111);
\path[dashed] (v11) edge[thick] (v110);
\path[dashed] (v1) edge[thick] (v);
\path (v0) edge[thick] (v);
\node[turn] (p1m) at (10.1,6) {$\textrm{Player $A$'s First Turn}$};
\node[turn] (p1m) at (10.1,4.5) {$\textrm{Player $B$'s First Turn}$};
\node[turn] (p1m) at (10.1,3) {$\textrm{Player $A$'s Second Turn}$};
\node[turn] (p1m) at (10.1,1.5) {$\textrm{Player $B$'s Second Turn}$};
\node (h1) at (-8.3,0) {};
\node (h2) at (-8.3,5.3) {};
\end{tikzpicture}
\caption{Depiction of a depth-4 \textsc{nand}-tree as a two-player game. Let $x$ be the input to $\textsc{nand}_4$ shown in the figure. This instance is $A$-winnable, and ${\cal P}_A(x)$ consists of the paths $\{p_1,p_2,p_3,p_4\}$, shown using solid lines. Fault nodes are those with double circles. Each path in ${\cal P}_A(x)$ encounters two faults at nodes where Player $A$ makes decisions. Therefore, ${\cal F}_A(x)=4$.}
\label{fig:nandtree}
\end{figure}
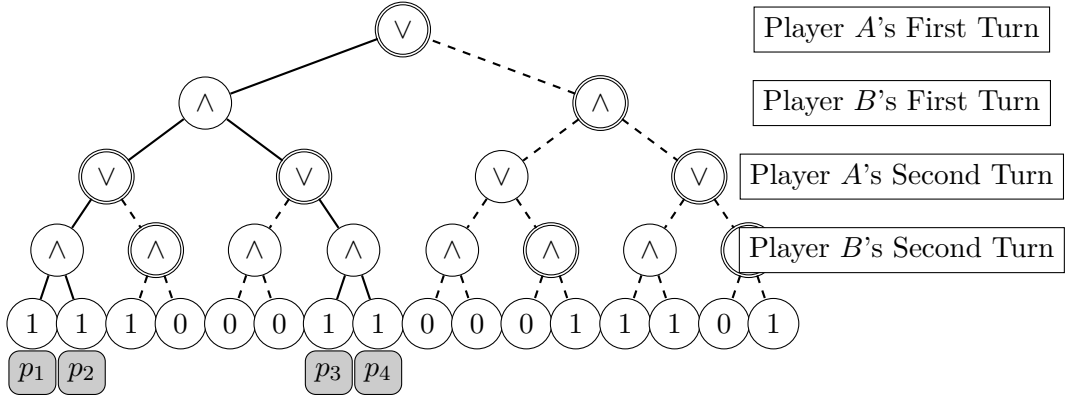

We first show a relationship between effective resistance of $G_{\textsc{nand}_d}(x)$ or $G_{\textsc{nand}_d}'(x)$ and~${\cal F}(x)$:
\begin{restatable}{lemma}{CCeER}\label{thm:kfault_conn}
For any $x\in\{0,1\}^{2^{d}}$, if $d$ is even, then we have $R_{s,t}(G_{\textsc{nand}_d}(x))\leq {\cal F}_A(x)$ and $R_{s',t'}(G'_{\textsc{nand}_d}(x))\leq {\cal F}_B(x)$, while if $d$ is odd, we have $R_{s,t}(G_{\textsc{nand}_d}(x))\leq 2 {\cal F}_A(x)$ and $R_{s',t'}(G'_{\textsc{nand}_d}(x))\leq 2{\cal F}_B(x)$.
\end{restatable}
The proof of Lemma \ref{thm:kfault_conn} can be found in Appendix \ref{app:proofs}.
An immediate corollary of Lemma \ref{thm:kfault_conn} and Theorem \ref{thm:easy-instances} is the following.
\begin{corollary} \label{cor:kfault}
The span program $P_{G_\phi}$ for $\phi=\textsc{nand}_d$ decides $\textsc{Eval}_{\textsc{nand}_d}$ restricted to the domain $X$ in $O(\max_{x\in X}{\cal F}(x))$ queries. In particular, it decides $k$-fault trees (on domain $F^d_k$) in $O(2^k)$ queries.
\end{corollary}

\begin{proof}[Proof of Theorem \ref{thm:superPoly}]
Theorem \ref{thm:superPoly} is now an immediate consequence of Corollary \ref{cor:kfault}, with $k$ set to $\log(d)$, along with the fact from \cite{ZKH12} that the classical query complexity of evaluating such formulas is $\Omega(d^{\log \log (d)})$. 
\end{proof}

\noindent We will use Corollary \ref{cor:kfault}, along with the following claim, to prove Theorem \ref{thm:stBig}:
\begin{claim}\label{claim:Cutparallel_series}
Let two graphs $G_1$ and $G_2$ each have nodes $s$ and $t$ and let $x^1\in\{0,1\}^{E(G_1)}$ and $x^2\in\{0,1\}^{E(G_2)}$. Suppose we create a new graph $G$ by identifying the $s$ nodes and the $t$ nodes (i.e. connecting the graphs in parallel), then
\begin{align}\label{eq:Cparallel}
C_{s,t}(G(x^1,x^2))=C_{s,t}(G_1(x^1))+C_{s,t}(G_2(x^2))
\end{align}
If we create a new graph $G$ by identifying the $t$ node of $G_1$ with the $s$ node of $G_2$ and relabeling this node $v\not\in\{s,t\}$ (i.e. connecting the graphs in series), then
\begin{align}\label{eq:Cseries}
C_{s,t}(G(x^1,x^2))=\min\{C_{s,t}(G_1(x^1)),C_{s,t}(G_2(x^2))\}.
\end{align}
\end{claim}

\begin{proof}[Proof of Theorem \ref{thm:stBig}]
Using Corollary \ref{cor:kfault}, our analysis shows that $st$-$\textsc{conn}_{G_{\textsc{nand}_d},F^d_1}$ can be decided in $O(1)$ queries. 

We apply Eq.\ \eqref{eq:theirbound} to compare to the analysis of \cite{BR12}. We must characterize the quantity $\max\{C_{s,t}(G_{\textsc{nand}_d}(x)):x\in F_1^d, \textrm{and } s,t \textrm{ are not connected}\}$,
since we already have 
\begin{align}
\max_{x\in F_1^d: s,t \textrm{ are connected}}R_{s,t}(G(x)) =O(1).
\end{align}

We now prove that for \emph{every} $x\in\{0,1\}^{2^d}$  such that $\textsc{nand}_d(x)=0$, $C_{s,t}(G_{\textsc{nand}_d}(x)) = 2^{\lfloor d/2\rfloor}$. Thus, for any promise $D$ on the input, as long as there exists some $x\in D$ such that $\textsc{nand}_d(x)=0$, we have $\max_{x\in D:\textsc{nand}_d(x)=0}C_{s,t}(G_{\textsc{nand}_d}(x))=2^{\lfloor{d/2}\rfloor}$. Intuitively, this is because every $st$-cut on any subgraph of $G_{\textsc{nand}_d}$ cuts across $2^{\lfloor d/2\rfloor}$ edges of $G_{\textsc{nand}_d}$. 

The proof is by induction on $d$.
For the base even case, $d=0$ and $G_{\textsc{nand}_0}$ is a single edge connecting $s$ and $t$. The only input $x\in\{0,1\}^{2^0}$ such that $\textsc{nand}_d(x)=0$ is $x=0$, in which case, the $st$-cut is $\kappa(s)=1$ and $\kappa(t)=0$, so the cut is across the unique edge in $G_{\textsc{nand}_0}$,
so $C_{s,t}(G_{\textsc{nand}_0}(x))=1$.

For the induction step, we treat even and odd separately. Suppose $d>0$
is odd. Then $\textsc{nand}_d(x) = \textsc{nand}_{d-1}(x^0)\wedge\textsc{nand}_{d-1}(x^1)$, where $x^0=(x_1,\dots,x_{2^{d-1}})$ and $x^1=(x_{2^{d-1}+1},\dots,x_{2^d})$. Thus $G_{\textsc{nand}_d}(x)$ involves composing two graphs $G_{\textsc{nand}_{d-1}}(x^0)$ and $G_{\textsc{nand}_{d-1}}(x^1)$ in series. 
Since we are assuming $s$ and $t$ are not connected in $G_{\textsc{nand}_d}(x)$, at least one of $G_{\textsc{nand}_{d-1}}(x^0)$ and $G_{\textsc{nand}_{d-1}}(x^1)$ must not be connected. Without loss of generality, suppose $G_{\textsc{nand}_{d-1}}(x^0)$ is not connected and $C_{s,t}(G_{\textsc{nand}_{d-1}}(x^0))\leq C_{s,t}(G_{\textsc{nand}_{d-1}}(x^1))$. By induction, $C_{s,t}(G_{\textsc{nand}_{d-1}}(x^0))=2^{(d-1)/2}$. Thus using Eq.\ \eqref{eq:Cseries} in Claim \ref{claim:Cutparallel_series},
\begin{align}
C_{s,t}(G_{\textsc{nand}_d}(x)) = 2^{(d-1)/2} = 2^{\lfloor d/2\rfloor}.
\end{align}

Now suppose $d>0$
is even. Then $\textsc{nand}_d(x) = \textsc{nand}_{d-1}(x^0)\vee \textsc{nand}_{d-1}(x^1)$. Thus $G_{\textsc{nand}_d}(x)$ involves composing two graphs $G_{\textsc{nand}_{d-1}}(x^0)$ and $G_{\textsc{nand}_{d-1}}(x^1)$ in parallel. 
Since we are assuming $s$ and $t$ are not connected in $G_{\textsc{nand}_d}(x)$, both of $G_{\textsc{nand}_{d-1}}(x^0)$ and  $G_{\textsc{nand}_{d-1}}(x^1)$ must not be connected, and so by induction, we have $C_{s,t}(G_{\textsc{nand}_{d-1}}(x^0))=C_{s,t}(G_{\textsc{nand}_{d-1}}(x^1))=2^{\lfloor(d-1)/2\rfloor}=2^{d/2-1}$, since $d$ is even. Thus using Eq.\ \eqref{eq:Cparallel} in Claim \ref{claim:Cutparallel_series},
\begin{align}
C_{s,t}(G_{\textsc{nand}_d}(x)) = 2^{d/2-1}+2^{d/2-1} = 2^{d/2}=2^{\lfloor d/2\rfloor}.
\end{align}

Therefore, using Eq.\ \eqref{eq:theirbound}, we have that the analysis of \cite{BR12} for $d$-depth \textsc{nand}-trees with inputs in $F_1^d$ gives a query complexity of $O(\sqrt{2^{\lfloor d/2\rfloor}})=O(N^{1/4})$, where $N=2^d$ is the number of input variables. Comparing with our analysis, which gives a query complexity of $O(1),$ we see there is a polynomial to constant improvement. 
\end{proof}

\subsection{Winning the NAND-tree Game}\label{sec:nandWin}

In this section, we describe a quantum algorithm that can be used to help a player make decisions while playing the \textsc{nand}-tree game. 
In particular, we consider the number of queries to $x$ needed by Player $A$ to make
decisions throughout the course of the game in order to win with probability $\geq 2/3$. (In this section, we focus on $A$-winnable trees,
but the case of $B$-winnable trees is similar.)

We first describe a naive strategy, which uses a quantum algorithm \cite{Rei09,R01} 
that decides if a depth-$\depth$ tree is winnable with bounded error in
$O(2^{d/2}\log d)$ queries. If Player $A$
must decide to move to node $v_0$ or $v_1$,
she
evaluates each subtree rooted at $v_0$ and $v_1$, amplifying the success probability to $\Omega(1/d)$ by using $O(\log d)$ repetitions, 
and moves to one that evaluates to 1. 
Since Player $A$ has $O(d)$ decisions to make, this strategy succeeds with bounded error, and since evaluating a \textsc{nand}-tree of depth $r$ costs $O(2^{r/2})$ quantum queries, the total query complexity is:
\begin{equation}
O\left(\textstyle 2\sum_{i=0}^{\frac{d}{2}}2^{\frac{d-2i}{2}}\log d\right) = O\left(2\sum_{i=0}^{\frac{d}{2}}2^{i}\log d\right) =O\left(2^{\frac{d}{2}}\log d\right)=O\left(\sqrt{N}\log\log N\right).
\end{equation}

This strategy does not use the fact that some subtrees may be
easier to win than others. For example, if one choice leads to a subtree with all leaves labeled by 1, whereas the other subtree has all leaves labeled by 0, the player
just needs to distinguish these two disparate cases. More generally, one of the subtrees might have a
 small positive witness size --- i.e., it is very winnable --- whereas
the other has a large positive witness size --- i.e., is not very winnable. 

Our strategy will be to move to the subtree whose formula corresponds to a graph with smaller effective resistance, unless the two subtrees are very close in effective resistance, in which case it doesn't matter which one we
choose. For a depth $d$ game on instance $x$, we show if $R_{s,t}(G_{\textsc{nand}_d}(x))$ is small and Player $B$ plays randomly, this strategy does better than the naive strategy, on average. 

We estimate the effective resistance of both
subtrees of the current node using the witness size estimation algorithm of \cite{IJ15}. In particular, in Appendix \ref{app:approx-analysis} we prove:
\begin{restatable}[$\mathtt{Est}$ Algorithm]{lemma}{appWit}\label{lem:est}
Let $\phi$ be an \textsc{and}-\textsc{or} formula
with constant fan-in $\subf$, $\vee$-depth $d_\vee$ and $\wedge$-depth $d_\wedge$. Then the quantum query complexity of estimating
$R_{s,t}(G_\phi(x))$ (resp. $R_{s,t}(G_{\phi}'(x))$) to relative accuracy
$\epsilon$ is 
$
\tO\left(\frac{1}{\eps^{3/2}}\sqrt{R_{s,t}(G_{\phi}(x))\subf^{d_\vee}}\right)
$ (resp.
$
\tO\left(\frac{1}{\eps^{3/2}}\sqrt{R_{s,t}(G'_{\phi}(x))\subf^{d_\wedge}}\right)
$).
\end{restatable}

Let $\mathtt{Est}(x)$ be the
algorithm from Lemma \ref{lem:est} with $\eps=\frac{1}{3}$, and $\phi=\textsc{nand}_d$, so $\subf=2$, and both $d_\vee$ and $d_\wedge$ are at most $\lceil d/2\rceil$. While estimating the effective resistance of two subtrees, we
only care about which of the subtrees has the smaller effective resistance, so we do not want to wait for both iterations of \texttt{Est}
to terminate.  Let $p(d)$ be some polynomial function in $d$ such
that \texttt{Est}$(x)$ always terminates after at most $p(d)\sqrt{R_{s,t}(G_{\phi}(x))}2^{d/4}$ queries, for all $x\in\{0,1\}^{2^d}$.  We define a subroutine,
$\texttt{Select}(x^0,x^1)$, that takes two instances, $x^0,x^1\in \{0,1\}^{2^{d-1}}$,
and outputs a bit $b$ such that $R_{s,t}(G_{\textsc{nand}_{d-1}}(x^b))\leq 2R_{s,t}(G_{\textsc{nand}_{d-1}}(x^{\bar{b}}))$, where $\bar{b}=b\oplus 1$. $\texttt{Select}$ works
as follows. It runs \texttt{Est}$(x^0)$ and $\texttt{Est}(x^1)$ in parallel. 
If one of these programs, say $\texttt{Est}(x^b)$,
outputs some estimate $w_b$, then it terminates the other program
after $p(d)\sqrt{w_b}2^{d/4}$ steps. If only the algorithm running on $x^b$ 
has terminated after
this time, it outputs $b$. If both programs have terminated, it outputs a bit
$b$ such that $w_b\leq w_{\bar{b}}$.
  In Appendix \ref{app:nandWin}, we prove the following lemma.

\begin{restatable}{lemma}{select}\label{lemma:select}
Let $x^0,x^1\in\{0,1\}^{2^d}$ be instances of \textsc{nand}$_d$ with at least one of them a 1-instance. Let $N=2^d$, and $w_{\min}=\min\{R_{s,t}(G_{\textsc{nand}_d}(x^0)),R_{s,t}(G_{\textsc{nand}_d}(x^1))\}$. Then
$\mathtt{Select}(x^0,x^1)$  terminates after
$\widetilde{O}\left(N^{1/4}\sqrt{w_{\min}}\right)$
 queries to
$(x^0,x^1)$ and outputs $b$ such that $R_{s,t}(G_{\textsc{nand}_d}(x^b))\leq
{2}R_{s,t}(G_{\textsc{nand}_d}(x^{\bar{b}}))$ with bounded error.
\end{restatable}

\noindent Using Lemma \ref{lemma:select}, we can prove the following (the inductive proof is in
Appendix \ref{app:nandWin}):
\begin{restatable}{theorem}{win}\label{thm:winning}
Let $x\in\{0,1\}^N$ for $N=2^{d}$
be an $A$-winnable input to $\textsc{nand}_d$. At every node $v$ 
where Player $A$ makes a decision, let Player $A$
use the $\mathtt{Select}$ algorithm in the following way.
Let $v_0$ and $v_1$ be the two children of $v$, with
inputs to the respective subtrees of $v_0$ and $v_1$ given
by $x^0$ and $x^1$ respectively. Then Player $A$ moves
to $v_b$ where $b$ is the outcome that occurs a majority of
times when $\mathtt{Select}(x^0,x^1)$
is run $O(\log d)$ times. 
Then if Player $B$, at his decision nodes, chooses left
and right with equal probability, Player $A$ will
win the game with probability at least $2/3$, and 
will use
$\widetilde{O}\left(N^{1/4}\sqrt{R_{s,t}(G_{\textsc{nand}_d}(x))}\right)$ queries on average, where the average
is taken over the randomness of Player $B$'s choices. 
\end{restatable}

\section{Acknowledgments}
Both authors thank Arjan Cornelissen and Alvaro Piedrafita for making us aware of a bug in Lemma~\ref{lem:time}.
SK is funded by the US Department of Defense. SJ acknowledges funding provided by the Institute for Quantum Information and Matter, an NSF Physics Frontiers Center (NFS Grant PHY-1125565) with support of the Gordon and Betty Moore Foundation (GBMF-12500028); and an NWO WISE Fellowship.

\bibliographystyle{plainnat}

\bibliography{nandbib}

\appendix

\section{Analysis of the Span Program for {\it st}-Connectivity}\label{app:st-proofs}

In this section, we analyze the complexity of our span-program-based algorithms, proving  Lemma \ref{lemma:both_witnesses}, first stated in Section \ref{sec:stconn}, which relates witness sizes of the span program $P_{G,\Ohm}$ to the effective resistance of graphs related to $G$.

We need the concept of a circulation, which is like a flow but with no source and no sink.
\begin{definition}[Circulation] 
A \emph{circulation} on a graph $G$ is a function $\theta:\overrightarrow{E}(G)\rightarrow\mathbb{R}$ such that:
\begin{enumerate}
\item For all $(u,v,\edgeL )\in \overrightarrow{E}(G)$, $\theta(u,v,\edgeL )=-\theta(v,u,\edgeL )$;\label{item1}
\item for all $u\in V(G)$, $\sum_{v,\edgeL : (u,v\edgeL )\in\overrightarrow{E}(G)}\theta(u,v,\edgeL )=0$. \label{item3}
\end{enumerate}
\end{definition}

The following easily verified observations will be useful in several of the remaining proofs in this section.

\begin{claim}\label{claim:flow-decomposition}
Let $\theta$ be a unit $st$-flow in some multigraph $G$. We can consider the corresponding vector $\ket{\theta}=\sum_{(u,v,\edgeL )\in\overrightarrow{E}(G)}\theta(u,v,\edgeL )\ket{u,v,\edgeL }$. Then $\ket{\theta}$ can be written as a linear combination of vectors corresponding to self-avoiding $st$-paths and cycles that are edge-disjoint from these paths.

Let $\sigma$ be a circulation on $G$. Then $\ket{\sigma}$ can be written as a linear combination of cycles in $G$. Furthermore, $\ket{\sigma}$ can be written as a linear combination of cycles such that each cycle goes around a face of $G$. 
\end{claim}

\noindent The next claim shows a direct correspondence between positive witnesses, and $st$-flows. 

\begin{claim}\label{claim:pos-wit}
Fix a span program $P_{G,\Ohm}$ as in \eqref{eq:P}. Call $\ket{w}\in H$ a positive witness in $P_{G,\Ohm}$ if $A\ket{w}=\tau$ (note that such a $\ket{w}$ is not necessarily a positive witness for any particular input $x$). Then if $\theta$ is a unit $st$-flow in $G$, $\frac{1}{2}\sum_{(u,v,\edgeL )\in\overrightarrow{E}(G)}\frac{\theta(u,v,\edgeL )}{\sqrt{\Ohm(\{u,v\},\lambda)}}\ket{u,v,\edgeL }$ is a positive witness in $P_{G,\Ohm}$, and furthermore, if $\ket{w}$ is a positive witness in $P_{G,\Ohm}$, then $\theta(u,v,\edgeL )=\sqrt{\Ohm(\{u,v\},\lambda)}(\braket{w}{u,v,\edgeL }-\braket{w}{v,u,\edgeL})$ is a unit $st$-flow in $G$.
\end{claim}
\begin{proof}
The proof is a straightforward calculation. Let $\theta$ be a unit $st$-flow on $G$. Then
\begin{eqnarray}
&& A\left(\frac{1}{2}\sum_{(u,v,\edgeL )\in\overrightarrow{E}(G)}\frac{\theta(u,v,\edgeL )}{\sqrt{\Ohm(\{u,v\},\lambda)}}\ket{u,v,\edgeL}\right) \nonumber\\
 &=& \frac{1}{2}\sum_{(u,v,\edgeL )\in\overrightarrow{E}(G)}\theta(u,v,\edgeL )(\ket{u}-\ket{v})\nonumber\\
&=& \frac{1}{2}\sum_{u\in V(G)}\left(\sum_{v,\edgeL:(u,v,\edgeL )\in \overrightarrow{E}(G)}\theta(u,v,\edgeL )\right)\ket{u}+\frac{1}{2}\sum_{v\in V(G)}\left(\sum_{u,\edgeL:(v,u,\edgeL)\in\overrightarrow{E}(G)}\theta(v,u,\edgeL)\right)\ket{v}\nonumber\\
&=& \frac{1}{2}\left(\ket{s}-\ket{t}\right)+\frac{1}{2}\left(\ket{s}-\ket{t}\right)=\tau.
\end{eqnarray}
Above we have used that $\theta(u,v,\edgeL )=-\theta(v,u,\edgeL)$, and $\sum_{v,\edgeL:(u,v,\edgeL )\in \overrightarrow{E}(G)}\theta(u,v,\edgeL )=0$ when $u\not\in\{s,t\}$, 1 when $u=s$, and $-1$ when $u=t$. 

To prove the second half of the claim, let $\ket{w}$ be such that $A\ket{w}=\tau$, and define $\theta(u,v,\edgeL )=\sqrt{\Ohm(\{u,v\},\lambda)}(\braket{w}{u,v,\edgeL }-\braket{w}{v,u,\edgeL})$. We immediately see that $\theta(u,v,\edgeL )=-\theta(v,u,\edgeL)$ for all $(u,v,\edgeL )$. Furthermore, we have:
\begin{align}
\ket{s}-\ket{t} =A\ket{w}
&=\sum_{u\in V(G)}\left(\sum_{(u,v,\edgeL )\in\overrightarrow{E}(G)}\sqrt{\Ohm(\{u,v\},\lambda)}\braket{u,v,\edgeL }{w}\right)\ket{u}\nonumber\\
&\qquad\qquad-\sum_{v\in V(G)}\left(\sum_{(u,v,\edgeL )\in\overrightarrow{E}(G)}\sqrt{\Ohm(\{u,v\},\lambda)}\braket{u,v,\edgeL }{w}\right)\ket{v} \nonumber\\
&=
\sum_{u\in V(G)}\left(\sum_{(u,v,\edgeL )\in\overrightarrow{E}(G)}\sqrt{\Ohm(\{u,v\},\lambda)}(\braket{u,v,\edgeL }{w}-\braket{v,u,\edgeL}{w})\right)\ket{u}\nonumber\\&=
\sum_{u\in V(G)}\left(\sum_{(u,v,\edgeL )\in\overrightarrow{E}(G)}\theta(u,v,\edgeL )\right)\ket{u}.
\end{align}
Thus, for all $u\in V(G(x))\setminus\{s,t\}$, $\sum_{(u,v,\edgeL )\in\overrightarrow{E}(G)}\theta(u,v,\edgeL )=0$, and $\sum_{(s,v,\edgeL )\in\overrightarrow{E}(G)}\theta(s,v,\edgeL)=\sum_{(v,t,\edgeL)\in\overrightarrow{E}(G)}\theta(v,t,\edgeL)=1$. Thus, $\theta$ is a unit $st$-flow on $G$.
\end{proof}

\noindent The next claim shows a direct correspondence between negative witnesses, and $s't'$-flows. 
\begin{claim}\label{claim:neg-wit}
For a planar graph $G$, fix a span program $P_{G,\Ohm}$ as in \eqref{eq:P}. Call a linear function $\omega:V(G)\rightarrow\mathbb{R}$ a negative witness if $\omega\tau=1$. Then $\theta((u,v,\edgeL )^\dagger)=\omega(u)-\omega(v)$ is a unit $s't'$-flow on $G'$, and furthermore, for every $s't'$-flow $\theta$ on $G'$ there is a negative witness $\omega$ such that $\theta((u,v,\edgeL )^\dagger)=\omega(u)-\omega(v)$ for all $(u,v,\edgeL )\in\overrightarrow{E}(G)$. 
\end{claim}
\begin{proof}
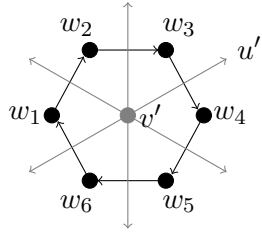
\begin{figure}[ht]
\centering
\begin{tikzpicture}[scale=1]
\filldraw[gray] (0,0) circle (.1);

\draw[->] (-1,0)--(-.58,.78);
\draw[->] (-.5,.866)--(.4,.866);
\draw[->] (.5,.866)--(.92,.08);
\draw[->] (1,0)--(.58,-.78);
\draw[->] (.5,-.866)--(-.4,-.866);
\draw[->] (-.5,-.866)--(-.92,-.08);

\filldraw (-1,0) circle (.1);
\filldraw (-.5,.866) circle (.1);
\filldraw (.5,.866) circle (.1);
\filldraw (1,0) circle (.1);
\filldraw (.5,-.866) circle (.1);
\filldraw (-.5,-.866) circle (.1);

\draw[gray,->] (0,0) -- (-1.299,.75);
\draw[gray,->] (0,0) -- (0,1.5);
\draw[gray,->] (0,0) -- (1.299,.75);
\draw[gray,->] (0,0) -- (1.299,-.75);
\draw[gray,->] (0,0) -- (0,-1.5);
\draw[gray,->] (0,0) -- (-1.299,-.75);

\node at (.3,0) {$v'$};
\node at (1.6,.9) {$u'$};

\node at (-1.35,0) {$w_1$};
\node at (-.675,1.1691) {$w_2$};
\node at (.675,1.1691) {$w_3$};
\node at (1.35,0) {$w_4$};
\node at (.675,-1.1691) {$w_5$};
\node at (-.675,-1.1691) {$w_6$};

\end{tikzpicture}
\caption{The duality between a cycle and a star.}\label{fig:face}
\end{figure}
When we consider the edges of $G$ as directed edges, we assign edge directions to the dual by orienting each dual edge $\pi/2$ radians counter-clockwise from the primal edge. 

Note that without loss of generality, if $\omega$ is a negative witness, we can assume $\omega(s)=1$ and $\omega(t)=0$. This is because $\norm{\omega A}$ and $\norm{\omega A\Pi_{H(x)}}$ are invariant under affine transformations of~$\omega$. 

We first show that if $\omega$ is a negative witness in $P_{G,\Ohm}$, then
$\theta:\overrightarrow{E}(G')\rightarrow\mathbb{R}$ defined
$\theta((u,v,\edgeL )^\dagger)=\omega(u)-\omega(v)$ is a unit $s't'$-flow on
$G'$. To begin with, we will define $\theta'((u,v,\edgeL
)^\dagger)=\omega(u)-\omega(v)$ on $\overrightarrow{E}(\overline{G}^\dagger)$,
so $\theta'$ agrees with $\theta$ everywhere $\theta$ is defined, and in
addition, $\theta'(s',t',\emptyset)=\theta'((s,t,\emptyset)^\dagger)=\omega(s)-\omega(t)=1$, and $\theta'(t',s',\emptyset)=-1$. Then clearly we have
$\theta'(u',v',\edgeL)=-\theta'(v',u',\edgeL)$ for all $(\{u',v'\},\edgeL)\in
E(\overline{G}^\dagger)$.

Next, every $v'\in V(\overline{G}^\dagger)$ corresponds to a face $f_{v'}$ of $\overline{G}$, and the edges coming out of $v'$ are dual to edges going clockwise around the face $f_{v'}$ (see Figure \ref{fig:face}). If $(w_1,w_2,\edgeL_1),\dots,(w_k,w_{k+1},\edgeL_k)$, for $w_{k+1}=w_1$, are the directed edges going clockwise around $f_{v'}$, then we have:
\begin{align}
0=\sum_{i=1}^k(\omega(w_i)-\omega(w_{i+1}))=\sum_{i=1}^{k}\theta'((w_i,w_{i+1},\edgeL_i)^
\dagger)=\!\!\!\!\sum_{\substack{u',\edgeL:\\ (\{v',u'\},\edgeL)\in E(\overline{G}')}}\theta'(v',u',\edgeL).
\end{align}
Thus, $\theta'$ is a circulation. Then, since $\theta'(s',t',\emptyset)=1$, if we remove the flow on this edge, which recovers $\theta$, we get a unit $s't'$-flow on $G'$. 

Next we will show that if $\theta$ is a unit $s't'$-flow on $G'$, then there exists a negative witness $\omega$ in $P_{G,\Ohm}$ such that for all $(u,v,\edgeL )\in \overrightarrow{E}(G)$, $\theta((u,v,\edgeL )^\dagger)=\omega(u)-\omega(v)$.

Define $\theta'$ to be the circulation on $\overline{G}^\dagger$ obtained from defining $\theta'(u',v',\edgeL)=\theta(u',v',\edgeL)$ for all $(u',v',\edgeL)\in \overrightarrow{E}(G')$, and $\theta'(s',t',\emptyset)=-\theta'(t',s',\emptyset)=1$. Then if we define $\ket{\theta'}=\sum_{(u,v,\edgeL )\in\overrightarrow{E}(\overline{G}^\dagger)}\theta'(u,v,\edgeL )\ket{u,v,\edgeL }$, we can express $\ket{\theta'}$ as a linear combination of cycles 
around the faces of $\overline{G}^\dagger$, $\ket{\theta'}=\sum_{f\in F(\overline{G}^\dagger)}\alpha_f\ket{\overrightarrow{C}_f}+\sum_{f\in F(\overline{G}^\dagger)}\alpha_f'\ket{\overleftarrow{C}_f}$, where if $w_{k+1}=w_1$ and $(w_1,w_2,\edgeL_1),\dots,(w_k,w_{k+1},\edgeL_k)$ is a clockwise cycle around $f$,
$\ket{\overrightarrow{C}_f}=\sum_{i=1}^{k}\ket{w_i,w_{i+1},\edgeL_i}$ is the clockwise cycle around the face $f$, and $\ket{\overleftarrow{C}_f}=\sum_{i=1}^k\ket{w_{i+1},w_i,\edgeL_i}$ is the counter-clockwise cycle around $f$. There is a one-to-one correspondance between vertices in $V(\overline{G})=V(G)$ and faces in $F(\overline{G}^\dagger)$, so we can define $\omega:V({G})\rightarrow\mathbb{R}$ by $\omega(v_f)=\frac{1}{2}(\alpha_f-\alpha_f')$. 

We claim that for all $(u,v,\edgeL )\in \overrightarrow{E}(\overline{G})$, $\omega(u)-\omega(v)=\theta'((u,v,\edgeL )^\dagger)$. 
Let $(u',v',\edgeL)$ be any edge in $\overrightarrow{E}(\overline{G}^\dagger)$. This edge is part of a clockwise cycle around one face in $\overline{G}^\dagger$, call it $f$, and a counter clockwise cycle around one face in $\overline{G}^\dagger$, call it $g$. Since these are the only two faces containing the edge $(u',v',\edgeL)$, we must have $\theta'(u',v',\edgeL)=\braket{u',v',\edgeL}{\theta'}=\alpha_f+\alpha_g'$. Since $\theta'(u',v',\edgeL)=-\theta'(v',u',\edgeL)$, we have $\alpha_f+\alpha_g'=-\alpha_f'-\alpha_g$. Thus:
\begin{align}
\omega(v_f)-\omega(v_g)=\frac{1}{2}\left(\alpha_f-\alpha_f'-\alpha_g+\alpha_g'\right)=\frac{1}{2}\left(\theta'(u',v',\edgeL)-\theta'(v',u',\edgeL)\right)=\theta'((v_f,v_g,\depth)^\dagger).
\end{align}
In particular, this means that $\omega(s)-\omega(t)=\theta'((s,t,\emptyset)^\dagger)=\theta'(s',t',\emptyset)=1$, so $\omega$ is a negative witness, and for all $(u,v,\edgeL )\in \overrightarrow{E}(G)$, $\omega(u)-\omega(v)=\theta((u,v,\edgeL)^\dagger )$. 
\end{proof}

Now we can prove the main result of this section, Lemma \ref{lemma:both_witnesses}:
\poswit*
\begin{proof}
If $x$ is a 1-instance, $s$ and $t$ are connected in $G(x)$, so there exists a unit $st$-flow on $G(x)$, which is a unit $st$-flow on $G$ that is supported only on $\overrightarrow{E}(G(x))$. Let $\theta$ be the flow on $G(x)$ such that $R_{s,t}(G(x),c)=\sum_{(\{u,v\},\edgeL)\in E(G(x))}\frac{\theta(u,v,\edgeL )^2}{\Ohm(\{u,v\},\edgeL)}$. By Claim \ref{claim:pos-wit}, $\ket{w}=\frac{1}{2}\sum_{(u,v,\edgeL )}\frac{\theta(u,v,\edgeL )}{\sqrt{\Ohm(\{u,v\},\lambda)}}\ket{u,v,\edgeL }$ is a positive witness in $P_{G,\Ohm}$, and since $\theta$ is supported on $\overrightarrow{E}(G(x))$, $\ket{w}\in H(x)$, and so $\ket{w}$ is a positive witness for $x$ in $P_{G,\Ohm}$. Thus
\begin{align}
w_+(x,P_{G,\Ohm})\leq \norm{\ket{w}}^2=\frac{1}{4}\sum_{(u,v,\edgeL )\in\overrightarrow{E}(G(x))}\Ohm(\{u,v\},\lambda)\theta(u,v,\edgeL )^2=\frac{1}{2}R_{s,t}(G(x), \Ohm).
\end{align}

On the other hand, let $\ket{w}$ be an optimal positive witness for $x$. By Claim \ref{claim:pos-wit}, $\theta(u,v,\edgeL )=\sqrt{\Ohm(\{u,v\},\lambda)}(\braket{u,v,\edgeL }{w}-\braket{v,u,\edgeL}{w})$ is a unit $st$-flow on $G$, and since $\ket{w}\in H(x)$, $\theta(u,v,\edgeL )$ is only non-zero on $\overrightarrow{E}(G(x))$, so $\theta$ is a unit $st$-flow on $G(x)$. Thus,
\begin{eqnarray}
R_{s,t}(G(x), \Ohm) &\leq& \sum_{(\{u,v\},\edgeL)\in E(G(x))}\frac{\theta(u,v,\edgeL )^2}{\Ohm(\{u,v\},\lambda)}=\frac{1}{2}\sum_{(u,v,\edgeL )\in\overrightarrow{E}(G(x))}\left(\braket{u,v,\edgeL }{w}-\braket{v,u,\edgeL}{w}\right)^2\nonumber\\
&=& \!\!\sum_{(u,v,\edgeL )\in\overrightarrow{E}(G(x))}\!\!\braket{u,v,\edgeL }{w}^2 - 
\!\!\!\!\sum_{(u,v,\edgeL )\in\overrightarrow{E}(G(x))}\!\!\!\!\braket{u,v,\edgeL }{w}\braket{v,u,\edgeL}{w}\leq 2\norm{\ket{w}}^2\label{eq:sumTheta}
\end{eqnarray}
where the last inequality is by Cauchy-Schwarz. Thus, $w_+(x,P_{G,\Ohm})=\frac{1}{2}R_{s,t}(G(x))$.

Now we prove that $w_-(x,P_{G,c})=2R_{s',t'}(G'(x),c')$.
Let $x\in\{0,1\}^{E(G)}$ be such that $s$ and $t$ are not connected in $G(x)$.
Fix an optimal negative witness $\omega$ for $x$. By Claim \ref{claim:neg-wit}
the linear function $\theta:\overrightarrow{E}(G')\rightarrow\mathbb{R}$
defined by $\theta((u,v,\edgeL )^\dagger)=\omega(u)-\omega(v)$ is a unit $s't'$
-flow on $G'$. Since $\omega$ is a negative witness for $x$, we also have:

\begin{align}
0=\norm{\omega A\Pi_{H(x)}}^2&=\sum_{{(u,v,\edgeL ) \in \overrightarrow{E}(G(x))}}\Ohm(\{u,v\},\lambda)(\omega(u)-\omega(v))^2\nonumber\\
&=\sum_{{(u,v,\edgeL ) \in \overrightarrow{E}(G(x))}}\Ohm(\{u,v\},\lambda)\theta((u,v,\edgeL )^\dagger)^2\nonumber\\
&=\sum_{{(u',v',\edgeL) \in \overrightarrow{E}(G')\setminus\overrightarrow{E}(G'(x))}}\frac{\theta(u',v',\edgeL)^2}{\Ohm'(\{u',v'\},\lambda)},
\end{align}
since $(u,v,\edgeL )\in \overrightarrow{E}(G(x))$ exactly when $(u,v,\edgeL )^\dagger\not\in \overrightarrow{E}(G'(x))$.
So $\theta$ is only supported on $\overrightarrow{E}(G'(x))$, and so it is a unit $s't'$-flow on $G'(x)$. Thus
\begin{align}
w_-(x,P_{G,\Ohm}) = \norm{\omega A}^2&=\sum_{{(u,v,\edgeL ) \in \overrightarrow{E}(G)}}\Ohm(\{u,v\},\lambda)(\omega(u)-\omega(v))^2\nonumber\\
&=\sum_{{(u',v',\edgeL) \in \overrightarrow{E}(G'(x))}}\frac{\theta(u',v',\edgeL)^2}{\Ohm'(\{u',v'\},\lambda)}
\geq 2R_{s',t'}(G'(x), \Ohm').
\end{align}

For the other direction, let $\theta$ be an $s't'$-flow in $G'(x)$ with minimal energy. By Claim \ref{claim:neg-wit}, there is a negative witness $\omega$ such that $\theta((u,v,\edgeL )^\dagger)=\omega(u)-\omega(v)$. Since $\theta$ is supported on edges $(u',v',\edgeL)\in \overrightarrow{E}(G'(x))$, which are exactly those edges such that $(u',v',\edgeL)^\dagger\not\in \overrightarrow{E}(G(x))$, we have 
\begin{align}0&=\sum_{\substack{(u,v,\lambda )\\\in\overrightarrow{E}(G(x))}}\Ohm(\{u,v\},\lambda)\theta((u,v,\lambda )^\dagger)^2
=\sum_{\substack{(u,v,\lambda)\\\in \overrightarrow{E}(G(x))}}\Ohm(\{u,v\},\lambda)(\omega(u)-\omega(v))^2=\norm{\omega A\Pi_{H(x)}}^2,
\end{align}
so $\omega$ is a negative witness for $x$ in $P_{G,\Ohm}$. Thus:
\begin{align}
w_-(x,P_{G,\Ohm})\leq \norm{\omega A}^2&=\sum_{\substack{(u,v,\edgeL ) \in \overrightarrow{E}(G)}}\Ohm(\{u,v\},\lambda)(\omega(u)-\omega(v))^2\nonumber\\
&=\sum_{\substack{(u',v',\edgeL) \in\overrightarrow{E}(G'(x))}}\frac{\theta(u',v',\edgeL)^2}{\Ohm'(\{u',v'\},\lambda)}
=2R_{s',t'}(G'(x),\Ohm'),
\end{align}
completing the proof.
\end{proof}

%%%%%%%%%%%%%%%%%%%%%%%%%%%%%%%%%%%%%%%%%%%%%%%%%%%%%%%%%%%%%%%%%%%%%%%%%%%%%%%
%%%%%%%%%%%%%%%%%%%%%%%%%%%%%%%%%%%%%%%%%%%%%%%%%%%%%%%%%%% 

\subsection{Time and Space Analysis of the Span Program Algorithm for {\it st}-Connectivity}\label{app:time}

In this section, we will give an upper bound on the time complexity of $st$-\textsc{conn}$_G$ in terms of the time complexity of implementing a step of a discrete-time quantum walk on $G$. %The analysis follows relatively straightforwardly from \cite[Section 5.3]{BR12}, but we include it here for completeness.
 At the end of this section, we discuss the space complexity of the algorithm.

We first describe the algorithm that can be derived from a span program, following the conventions of \cite{IJ15}. Throughout this section, we will let $\Pi_S$ denote the orthogonal projector onto an inner product space $S$.  For a span program $P=(H,U,A,\tau)$, the corresponding algorithm performs phase estimation on the unitary $(2\Pi_{H(x)}-I)(2\Pi_{\ker A}-I)$ applied to initial state $\ket{w_0}=A^+\tau$, where $\Pi_{H(x)}$ denotes the orthogonal projector onto $H(x)$, and $\Pi_{\ker A}$ denotes the orthogonal projector onto the kernel of $A$, and $A^+$ denotes the pseudo-inverse of $A$. To decide a function $f$ on domain $D$, it is sufficient to perform phase estimation to precision $O\left(\sqrt{\max_{x\in D:f(x)=1}w_+(x)\times \max_{x\in D:f(x)=0}w_-(x)}\right)$.

In case of the $st$-connectivity span program $P_{G,\Ohm}$ in \eqref{eq:P}, it is a simple exercise to see that $2\Pi_{H(x)}-I$ can be implemented in $O(1)$ quantum operations, including 2 queries to $x$. The reflection $2\Pi_{\ker A}-I$ is independent of $x$, and so requires 0 queries to implement, however, it could still require a number of gates that grows quickly with the size of $G$. We will show that implementing $2\Pi_{\ker A}-I$ can be reduced to implementing a discrete-time quantum walk on $G$, a task which could be quite easy, depending on the structure of $G$ (for example, in the case that $G$ is a complete graph on $n$ vertices, this can be done in $O(\log n)$ gates \cite{BR12}).

For a multigraph $G$ and weight function $\Ohm$, we define a \emph{quantum walk step on $G$} to be a unitary $U_{G,\Ohm}$ that acts as follows for any $u\in V(G)$:
\begin{align}
U_{G,\Ohm}:\ket{u}\ket{0}\mapsto \frac{1}{\sqrt{\sum_{v,\lambda:(u,v,\lambda)\in \overrightarrow{E}(G)}\Ohm(\{u,v\},\lambda)}}\sum_{v,\lambda:(u,v,\lambda)\in \overrightarrow{E}(G)}\sqrt{\Ohm(\{u,v\},\lambda)}\ket{u}\ket{u,v,\lambda}.
\end{align}

\timeComp*
\noindent This theorem follows from Lemma \ref{lem:time}, stated below, and Lemma \ref{lem:initial-state}, which deals with the construction of the algorithm's initial state.

\setcounter{theorem}{31}
\begin{lemma}\label{lem:time}
Let $A$ be defined as in \eqref{eq:P}. Let $S_{G,\Ohm}$ be an upper bound on the time complexity of implementing $U_{G,\Ohm}$. Then $2\Pi_{\ker A}-I$ can be implemented to any constant precision in time complexity $O(S_{G,\Ohm}/\sqrt{\delta})$, where $\delta$ is the spectral gap of the symmetric normalized Laplacian of $(G,\Ohm)$.
\end{lemma}  
\begin{proof}
This analysis follows \cite{BR12} (see also \cite{IJ15}). Let 
\begin{align}
d(u)=\sum_{v,\lambda:(u,v,\lambda)\in \overrightarrow{E}(G)}\Ohm(\{u,v\},\lambda).
\end{align}

Define spaces $Z$ and $Y$ as follows. 
\begin{equation}
Z=\mathrm{span}\left\{ \ket{z_u}:=\sum_{v,\lambda:(u,v,\lambda)\in\overrightarrow{E}(G)}\frac{\sqrt{\Ohm(\{u,v\},\lambda)}}{\sqrt{2d(u)}}\left(\ket{0,u,u,v,\lambda}+\ket{1,u,v,u,\lambda}\right): u\in V(G) \right\}
\end{equation}

\begin{equation}
Y=\mathrm{span}\left\{
\ket{y_{u,v,\lambda}}:=\frac{\ket{0,u,u,v,\lambda}-\ket{1,v,u,v,\lambda}}{\sqrt{2}}:(u,v,\lambda)\in \overrightarrow{E}(G)
\right\}
\end{equation}

\noindent Define isometries whose column-spaces are $Z$ and $Y$ respectively:
\begin{align}
M_Z=\sum_{u\in V(G)}\ket{z_u}\bra{u}\quad\mbox{and}\quad M_Y=\sum_{(u,v,\lambda)\in\overrightarrow{E}(G)}\ket{y_{u,v,\lambda}}\bra{u,v,\lambda}.
\end{align}

\noindent Now we note that for any $(\{u,v\},\lambda)\in E(G)$, we have the following:
\begin{equation}
\braket{z_u}{y_{u,v,\lambda}}=\frac{\sqrt{\Ohm(\{u,v\},\lambda)}}{2\sqrt{d(u)}},\quad\mbox{and}\quad \braket{z_v}{y_{u,v,\lambda}}=-\frac{\sqrt{\Ohm(\{u,v\},\lambda)}}{2\sqrt{d(v)}}.
\end{equation}
Thus, we can calculate:
\begin{eqnarray}
M_Z^\dagger M_Y &=& \sum_{(u,v,\lambda)\in\overrightarrow{E}(G)}\left(\frac{\ket{u}}{2\sqrt{d(u)}}-\frac{\ket{v}}{2\sqrt{d(v)}}\right)\sqrt{\Ohm(\{u,v\},\lambda)}\bra{u,v,\lambda}\nonumber\\
&=& \sum_{u'\in V(G)}\frac{\ketbra{u'}{u'}}{2\sqrt{d(u')}}\sum_{(u,v,\lambda)\in\overrightarrow{E}(G)}\sqrt{\Ohm(\{u,v\},\lambda)}(\ket{u}-\ket{v})\bra{u,v,\lambda}\nonumber\\
&=&\sum_{u'\in V(G)}\frac{\ketbra{u'}{u'}}{2\sqrt{d(u')}} A =: \widetilde{A}.
\end{eqnarray}
While $AA^\dagger=2L$ is twice the Laplacian of $G$, $\widetilde{A}\widetilde{A}^\dagger$ is half the \emph{symmetric normalized Laplacian} of $G$, $L^{\mathrm{sym}}$:
\begin{eqnarray}
\widetilde{A}\widetilde{A}^\dagger %=\frac{1}{4}\sum_{u\in V(G)}\frac{\ket{u}\bra{u}}{\sqrt{d(u)}}AA^\dagger \sum_{u\in V(G)}\frac{\ket{u}\bra{u}}{\sqrt{d(u)}}
&=& \frac{1}{2}\sum_{u\in V(G)}\frac{\ket{u}\bra{u}}{\sqrt{d(u)}} L \sum_{u\in V(G)}\frac{\ket{u}\bra{u}}{\sqrt{d(u)}}=\frac{1}{2}L^{\mathrm{sym}}.
\end{eqnarray}
Define $W=(2\Pi_Z-I)(2\Pi_Y-I)$. We can define the discriminant of $W$ by $D(W)=M_Z^\dagger M_Y=\widetilde{A}$, and let $\{\cos\theta_j\}_j$ for $\theta_j\in [0,\pi/2]$ enumerate its singular values. 
By \cite{sze04}, the $(-1)$-eigenspace of $W$ is exactly $(Y\cap Z^\bot)\oplus (Y^\bot\cap Z)$, the $(+1)$-eigenspace is exactly $(Y\cap Z)\oplus (Y^\bot\cap Z^\bot)$, and the remaining eigenvalues of $W$ are $\{e^{\pm 2i\theta_j}\}_j$. Thus if $\delta$ is the smallest nonzero eigenvalue of $L^\mathrm{sym}$, then $\sqrt{\delta/2}$ is the smallest nonzero singular value of $\widetilde{A}$. Let $\tau\in [0,\pi/2]$ be such that $\cos\tau=\sqrt{\delta/2}$. Then the smallest non-zero phase of $-W$ is $\pm|\pi-2\tau|$. Using the identity $\cos \tau \leq \frac{\pi}{2}-\tau$ for any $\tau\in [0,\pi/2]$, we have phase gap at least 
\begin{equation}
|\pi-2\tau|\geq 2\cos\tau=\sqrt{2\delta}.
\end{equation}
Thus, using phase estimation to precision $\Theta(\sqrt{\delta})$, we can implement a reflection around the $(-1)$-eigenspace of $W$ to constant precision. Let $R_W$ denote this reflection. We now argue that $V=M_Y^\dagger R_W M_Y = 2\Pi_{\ker A}-I$.\footnote{An earlier version of this work erroneously claimed that $M_Y^\dagger W M_Y$ would implement this reflection. We thank Arjan Cornelissen and Alvaro Piedrafita for finding this error.} 
Note that the rows of $\widetilde{A}=M_Z^\dagger M_Y$ are non-zero multiples of the rows of $A$, so $\mathrm{row}(\widetilde{A})=\mathrm{row}(A)$, and thus $\ker(M_Z^\dagger M_Y)=\ker A$.
First, notice that on the image of $M_Y$, $Y$, $R_W$ acts as the reflection around $Y\cap Z^\bot$. 
Suppose $\ket{\psi}\in\ker A$. Then $\ket{\psi}\in\ker (M_Z^\dagger M_Y)=\ker(\widetilde{A})$, so $M_Y\ket{\psi}\in \ker M_Z^\dagger = Z^\bot$. On the other hand, if $M_Y\ket{\psi}\in Y\cap Z^\bot$, then since $M_Y\ket{\psi}\in Z^\bot$, we have $M_Z^\dagger M_Y\ket{\psi}=0$, so $\ket{\psi}\in \ker \widetilde{A}=\ker A$. Thus $M_Y$ maps $\ker A$ to the $(-1)$-eigenspace of $W$, and $(\ker A)^\bot$ to its orthogonal complement.

It only remains to argue that each of $M_Y$, $2\Pi_Z-I$ and $2\Pi_Y-I$ can be implemented in time complexity at most $O(S_{G,\Ohm})$. 

We first show that we can implement the isometry $M_Y$, or rather a unitary $U_Y$ that acts as $\ket{0}\ket{0}\ket{u,v,\lambda}\mapsto M_Y\ket{u,v,\lambda}=\ket{y_{u,v,\lambda}}$. 
First, use $HX$ on the first qubit to perform the map:
\begin{align}
\ket{0}\ket{0}\ket{u,v,\lambda}\mapsto\ket{-}\ket{0}\ket{u,v,\lambda}.
\end{align}
Conditioned on the value of the first register, copy either $u$ or $v$ into the second register to get:
\begin{align}
\frac{1}{\sqrt{2}}(\ket{0,u,u,v,\lambda}-\ket{1,v,u,v,\lambda})=\ket{y_{u,v,\lambda}}.
\end{align}
Thus, we can implement $U_Y$ in the time it takes to write down a vertex of $G$, $O(\log|V(G)|)$, which is at most $O(S_{G,\Ohm})$. Using the ability to implement $U_Y$, we can implement $2\Pi_Y-I$ as $U_YR_YU_Y^\dagger$, where $R_Y$ is the reflection that acts as the identity on computational basis states of the form $\ket{0}\ket{0}\ket{u,v,\lambda}$, and reflects computational basis states without this form. 

Next, we implement a unitary $U_Z$ that acts as $\ket{0}\ket{u}\ket{0}\mapsto M_Z\ket{u}=\ket{z_u}$. First, use the quantum walk step $U_{G,\Ohm}$, which can be implemented in time $S_{G,\Ohm}$, to perform:
\begin{align}
\ket{+}\ket{u}\ket{0}\mapsto\frac{1}{2\sqrt{d(u)}}\sum_{v,\lambda:(u,v,\lambda)\in\overrightarrow{E}(G)}\sqrt{\Ohm(\{u,v\},\lambda)}(\ket{0}+\ket{1})\ket{u}\ket{u,v,\lambda}.
\end{align}
Conditioned on the bit in the first register, swap the third and fourth registers, to get:
\begin{align}
\frac{1}{2\sqrt{d(u)}}\sum_{v,\lambda:(u,v,\lambda)\in\overrightarrow{E}(G)}\sqrt{\Ohm(\{u,v\},\lambda)}(\ket{0}\ket{u}\ket{u,v,\lambda}+\ket{1}\ket{u}\ket{v,u,\lambda})=\ket{z_u}.
\end{align}
The total cost of implementing $U_Z$ is $O(S_{G,\Ohm}+\log |V(G)|)=O(S_{G,\Ohm})$. Thus, we can implement $2\Pi_Z-I$ in $O(S_{G,\Ohm})$ quantum gates. 
\end{proof}

\begin{lemma}\label{lem:initial-state}
Let $A$ and $\tau$ be defined as in \eqref{eq:P}. Let $S_{G,\Ohm}$ be an upper bound on the complexity of implementing $U_{G,\Ohm}$. Then the initial state of the algorithm, $\frac{\ket{w_0}}{\norm{\ket{w_0}}}$ where $\ket{w_0}=A^+\tau$, can be approximated in time $O(S_{G,\Ohm})$. 
\end{lemma}
\begin{proof}
Without loss of generality, we can assume that $G$ includes the edge $(\{s,t\},\emptyset)$ (we can simply not include it in any subgraph). Furthermore, we  set $\Ohm(\{s,t\},\emptyset)=1/r$, for some positive $r$ to be specified later, so that $A\ket{s,t,\emptyset}=r^{-1/2}\tau$. This has no effect on other edges in $G$. 
Note that 
\begin{align}
\Pi_{(\ker A)^\bot}\ket{s,t,\emptyset}=A^+A\ket{s,t,\emptyset}=r^{-1/2}A^+\tau=r^{-1/2}\ket{w_0},
\end{align}
so 
\begin{align}
\ket{s,t,\emptyset}=r^{-1/2}\ket{w_0}+\ket{w_0^\bot}
\end{align}
for some $\ket{w_0^\bot}\in \ker A$. 
Thus, constant precision phase estimation on $2\Pi_{\ker A}-I$ maps $\ket{s,t,\emptyset}$ to 
\begin{align}
r^{-1/2}\ket{0}\ket{w_0}+\ket{1}\ket{w_0^\bot}.
\end{align}
Using quantum amplitude amplification \cite{brassard2002quantum}, we can amplify the amplitude on the $\ket{0}\ket{w_0}$ part of this arbitrarily close to 1 using a number of calls to $2\Pi_{\ker A}-I$ proportional to $\norm{r^{-1/2}\ket{w_0}}^{-1}$. 

In fact, it is straightforward to show that for any $\ket{\mu}\in \textrm{row}A$, the vector $\ket{\nu}$ with smallest norm that satisfies $A\ket{\nu}=\ket{\mu}$, is 
$A^+\ket{\mu}$ \cite{IJ15}. Using this fact along with Claim \ref{claim:pos-wit} and Definition \ref{def:posNegWit}, we have $\norm{\ket{w_0}}^2=R_{s,t}(G,\Ohm)$.

Let $R=R_{s,t}(G\setminus\{(\{s,t\},\emptyset)\},\Ohm)$ be the effective resistance of $G$ without the edge $(s,t,\emptyset)$. Now we can think of $(\{s,t\},\emptyset)$ and $G\setminus\{(\{s,t\},\emptyset)\}$ as two graphs in parallel, so using Claim \ref{claim:parallel_series}, we have
\begin{align}
\|\ket{w_0}\|^2=\frac{1}{1/R+1/r}.
\end{align}
Setting $r=R$, we have $\|\ket{w_0}\|^2=R/2$ and $\|r^{-1/2}\ket{w_0}\|^{-1}=O(1)$.
Thus, using $O(1)$ calls to $2\Pi_{\ker A}-I$, we can approximate the initial state $\ket{w_0}$. 
\end{proof}

Finally, we note that the space required by the algorithm, in addition to whatever auxiliary space we need to implement $U_{G,\Ohm}$, is $O(\max\{\log|E(G)|,\log|V(G)|\})$. $U_Y$ and $U_{G,\Ohm}$ each act on a Hilbert space of dimension less than $4|V(G)|^2|E(G)|$, so can in principle be implemented on $O(\max\{\log|E(G)|,\log|V(G)|\}+\log(1/\delta))$ qubits, where the $\log(1/\delta)$ term accounts for the phase register in the $O(\sqrt{\delta})$-precision phase estimation. However, a time-efficient implementation of $U_{G,c}$ may also make use of some number $S_{G,\Ohm}'$ of auxiliary qubits. We use these unitaries to perform phase estimation on $(2\Pi_{H(x)}-I)(2\Pi_{\ker A}-I)$ to precision 
\begin{align}
O\left(\min_{\Ohm}\sqrt{\max_{x\in D:\phi(x)=1} R_{s,t}(G_\phi(x),\Ohm) \times \max_{x\in D:\phi(x)=0}R_{s',t'}(G_\phi'(x),\Ohm')}\right)=O(|E(G)|).
\end{align} 
Thus we need $O(\log(|E(G)|)$ qubits to store the output of the phase estimation. Putting everything together gives the claimed space complexity.

%%%%%%%%%%%%%%%%%%%%%%%%%%%%%%%%%%%%%%%%%%%%%%%%%%%%%%%%%%%%%%%%%%%%%%%%
%%%%%%%%%%%%%%%%%%%%%%%%%%%%%%%%%%%%%%%%%%%%%%%%%%%%%%%%%%%%%%%%%%%%%%%%
%%%%%%%%%%%%%%%%%%%%%%%%%%%%%%%%%%%%%%%%%%%%%%%%%%%%%%%%%%%%%%%%%%%%%%%%
%%%%%%%%%%%%%%%%%%%%%%%%%%%%%%%%%%%%%%%%%%%%%%%%%%%%%%%%%%%%%%%%%%%%%%%%
%%%%%%%%%%%%%%%%%%%%%%%%%%%%%%%%%%%%%%%%%%%%%%%%%%%%%%%%%%%%%%%%%%%%%%%%
%%%%%%%%%%%%%%%%%%%%%%%%%%%%%%%%%%%%%%%%%%%%%%%%%%%%%%%%%%%%%%%%%%%%%%%%
%%%%%%%%%%%%%%%%%%%%%%%%%%%%%%%%%%%%%%%%%%%%%%%%%%%%%%%%%%%%%%%%%%%%%%%%
%%%%%%%%%%%%%%%%%%%%%%%%%%%%%%%%%%%%%%%%%%%%%%%%%%%%%%%%%%%%%%%%%%%%%%%%
%%%%%%%%%%%%%%%%%%%%%%%%%%%%%%%%%%%%%%%%%%%%%%%%%%%%%%%%%%%%%%%%%%%%%%%%
%%%%%%%%%%%%%%%%%%%%%%%%%%%%%%%%%%%%%%%%%%%%%%%%%%%%%%%%%%%%%%%%%%%%%%%%
%%%%%%%%%%%%%%%%%%%%%%%%%%%%%%%%%%%%%%%%%%%%%%%%%%%%%%%%%%%%%%%%%%%%%%%%
%%%%%%%%%%%%%%%%%%%%%%%%%%%%%%%%%%%%%%%%%%%%%%%%%%%%%%%%%%%%%%%%%%%%%%%%
%%%%%%%%%%%%%%%%%%%%%%%%%%%%%%%%%%%%%%%%%%%%%%%%%%%%%%%%%%%%%%%%%%%%%%%%
%%%%%%%%%%%%%%%%%%%%%%%%%%%%%%%%%%%%%%%%%%%%%%%%%%%%%%%%%%%%%%%%%%%%%%%%
%%%%%%%%%%%%%%%%%%%%%%%%%%%%%%%%%%%%%%%%%%%%%%%%%%%%%%%%%%%%%%%%%%%%%%%%
%%%%%%%%%%%%%%%%%%%%%%%%%%%%%%%%%%%%%%%%%%%%%%%%%%%%%%%%%%%%%%%%%%%%%%%%
%%%%%%%%%%%%%%%%%%%%%%%%%%%%%%%%%%%%%%%%%%%%%%%%%%%%%%%%%%%%%%%%%%%%%%%%
%%%%%%%%%%%%%%%%%%%%%%%%%%%%%%%%%%%%%%%%%%%%%%%%%%%%%%%%%%%%%%%%%%%%%%%%
%%%%%%%%%%%%%%%%%%%%%%%%%%%%%%%%%%%%%%%%%%%%%%%%%%%%%%%%%%%%%%%%%%%%%%%%

\section{Formula Evaluation and {\it st}-Connectivity}\label{app:formula}

In this section, we prove the correspondence between evaluating the formula $\phi$, and solving $st$-connectivity on the graph $G_\phi$. We first give a formal definition of $G_\phi$.
\begin{definition}[$G_\phi$]\label{def:Gphi}
If $\phi=x_i$ is a single-variable formula, then  $V(G_\phi)=\{s,t\}$ and $E(G_\phi)=\{(\{s,t\},x_i)\}$. 

If $\phi=\phi_1\wedge\dots\wedge \phi_l$, then define $V(G_\phi)=\{(i,v):i\in [\subf], v\in V(G_{\phi_i})\setminus\{s,t\}\}\cup\{s,s_2,\dots,s_{\subf},t\}$ and, letting $s_1=s$ and $s_{\subf+1}=t$, define:
\begin{align}
E(G_\phi)=&\left\{\left(\{(i,u),(i,v)\},x_j\right):i\in [\subf], u,v\in V(G_{\phi_i})\setminus\{s,t\}, (\{u,v\},x_j)\in E(G_{\phi_i})\right\}\nonumber\\
&\cup \{(\{(i,u),s_i\},x_j):i\in[\subf], u\in V(G_{\phi_i}),(\{s,u\},x_j)\in E(G_{\phi_i})\}\nonumber\\
&\cup \{(\{(i,u),s_{i+1}\},x_j):i\in [\subf], u\in V(G_{\phi_i}), (\{t,u\},x_j)\in E(G_{\phi_i})\}.
\end{align}

If $\phi=\phi_1\vee \dots\vee\phi_\subf$ define $V(G_\phi)=\{(i,v):i\in [\subf], v\in V(G_{\phi_i})\setminus \{s,t\}\}\cup \{s,t\}$ and 
\begin{eqnarray}
E(G_{\phi})=&\{(\{(i,u),(i,v)\},x_j):i\in [\subf], u,v\in V(G_{\phi_i})\setminus \{s,t\}, (\{u,v\},x_j)\in E(G_{\phi_i})\}\nonumber\\
&\cup \{(\{(i,u),s\},x_j):i\in [\subf], u\in V(G_{\phi_i}),(\{u,s\},x_j)\in E(G_{\phi_i})\}\nonumber\\
&\cup \{(\{(i,u),t\},x_j):i\in [\subf], u\in V(G_{\phi_i}), (\{u,t\},x_j)\in E(G_{\phi_i})\}.
\end{eqnarray}
\end{definition}

\noindent In order to prove Lemma \ref{lem:formula-st}, we will first prove Lemma \ref{claim:primeForm}:
\begin{restatable}{lemma}{primeForm}\label{claim:primeForm}
For an \textsc{and}-\textsc{or} formula $\phi$ on $\{0,1\}^N$, define $\phi'$ to be the formula obtained by replacing $\vee$-nodes with $\wedge$-nodes and $\wedge$-nodes with $\vee$-nodes in $\phi$. Then for all $x\in \{0,1\}^N$, if $\bar{x}$ denotes the bitwise complement of $x$, then $\phi'(x)=\neg\phi(\bar{x})$. Furthermore up to an isomorphism that maps $s$ to $s'$, $t$ to $t'$, and an edge labeled by any label $\lambda$ to an edge labeled by $\lambda$, we have $G_\phi'=G_{\phi'}$ and $G_\phi'(x)=G_{\phi'}(\bar{x})$.
\end{restatable}
\begin{proof}
The first part of the proof is by induction. Suppose $\phi$ has depth $0$, so $\phi=x_i$ for some variable $x_i$. Then $\phi'(x)=\phi(x)=\neg(\phi(\bar{x}))$. So suppose $\phi=\phi_1\wedge\dots\wedge\phi_\subf$. Then $\phi'=\phi_1'\vee\dots\vee\phi_\subf'$. Then by the induction hypothesis, 
\begin{align}
\phi'(x)=\phi'_1(x)\vee\dots\vee\phi_\subf'(x)=(\neg\phi_1(\bar{x}))\vee\dots\vee (\neg\phi_\subf(\bar{x}))=\neg(\phi_1(\bar{x})\wedge\dots\wedge\phi_\subf(\bar{x}))=\neg\phi(\bar{x})
\end{align}
where the second to last equality is de Morgan's law. The case $\phi=\phi_1\vee\dots\vee\phi_\subf$ is similar.

\begin{figure}[ht]
\centering
\begin{tikzpicture}

\draw [very thick] (0,0)--(0,1);
\draw [very thick] (0,1.5)--(0,3.25);
\node at (0,1.325) {$\vdots$};
\filldraw (0,3.25) circle (.075);
\filldraw (0,2.5) circle (.075);
\filldraw (0,1.75) circle (.075);
\filldraw (0,.75) circle (.075);
\filldraw (0,0) circle (.075);

\filldraw[gray] (-2,1.625) circle (.075);
\filldraw[gray] (2,1.625) circle (.075);

\draw[gray, very thick] plot [smooth, tension=1] coordinates {(-2,1.625) (0,2.875) (2,1.625)};
\draw[gray, very thick] plot [smooth, tension=1] coordinates {(-2,1.625) (0,2.125) (2,1.625)};
\draw[gray, very thick] plot [smooth, tension=1] coordinates {(-2,1.625) (0,0.375) (2,1.625)};

\node at (0,-.3) {$t$};
\node at (0,3.5) {$s$};
\node at (-2.25,1.625) {$s'$};
\node at (2.25,1.625) {$t'$};

\node at (.35,3) {$G_{\phi_1}$};
\node at (.35,2.25) {$G_{\phi_2}$};
\node at (.35,.15) {$G_{\phi_\subf}$};

\node at (-1.15,2.9) {\color{gray}$G'_{\phi_1}$};
\node at (-1.15,1.65) {\color{gray}$G'_{\phi_2}$};
\node at (-1.15,.45) {\color{gray}$G'_{\phi_\subf}$};

\draw[dashed] plot [smooth] coordinates {(0,0) (-2,0) (-3,1.625) (-2,3.1) (0,3.25)};
\draw[dashed,gray] plot [smooth] coordinates {(-2,1.625) (-1.65,3) (0,3.75) (1.65,3) (2,1.625)};

\node at (1.5,1.6) {\color{gray}$\vdots$};

\end{tikzpicture}
\caption{$\overline{G}_{\phi}$ shown in black, and its dual, $\overline{G}_{\phi'}$, shown in grey. The thick lines represent graphs. Edges in $G_{\phi_i}$ are dual to edges in $G'_{\phi_i}$, and the dotted edge $(\{s,t\},\emptyset)$ is dual to $(\{s',t'\},\emptyset)$.}\label{fig:dual-proof}
\end{figure}
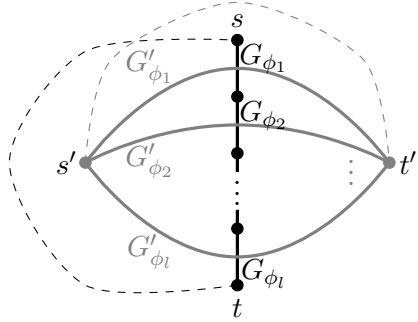

We will now prove that $\overline{G}_\phi^\dagger=\overline{G}_{\phi'}$, and furthermore, dual edges have the same label, by induction on the depth of $\phi$, from which the result follows immediately. 

If $\phi=x_i$ is a depth-0 formula, then $\phi'=x_i$. In that case, $G_\phi$ is just an edge from $s$ to $t$, labeled by $x_i$, and $G_{\phi}'$ is just an edge from $s'$ to $t'$ labeled $x_i$, so $G'_\phi=G_{\phi'}$.

For the inductive step, to show that $\overline{G}_{\phi}$ and $\overline{G}_{\phi'}$ are dual, and therefore $G_\phi' = G_{\phi'}.$ It suffices to exhibit a bijection $\zeta:V(\overline{G}_{\phi'})\rightarrow F(\overline{G}_{\phi})$ such that $(\{u,v\},x_j)\in E(\overline{G}_{\phi'})$ if and only if the faces $\zeta(u)$ and $\zeta(v)$ are separated by an edge in $E(\overline{G}_{\phi})$ with the label $x_j$. We first consider the case that $\phi=\phi_1\wedge\dots\wedge\phi_\subf$, so $\phi'=\phi_1'\vee\dots\vee\phi_\subf'$.  Then, $\overline{G}_\phi$ consists of the graphs $G_{\phi_1},\dots,G_{\phi_\subf}$, chained together in series as in Figure \ref{fig:dual-proof}, with an additional edge from $s$ to $t$, so the faces of $\overline{G}_\phi$ are exactly all the interior faces of each $G_{\phi_i}$, as well as the two faces on either side of the $st$-edge $(\{s,t\},\emptyset)$, which we will denote by $f^{s'}$ and $f^{t'}$. That is, adding an $i$ to the label of each internal face of $G_{\phi_i}$:
\begin{align}
F(\overline{G}_{\phi})=\{(i,f):i\in [\subf],f\in F(\overline{G}_{\phi_i})\setminus\{f^{s'},f^{t'}\}\}\cup\{f^{s'},f^{t'}\},
\end{align}
since $F(\overline{G}_{\phi_i})\setminus \{f^{s'},f^{t'}\}=F(G_{\phi_i})\setminus \{f^E\}$, where $f^E$ is the external face. Since $\phi'=\phi_1'\vee\dots\vee \phi_\subf'$ we also have
\begin{align}
V(\overline{G}_{\phi'})=V(G_{\phi'})=\{(i,v):i\in [\subf], v\in V(G_{\phi_i'})\setminus\{s,t\}\}\cup\{s',t'\},
\end{align}
where we will use the labels $s'$ and $t'$ in anticipation of the isometry between $G_\phi'$ and $G_{\phi'}.$

By the induction hypothesis, for each $i\in [\subf]$, there exists a bijection
$\zeta_i:V(\overline{G}_{\phi_i'})\rightarrow F(\overline{G}_{\phi_i})$ such
that for all $u,v\in V(\overline{G}_{\phi_i'})=V(G_{\phi_i'})$,
$(\{u,v\},x_j)\in E(\overline{G}_{\phi_i'})$ if and only if $\zeta_i(u)$ and
$\zeta_i(v)$ are faces separated by an edge with the label $x_j$. We define
$\zeta$ by $\zeta(i,v)=(i,\zeta_i(v))$ for all $i\in [\subf]$ and $v\in
V(G_{\phi_i'})\setminus\{s,t\}$, $\zeta(s')=f^{s'}$, and $\zeta(t')=f^{t'}$. By
the induction hypothesis, we can see that for any edge $(\{u,v\},x_j)\in
E(\overline{G}_{\phi'})\setminus(\{s',t'\},\emptyset)$, $\zeta(u)$ and $\zeta(v)$ are separated by an edge
labeled $x_j$. This is because this edge is in one of the $G_{\phi_i'}$, and so
it has a dual edge in $G_{\phi_i}$, by the induction hypothesis (see Figure
\ref{fig:dual-proof}). The only other edge in $\overline{G}_{\phi'}$ is the
edge $(\{s',t'\},\emptyset)$, and $\zeta(s')$ and $\zeta(t')$ are exactly those
faces on either side of $(\{s,t\},\emptyset)$ in $\overline{G}_{\phi}$,
completing the proof that $\overline{G}_\phi^\dagger=\overline{G}_{\phi'}$.

If $\phi=\phi_1\vee\dots\vee \phi_\subf$, then $\phi'=\phi_1'\wedge\dots\wedge\phi_\subf'$, and a nearly identical proof shows that $\overline{G}_\phi^\dagger=\overline{G}_{\phi'}$. 

Now that we have shown an isomorphism between $G'_\phi$ and $G_{\phi'}$, note that $G'_\phi(x)$ is the subgraph of $G'_\phi$ that includes all those edges where $x_e=0$. On the other hand $G_{\phi'}(x)$ is the graph that includes all those edges where $x_e=1.$ Taking the bitwise negation of $x$, we find that 
$G'_\phi(x)=G_{\phi'}(\bar{x})$.
\end{proof}

\noindent Lemma \ref{claim:primeForm} allows us to prove Claim \ref{claim:compose}:
\compose*

\begin{proof}%[Proof of Claim \ref{claim:compose}]
If $\phi = \phi_1\vee\cdots\vee\phi_l$, then 
$\phi' = \phi_1'\wedge\cdots\wedge\phi_l'$. From Lemma \ref{claim:primeForm}, $G'_\phi(x) = G_{\phi'}(\bar{x})$, which using Definition \ref{def:Gphi} is composed of $\{G_{\phi_i'}(\bar{x})\}_{i=1}^{\subf}$ in series. But using the isomorphism of Lemma \ref{claim:primeForm} again, this is just $\{G'_{\phi_i}(x)\}_{i=1}^\subf$ composed in series. 
The proof for $\phi = \phi_1\wedge\cdots\wedge\phi_l$ is similar.
\end{proof}

Now we can prove Lemma \ref{lem:formula-st}, which relates the existence of a path in $G_\phi(x)$ or $G'_\phi(x)$ to the value of the function $\phi(x):$
\formulast*
\begin{proof}
We will prove the statement by induction on the depth of $\phi$. If $\phi=x_j$ has depth 0, then $G_{\phi}$ is just an edge $(\{s,t\},x_j)$, and $G_{\phi}'$ is just an edge $(\{s',t'\},x_j)$. Thus $s$ and $t$ are connected in $G_{\phi}(x)$ if and only if $x_j=1$, in which case $\phi$ evaluates to 1, and $s'$ and $t'$ are connected in $G_{\phi'}$ if and only if $x_j=0$, in which case $\phi$ evaluates to $0$. 

If $\phi=\phi_1\wedge\dots\wedge \phi_\subf$, then $G_{\phi}$ consists of $G_{\phi_1},\dots,G_{\phi_\subf}$ connected in series from $s$ to $t$, and moreover, $G_{\phi}(x)$ consists of $G_{\phi_1}(x),\dots,G_{\phi_\subf}(x)$ connected in series from $s$ to $t$. Thus an $st$-path in $G_{\phi}(x)$ consists of an $st$-path in $G_{\phi_1}(x)$, followed by an $st$-path in $G_{\phi_2}(x)$, etc., up to an $st$-path in $G_{\phi_\subf}(x)$. Thus, $s$ and $t$ are connected in $G_{\phi}(x)$ if and only if $s$ and $t$ are connected in each $G_{\phi_1}(x),\dots,G_{\phi_\subf}(x)$, which happens if and only if $\phi_1(x)\wedge\dots\wedge \phi_\subf(x)=1$. 

On the other hand, by Claim \ref{claim:compose}, $G_{\phi}'$ consists of $G_{\phi_1}',\dots,G_{\phi_\subf}'$ connected in parallel between $s'$ and $t'$. So any $s't'$-path in $G_{\phi}'(x)$ is an $s't'$-path in one of the $G_{\phi_i}'(x)$, which is equivalent to an $st$-path in one of $G_{\phi_i'}(\bar{x})$. 
Thus, by Lemma \ref{claim:primeForm} $s'$ and $t'$ are connected in $G_{\phi}'(x)$ if and only if $\phi_1'(\bar{x})\vee \dots\vee\phi_\subf'(\bar{x})=
\neg\phi_1(x)\vee \dots\vee\neg\phi_\subf(x)=1$. By de Morgan's law is true if and only if $\phi(x)=\phi_1(x)\wedge \dots\wedge\phi_\subf(x)=0$. 

\noindent The case when $\phi=\phi_1\vee\dots\vee\phi_\subf$ is similar. 
\end{proof}

%%%%%%%%%%%%%%%%%%%%%%%%%%%%%%%%%%%%%%%%%%%%%%%%%%%%%%%%%%%%%%%%%%%%%%%%%
%%%%%%%%%%%%%%%%%%%%%%%%%%%%%%%%%%%%%%%%%%%%%%%%%%%%%%%%%%%%%%%%%%%%%%%%%
%%%%%%%%%%%%%%%%%%%%%%%%%%%%%%%%%%%%%%%%%%%%%%%%%%%%%%%%%%%%%%%%%%%%%%%%%
%%%%%%%%%%%%%%%%%%%%%%%%%%%%%%%%%%%%%%%%%%%%%%%%%%%%%%%%%%%%%%%%%%%%%%%%%
%%%%%%%%%%%%%%%%%%%%%%%%%%%%%%%%%%%%%%%%%%%%%%%%%%%%%%%%%%%%%%%%%%%%%%%%%
%%%%%%%%%%%%%%%%%%%%%%%%%%%%%%%%%%%%%%%%%%%%%%%%%%%%%%%%%%%%%%%%%%%%%%%%%

\section{Classical Lower Bound on Class of Promise Boolean Formulas}\label{app:classical-lb}

In this section, we consider the query complexity of classical algorithms for \textsc{and}-\textsc{or} formulas, proving Theorem \ref{thm:class}. To do this, we use a recent tool from Ben-David and Kothari
\cite{BK16}. They show that the bounded-error classical randomized query complexity of a
function $f$, denoted $R(f)$, satisfies $R(f)=\Omega (RS(f))$, where $RS(f)$ is
the randomized sabotage complexity,
defined presently. Furthermore, they prove that for a composed function $f\circ g$, $RS(f\circ
g)\geq(RS(f)RS(g))$.

If $f:D\rightarrow\{0,1\}$, with $D\subseteq\{0,1\}^N$, let $f_{\textrm{sab}}:D_\textrm{sab}\rightarrow\{0,1\}$, where 
\begin{align}
D_\textrm{sab}=\{x\in \{0,1,*\}^N\cup\{0,1,\dagger\}^N:x \textrm{ is consistent with } y, y'\in D, 
\textrm{ s.t. } f(y)\neq f(y')\}.
\end{align}
We say $x\in\{0,1,*,\dagger\}^N$ is consistent with $y\in\{0,1\}^N$ if $x_i=y_i$ for all $i\in [N]$ such that $x_i\in\{0,1\}$.  Then, $f_\textrm{sab}(x)=1$
if $x\in\{0,1,*\}^N$, and $f_\textrm{sab}(x)=0$ if $x\in\{0,1,\dagger\}^N$.
Finally, the randomized sabotage complexity is given by $RS(f)=R_0(f_\textrm{sab}),$ where $R_0(f)$ is the zero-error randomized query complexity of $f$. (For further classical query complexity definitions, see \cite{BK16}.)

We first bound the sabotage complexity of $\textsc{and}|_{D_{N,h}}$ and $\textsc{or}|_{D_{N,h}'}$:
\begin{lemma}\label{lemma:sabCompAndOr}
$RS\left(\textsc{or}|_{D_{N,h}'}\right)=RS\left(\textsc{and}|_{D_{N,h}}\right)=\Omega(N/h)$.
\end{lemma}
\begin{proof}
For $x\in [D'_{N,h}]_\textrm{sab}$ to be consistent with $y,y'\in
D'_{N,h}$ such that $\textsc{or}(y)\neq\textsc{or}(y'),$
we must have that $x\in \{0,*\}^N\cup\{0,\dagger\}^N$. Furthermore, the number
of $*$'s or $\dagger$'s in $x$ must be at least $h$. Thus the sabotaged problem
reduces to finding at least one marked item out of $n$, promised there are at
least $h$ marked items. The randomized bounded-error query complexity of this task is $\Omega(N/h)$, and so by Theorem 3 in \cite{BK16},
\begin{align}
RS\left(\textsc{or}|_{D'_{N,h}}\right) =R_0\left((\textsc{or}|_{D'_{N,h}})_{\text{sab}}\right)=\Omega\left(R\left((\textsc{or}|_{D'_{N,h}})_{\text{sab}}\right)\right)=\Omega(N/h).
\end{align}
The proof for $\textsc{and}$ is similar.
\end{proof}

\noindent The next corollary follows immediately from Lemma \ref{lemma:sabCompAndOr} and the composition property of sabotage complexity:
\begin{corollary}\label{corr:classicalBound}
Let $\phi=\phi_1\circ \phi_2\circ\cdots\circ \phi_l$, where for
each $i\in[l]$, $\phi_i=\textsc{or}|_{D_{N_i,h_i}'}$ or $\phi_i=\textsc{and}|_{D_{N_i,h_i}}$. Then
$R(\phi)=\Omega\left(\prod_{i=1}^lN_i/h_i\right).$
\end{corollary}

Now that we understand the query complexity of symmetric composed \textsc{and}-\textsc{or}
formulas, we can look at how this compares to the quantum query complexity of
evaluating such functions. We now prove the following lemma.
\begin{lemma}\label{lemm:quantumBound}
Let $\phi=\phi_1\circ \phi_2\circ\cdots\circ \phi_l$, where for
each $i\in[l]$, $\phi_i=\textsc{or}|_{D_{N_i,h_i}'}$ or $\phi_i=\textsc{and}|_{D_{N_i,h_i}}$. Let $D$ be the domain of $\phi.$ Then
\begin{align}
\frac{\prod_{i=1}^lN_i}{\prod_{i=1}^lh_i}=\left(\max_{x\in D: \phi(x)=1}R_{s,t}(G_\phi(x))\right) \left(\max_{x\in D: \phi(x)=0}R_{s,t}(G'_\phi(x))\right).
\end{align}
\end{lemma}

\begin{proof}
The proof follows by induction on the number of compositions. First suppose that $\phi =\textsc{or}|_{D_{N,h}'}$. Then $G_\phi$ consists of $N$ edges connected in parallel between $s$ and $t$, and $G_\phi'$ consists of $N$ edges connected in series. The only input $x$ such that $\phi(x)=0$ is the all zeros input. Therefore $\max_{x\in D:\phi(x)=0}R_{s,t}(G_\phi'(x))=N.$ Now notice (using Claim \ref{claim:parallel_series}) that $R_{s,t}(G_\phi(x))=1/|x|.$ However because of the domain of $\textsc{or}_{N_i,h_i}$, inputs $x$ have $|x|\geq h$, so $\max_{x\in D:\phi(x)=1}R_{s,t}(G_\phi(x))=1/h.$ Thus 
\begin{align}
 N/h=\left(\max_{x\in D: \phi(x)=1}R_{s,t}(G_\phi(x))\right) \left(\max_{x\in D: \phi(x)=0}R_{s,t}(G'_\phi(x))\right).
\end{align} 
A similar analysis holds for the base case $\phi =\textsc{and}|_{D_{N,h}}$.

Now for the inductive step, let $\phi=\phi_1\circ \xi $ for $\xi= \phi_2\circ\dots\circ\phi_{\subf}$,  where for each $i$, $\phi_i$ is either $\textsc{or}|_{D_{N_i,h_i}'}$ or $\textsc{and}|_{D_{N_i,h_i}}$. Let $D_\xi$ be the domain of $\xi$ and  let $x^j\in D_\xi$ denote the bits of $x$ that are input to the $j^{\textrm{th}}$ copy of $\xi$. Suppose first that $\phi_1=\textsc{or}|_{D_{N_1,h_1}'}$.  $G_\phi'$ is formed by taking the $N_1$ graphs $G_{\xi}'$ and connecting them in series. 
The only way $\phi(x)=0$ is if the input $x^j\in D_\xi$ to each of the $\xi$ functions satisfies $\xi(x^j)=0$, so by Claim \ref{claim:parallel_series}
\begin{align}
\max_{x\in D: \phi(x)=0}R_{s,t}(G_\phi'(x))
=N_1\max_{y\in D_\xi:\xi(y)=0}R_{s,t}(G_{\xi}'(y)).
\end{align}

On the other hand, $G_\phi$ is formed by taking $N_1$ graphs $G_{\xi}$ and connecting them in parallel. Using Claim \ref{claim:parallel_series}, if $x^j\in D_\xi$ is the input to $j^{\textrm{th}}$ function $\xi$, we have
\begin{align}
R_{s,t}(G_\phi(x))=\left(\sum_{j=1}^N\frac{1}{R_{s,t}(G_{\xi}(x^j))}\right)^{-1}.
\end{align}
Thus larger values for $R_{s,t}(G_\phi(x))$ come from cases where
$R_{s,t}(G_{\xi}(x^j))$ are large. Now
\begin{align}
\max_{x\in D_\xi}R_{s,t}(G_{\xi}(x))=\infty,
\end{align}
which occurs when $\xi(y)=0$. 
Because of the promise on the domain of $\phi_1$, there must be at least $h_1$ of the $N_1$ subformulas $\xi$ that evaluate to $1$. On each of those subformulas, we want to
have an input $x^j\in D_\xi$ that maximizes the effective resistance of that
subformula. Therefore, we have
\begin{align}
\max_{x\in D:\phi(x)=1}R_{s,t}(G_\phi(x))=&\left(\frac{h_1}{\max_{y\in D_\xi:\xi(y)=1}R_{s,t}(G_{\xi_j}(y)}\right)^{-1}
=\frac{\max_{y\in D_\xi:\xi(y)=1}R_{s,t}(G_{\xi_j}(y))}{h_1}.
\end{align}
Therefore, using the inductive assumption,
\begin{align}
&\left(\max_{x\in D:\phi(x)=1}R_{s,t}(G_\phi(x))\right)\left(\max_{x\in D:\phi(x)=0}R_{s,t}(G_\phi'(x))\right)\nonumber\\
=&\frac{N_1}{h_1}\max_{y\in D_\xi:\xi(y)=1}R_{s,t}(G_{\xi_j}(y^j))\max_{y\in D_\xi:\xi(y)=0}R_{s,t}(G_{\xi}'(y))
=\frac{\prod_{i=1}^lN_i}{\prod_{i=1}^lh_i}.
\end{align}
The inductive step for $\phi_1 = \textsc{and}|_{D_{N,h}}$ is similar.
\end{proof}
 
\noindent Corollary \ref{corr:classicalBound} and Lemma \ref{lemm:quantumBound} give Theorem \ref{thm:class}.

%%%%%%%%%%%%%%%%%%%%%%%%%%%%%%%%%%%%%%%%%%%%%%%%%%%%%%%%%%%%%%%%%%%%%%%%%
%%%%%%%%%%%%%%%%%%%%%%%%%%%%%%%%%%%%%%%%%%%%%%%%%%%%%%%%%%%%%%%%%%%%%%%%%

%%%%%%%%%%%%%%%%%%%%%%%%%%%%%%%%%%%%%%%%%%%%%%%%%%%%%%%%%%%%%%%%%%%%%%%%%
%%%%%%%%%%%%%%%%%%%%%%%%%%%%%%%%%%%%%%%%%%%%%%%%%%%%%%%%%%%%%%%%%%%%%%%%%
%%%%%%%%%%%%%%%%%%%%%%%%%%%%%%%%%%%%%%%%%%%%%%%%%%%%%%%%%%%%%%%%%%%%%%%%%

%%%%%%%%%%%%%%%%%%%%%%%%%%%%%%%%%%%%%%%%%%%%%%%%%%%%%%%%%%%%%%%%%%%%%%%%%%%%%%%
%%%%%%%%%%%%%%%%%%%%%%%%%%%%%%%%%%%%%%%%%%%%%%%%%%%%%%%%%%% 

\section{NAND-tree Proofs}\label{app:proofs}

\subsection{Relationship Between Faults and Effective Resistance}
\noindent In this section, we prove 
 Lemma \ref{thm:kfault_conn}:
\CCeER*

\begin{proof}
We will give a proof for ${\cal F}_A(x)$; the case of ${\cal F}_B(x)$ is similar. 

First, $R_{s,t}(G_{\textsc{nand}_d}(x))=\infty$ if and only if $s$ and $t$ are
not connected in $G_{\textsc{nand}_d}(x)$, which, by Lemma \ref{lem:formula-st}, occurs if and only if $x$ is a 0-instance. This means exactly that $x$ is not $A$-winnable, which, by Eq.\ (\ref{eq:FaultDef}), holds
if and only if ${\cal F}_A(x)=\infty$. Thus, suppose this is not the case, so
${\cal F}_A(x)<\infty$.

The rest of the proof is by induction. We need to look at both odd and even cases. For the case of $d=0$, the only $A$-winnable input in $\{0,1\}^{2^0}$ is $x=1$. In that case, using Eq.\ (\ref{eq:FaultDef}), ${\cal F}_A(x)=1$, since there are no decision nodes for Player $A$, and since $G_{\textsc{nand}_0}(x)$ is just a single edge
from $s$ to $t$, $R_{s,t}(G_{\textsc{nand}_0}(x))=1$.

Let $x\in\{0,1\}^{2^{d}}$ be any $A$-winnable input with $d>1$. We let $x^0$ be the first $2^{d-1}$ bits of $x$ and $x^1$ be the last $2^{d-1}$ bits of $x$, so $x=(x^0,x^1)$.

We first consider odd $d>1$. Using the definition of $G_\phi$ from Section \ref{sec:nandgraphs}, and the fact that for $d$ odd, the root node is an $\wedge$-node, we see that $G_{\textsc{nand}_d}(x)$ consists of $G_{\textsc{nand}_{d-1}}(x^0)$ and $G_{\textsc{nand}_{d-1}}(x^1)$ connected in series, so by Claim \ref{claim:parallel_series} 
\begin{align}\label{eq:oddInd1}
R_{s,t}(G_{\textsc{nand}_d}(x))=R_{s,t}(G_{\textsc{nand}_{d-1}}(x^{0}))+
R_{s,t}(G_{\textsc{nand}_{d-1}}(x^{1})).
\end{align} 
Now the root can not be a fault, because it is a decision node for Player $B$, but we know the tree is $A$-winnable, so no choice Player $B$ makes would allow her to win the game. Therefore,  both subtrees connected to the root node must be $A$-winnable. Using Eq.\ (\ref{eq:FaultDef}) we have
\begin{align}\label{eq:oddInd2}
{\cal{F}}_A(x^0)+{\cal{F}}_A(x^1)\leq \max_{b\in\{0,1\}}{2\cal{F}}_A(x^b)=2{\cal{F}}_A(x).
\end{align}
Combining Eqs. (\ref{eq:oddInd1}) and (\ref{eq:oddInd2}) and the inductive assumption for even depth trees, we have for odd $d$, 
\begin{align}
R_{s,t}(G_{\textsc{nand}_d}(x))\leq 2{\cal{F}}_A(x).
\end{align}

Now we consider the case that $d$ is even, so the root is a decision node for Player $A$. Consequently, the root node is a $\vee$-node, so by Claim \ref{claim:parallel_series}
\begin{align}
R_{s,t}(G_{\textsc{nand}_d}(x))=\left(\frac{1}{R_{s,t}(G_{\textsc{nand}_{d-1}}(x^{0}))}+\frac{1}{R_{s,t}(G_{\textsc{nand}_{d-1}}(x^{1}))}\right)^{-1}.
\label{eq:step5}
\end{align}

Suppose the root is a fault.
Without loss of generality, let's assume the subtree with input $x^0$ is not
$A$-winnable. Then $R_{s,t}(G_{\textsc{nand}_{d-1}}(x^{0}))=\infty$
so Eq.\ (\ref{eq:step5}) becomes
\begin{align}
R_{s,t}(G_{\textsc{nand}_d}(x))=R_{s,t}(G_{\textsc{nand}_{d-1}}(x^{1})).
\end{align}
Using the inductive assumption for odd depth trees, Eq.\ (\ref{eq:FaultDef}), and the fact that the root is a fault, we have
\begin{align}
R_{s,t}(G_{\textsc{nand}_d}(x))\leq 2{\cal {F}}(x^1)={\cal {F}}(x).
\end{align}

If the root is not a fault, then both $R_{s,t}(G_{\textsc{nand}_{d-1}}(x^{0}))$ and $R_{s,t}(G_{\textsc{nand}_{d-1}}(x^{1}))$ are finite, so from \eqref{eq:step5}, and using the inductive assumption, we have 
\begin{align}
R_{s,t}(G_{\textsc{nand}_d}(x)) & \leq \frac{1}{2}\max\{R_{s,t}(G_{\textsc{nand}_{d-1}}(x^{0})),R_{s,t}(G_{\textsc{nand}_{d-1}}(x^{1}))\}\nonumber\\
& \leq \max\{{\cal{F}}(x^0),{\cal{F}}(x^1)\}
={\cal{F}}(x).
\end{align}

\noindent A similar analysis for ${\cal F}_B(x)$ completes the proof. 
\end{proof}

\subsection{Estimating Effective Resistances}\label{app:approx-analysis}
In this section, we will prove Lemma \ref{lem:est}, which bounds the query complexity of estimating the effective resistance of a graph corresponding to a Boolean formula.
 In \cite{IJ15}, Ito and Jeffery  describe a quantum query algorithm to estimate the positive or negative
witness size of a span program given access to $\mathcal{O}_x$. We will describe how to use this algorithm to estimate the effective resistance of graphs $G_\phi(x)$ or $G'_\phi(x)$.

Ref.~\cite{IJ15} define the
\emph{approximate positive and negative witness sizes}, $\tilde{w}_+(x,P)$ and $\tilde{w}_-(x,P)$. These are similar to the positive and negative witness sizes, but with the conditions $\ket{w}\in H(x)$ and $\omega A\Pi_{H(x)}=0$ relaxed. 

\begin{definition}[Approximate Positive Witness]
For any span program $P$ on $\{0,1\}^N$ and $x\in\{0,1\}^N$, we define the
\emph{positive error of $x$ in $P$} as: 
\begin{equation}
e_+(x)=e_+(x,P):=\min\left\{
\norm{\Pi_{H(x)^\bot}\ket{w}}^2:A\ket{w}=\tau\right\}.
\end{equation}
We say $\ket{w}$ is an \emph{approximate positive witness} for
$x$ in $P$ if $\norm{\Pi_{H(x)^\bot}\ket{w}}^2=e_+(x)$ and $A\ket{w}=\tau.$
 We define the \emph{approximate positive witness size} as
\begin{equation}
\tilde{w}_+(x)=\tilde{w}_+(x,P):=\min\left\{\norm{\ket{w}}^2:A\ket{w}=\tau,
\norm{\Pi_{H(x)^\bot}\ket{w}}^2=e_+(x)\right\}.
\end{equation}
\end{definition}

\noindent  If $x\in P_1$, then $e_+(x)=0$. In that case, an
approximate positive witness for $x$ is a positive witness, and
$\tilde w_+(x)=w_+(x)$. For negative inputs,
the positive error is larger than 0. 

\noindent We can define a similar notion of approximate negative witnesses: 

\begin{definition}[Approximate Negative Witness]
For any span program $P$ on $\{0,1\}^N$ and $x\in\{0,1\}^N$,
we define the \emph{negative error of $x$ in $P$} as:
\begin{equation}
e_-(x)=e_-(x,P):=\min\left\{\norm{\omega A\Pi_{H(x)}}^2: 
\omega\in \mathcal{L}(U,\mathbb R), \omega\tau=1\right\}.
\end{equation}
Any $\omega$ such that $\norm{\omega A\Pi_{H(x)}}^2=e_-(x,P)$ is called an 
\emph{approximate negative witness} for $x$ in $P$. We define 
the \emph{approximate negative witness size} as
\begin{equation}
\tilde{w}_-(x)=\tilde{w}_-(x,P):=\min\left\{\norm{\omega A}^2:
\omega\in \mathcal{L}(U,\mathbb R), 
\omega\tau=1,\norm{\omega A\Pi_{H(x)}}^2=e_-(x,P)\right\}.
\end{equation}
\end{definition}

\noindent  If $x\in P_0$, then $e_-(x)=0$. In that case, an
approximate negative witness for $x$ is a negative witness, and
$\tilde w_-(x)=w_-(x)$. For positive inputs,
the negative error is larger than 0. 

Then Ito and Jeffery give the following result:

\begin{theorem}[\cite{IJ15}]\label{theorem:span-est}
Fix $X\subseteq\{0,1\}^N$ and $f:X\rightarrow \mathbb{R}_{\geq 0}$. Let $P$ be a
span program such that for all $x\in X$, $f(x)=w_+(x,P)$ and define
$\widetilde{W}_-=\widetilde{W}_-(P,f)=\max_{x\in X}\tilde{w}_-(x,P)$.
There exists a quantum algorithm that estimates $f$ to relative error $\eps$
and that uses $\tO\left(\frac{1}{\eps^{3/2}}\sqrt{w_+(x)\widetilde{W}_-}\right)$ queries. 
Similarly, let $P$ be a span program such that for all $x\in X$,
$f(x)=w_-(x,P)$ and define
$\widetilde{W}_+=\widetilde{W}_+(P,f)=\max_{x\in X}\tilde{w}_+(x,P)${}.
Then there exists a quantum algorithm that estimates $f$ to accuracy
$\eps$ and that uses $\tO\left(\frac{1}{\eps^{3/2}}\sqrt{w_-(x)\widetilde{W}_+}\right)$
queries.
\end{theorem}

We will apply Theorem \ref{theorem:span-est} to the span program $P_{G,c}$ defined in Eq.\ \eqref{eq:P}, with $G=G_\phi$.
Throughout this section, we will always set the weight function
$\Ohm$ to take value one on all edges of the graph $G$. In this, case, to simplify notation, we will denote the span program $P_{G,\Ohm}$ as $P_G.$
To apply Theorem \ref{theorem:span-est}, we need bounds on $\widetilde{W}_+(P_{G_\phi})$ and $\widetilde{W}_-(P_{G_\phi})$. We will prove:
\begin{restatable}{lemma}{appnegwit}\label{lemma:approx_neg_wit}
 For any formula $\phi$, its $\wedge$-depth is the largest number of
 $\wedge$-labeled nodes on any path from the root to a leaf.  Let $\phi$ be any
 \textsc{and}-\textsc{or} formula with maximum fan-in $\subf$, $\wedge$-depth
 $\depth_\wedge$, and $\vee$-depth $\depth_\vee$. Then
 $\widetilde{W}_+(P_{G_{\phi}})\leq \frac{1}{2}\subf^{\depth_\wedge}$ and
 $\widetilde{W}_-(P_{G_{\phi}})\leq 2\subf^{\depth_\vee}$.
\end{restatable}

Then, applying Lemma \ref{lemma:approx_neg_wit} and Theorem 
\ref{theorem:span-est}, we have the main result of this section, which was first stated in Section \ref{sec:nandWin}:

\appWit*
\begin{proof}[Proof of Lemma \ref{lem:est}]
By Theorem \ref{theorem:span-est}, since $R_{s,t}(G_\phi(x))=\frac{1}{2}w_+(x, P_{G_\phi})$ (Lemma \ref{lemma:both_witnesses}), 
we can estimate this quantity using a number of queries that depends on $\widetilde{W}_-(P_{G_{\phi}})$. By Lemma \ref{lemma:approx_neg_wit}, we have that 
$\widetilde{W}_-(P_{G_{\phi}})\leq 2\subf^{\depth_\vee},$ 
so we can estimate $w_+(x)=R_{s,t}(G_\phi(x))$ in 
$\textstyle\widetilde{O}\left(\frac{1}{\eps^{2/3}}\sqrt{w_+(x)}\widetilde{W}_-^{1/2}\right)=\widetilde{O}\left(\frac{1}{\eps^{2/3}}\sqrt{R_{s,t}(G_\phi(x))\subf^{\depth_\vee}}\right) $
queries. Similarly, $R_{s,t}(G'_\phi(x))={2}w_-(x,P_{G_\phi})$ for all 0-instances, and $\widetilde{W}_+\leq \frac{1}{2}\subf^{\depth_\wedge}$, so  we can estimate $R_{s,t}(G'_\phi(x))$ in $\widetilde{O}\left(\frac{1}{\eps^{2/3}}\sqrt{R_{s,t}(G_\phi'(x))\subf^{\depth_\wedge}}\right)$ queries.
\end{proof}

To prove Lemma \ref{lemma:approx_neg_wit}, we will use
the following observation, which gives an upper bound on the length of the longest self-avoiding $st$-path in $G_\phi$, in terms of the $\wedge$-depth of $\phi$. This bound is not tight in general. 

\begin{claim}\label{claim:depths}
Let $\phi$ be an \textsc{and}-\textsc{or} formula with constant fan-in $\subf$. If $\phi$ has $\wedge$-depth $\depth_\wedge$, then the longest self-avoiding path connecting $s$ and $t$ in $G_\phi$ has length at most $\subf^{\depth_\wedge}$. 
\end{claim}
\begin{proof}
We will prove the statement by induction. If $\phi$ has $\wedge$-depth $\depth_\wedge=0$, then it has no $\wedge$-nodes. Thus, it is easy to see that $G_\phi$ has only two vertices, $s$ and $t$, with some number of edges connecting them, so every $st$-path has length 1. 

Suppose $\phi$ has $\wedge$-depth $\depth_\wedge>0$. First, suppose $\phi=\phi_1\wedge\dots\wedge\phi_\subf$. Then since $G_\phi$ consists of $G_{\phi_1},\dots,G_{\phi_\subf}$ connected in series, any $st$-path in $G_\phi$ consists of an $st$-path in $G_{\phi_1}$, followed by an $st$-path in $G_{\phi_2}$, etc. up to an $st$-path in $G_{\phi_\subf}$, so if $\depth_\wedge(\phi_i)$ is the $\wedge$-depth of $\phi_i$, then the longest $st$-path in $G_\phi$ has length at most:
\begin{align}
\subf^{\depth_\wedge(\phi_1)}+\dots+\subf^{\depth_\wedge(\phi_\subf)}\leq \subf\subf^{\depth_\wedge-1}=\subf^{\depth_\wedge}.
\end{align}

If $\phi=\phi_1\vee\dots\vee\phi_\subf$, then $\max_i \depth_\wedge(\phi_i)=\depth_\wedge(\phi)=\depth_\wedge$, 
and $G_{\phi}$ consists of $G_{\phi_1},\dots,G_{\phi_\subf}$, connected in parallel. Any self-avoiding $st$-path must include exactly one edge adjacent to $s$ and one edge adjacent to $t$. However, any path that includes an edge from $G_{\phi_i}$ and $G_{\phi_j}$ for $i\neq j$ must go through $s$ or $t$, so it must have more than one edge adjacent to $s$, or more than one edge adjacent to $t$, so such a path can never be a self-avoiding $st$-path. Thus, any self-avoiding $st$-path must be contained completely in one of the $G_{\phi_i}$. The longest such path is thus the longest self-avoiding $st$-path in any of the $G_{\phi_i}$, which, by induction, is $\max_i \subf^{\depth_\wedge({\phi_i})}= \subf^{\depth_\wedge}$. 
\end{proof}

\noindent Now we can prove Lemma \ref{lemma:approx_neg_wit}:

\begin{proof}[Proof of Lemma \ref{lemma:approx_neg_wit}]
To begin, we will prove the upper bound on $\widetilde{W}_+$. Suppose $\ket{\tilde w}$ is an optimal approximate positive witness for $x$. By Claim \ref{claim:pos-wit}, if $\ket{\tilde w}$ is an approximate positive witness, then since $A\ket{\tilde w}=\tau$, and $\Ohm$ has unit value on all edges of $G$, $\theta(u,v,\edgeL )=\braket{u,v,\edgeL }{\tilde w}-\braket{v,u,\edgeL}{\tilde w}$ is a unit flow on $G$. Since $\ket{\tilde w}$ is an approximate positive witness for $x$, it has minimal error for $x$, so it minimizes $\norm{\Pi_{H(x)^\bot}\ket{\tilde{w}}}^2$, and since it is optimal, it minimizes $\norm{\ket{\tilde{w}}}^2$ over all approximate positive witnesses. Define $\ket{\theta}=\sum_{(u,v,\edgeL )\in\overrightarrow{E}(G)}\theta(u,v,\edgeL )\ket{u,v,\edgeL }$, so we know that $\frac{1}{2}\ket{\theta}$ also maps to $\tau$ under $A$, so is also a positive witness in $P_{G_\phi}$.Then we have 
\begin{equation}
\norm{\Pi_{H(x)^\bot}\ket{\theta}}^2=2\!\!\!\!\sum_{\substack{(u,v,\lambda)\in\\ \overrightarrow{E}(G)\setminus \overrightarrow{E}(G(x))}}\!\!\!\!\braket{u,v,\lambda}{\tilde{w}}^2-2\!\!\!\!\sum_{\substack{(u,v,\lambda)\in\\ \overrightarrow{E}(G)\setminus \overrightarrow{E}(G(x))}}\!\!\!\!\braket{u,v,\lambda}{\tilde{w}}\braket{v,u,\lambda}{\tilde{w}}\leq \norm{2\Pi_{H(x)^\bot}\ket{\tilde w}}^2,
\end{equation}
where the last inequality uses Cauchy-Schwarz, so $\frac{1}{2}\ket{\theta}$ is also an approximate positive witness for $x$. Similarly,
\begin{equation}
\|\ket{\theta}\|^2\leq \| 2\ket{\tilde w}\|^2,
\end{equation}
so $\frac{1}{2}\ket{\theta}$ is optimal.

By Claim \ref{claim:flow-decomposition}, we can consider a decomposition of $\ket{\theta}$ into self-avoiding paths $p_i$ and cycles $c_i$ such that all cycles are disjoint from all paths,
$\ket{\theta}=\sum_{i=1}^r\alpha_i\ket{p_i}+\sum_{i=1}^{r'}\beta_i\ket{c_i}$,
where for each $i$,
\begin{align}
\ket{p_i}=\sum_{j=1}^{L_i}\ket{u^{(i)}_j,u^{(i)}_{j+1},\lambda_{i,j}}-\sum_{j=1}^{L_i}\ket{u^{(i)}_{j+1},u^{(i)}_{j},\lambda_{i,j}},\\
\ket{c_i}=\sum_{j=1}^{L_i'}\ket{v^{(i)}_j,v^{(i)}_{j+1},\lambda_{i,j}'}-\sum_{j=1}^{L_i'}\ket{v^{(i)}_{j+1},v^{(i)}_{j},\lambda_{i,j}'}
\end{align}
where $v^{(i)}_{L_j'+1}=v^{(i)}_1$ and $\{\lambda_{i,j}\}_{i,j}\cap \{\lambda_{i,j}'\}_{i,j}=\emptyset$. 
It's easy to see (in the case of unit edge weights) that $A\ket{c_i}=0$ for all $i$, so 
\begin{align}
A\frac{1}{2}\sum_{i=1}^r\alpha_i\ket{p_i}=A\frac{1}{2}\ket{\theta}=\tau.
\end{align}
Let $\ket{\theta'}=\sum_{i=1}^r\alpha_i\ket{p_i}$. Then since $c_i$ and $p_j$ have no common edges, we have $\braket{c_i}{p_j}=0$, and also $\bra{c_i}(I-\Pi_{H(x)})\ket{p_j}=0$, so the error of $\frac{1}{2}\ket{\theta'}$ is $\frac{1}{4}\norm{\Pi_{H(x)^\bot}\ket{\theta'}}^2\leq \frac{1}{4}\norm{\Pi_{H(x)^\bot}\ket{\theta}}^2$, so $\frac{1}{2}\ket{\theta'}$ also has minimal error. Furthermore, $\norm{\ket{\theta'}}^2\leq \norm{\ket{\theta}}^2$, with equality if and only if there are no cycles in the decomposition. By the optimality of $\frac{1}{2}\ket{\theta}$ as an approximate positive witness for $x$, we can conclude that $\ket{\theta}=\sum_{i=1}^r\alpha_i\ket{p_i}$, and since $A\ket{p_i}=2\tau$ for all $i$, and $A\ket{\theta}=2\tau$, we have $\sum_{i=1}^r{\alpha_i}=1$. Then
\begin{align}
\norm{\ket{\theta}}^2\leq \max_i\norm{\ket{p_i}}^2=\max_i 2L_i.
\end{align}
Since the longest self-avoiding $st$-path in $G_{\phi}$ has length at most $\subf^{\depth_\wedge}$, and each $L_i$ is the length of a self-avoiding path in $G_{\phi}$, we have $\tilde{w}_+(x,P_{G_\phi})\leq \frac{1}{4}2 \subf^{\depth_\wedge}=\frac{1}{2}\subf^{\depth_\wedge}$. Thus $\widetilde{W}_+(P_{G_\phi})\leq \frac{1}{2}\subf^{\depth_\wedge}$.

Next we prove the bound on $\widetilde{W}_-$. A min-error approximate negative witness for $x$ in $P_{G_\phi}$ is a function $\omega:V(G_\phi)\rightarrow\mathbb{R}$ such that $\omega\tau=\omega(s)-\omega(t)=1$, and $\norm{\omega A\Pi_{H(x)}}^2=\sum_{(u,v,\edgeL )\in\overrightarrow{E}(G_{\phi}(x))}(\omega(u)-\omega(v))^2$ is minimized. By Claim \ref{claim:neg-wit}, since $\omega\tau=1$, the function $\theta:\overrightarrow{E}(G'_\phi)\rightarrow\mathbb{R}$ defined by $\theta((u,v,\edgeL )^\dagger)=\omega(u)-\omega(v)$ is a unit $s't'$-flow on $G'_\phi=G_{\phi'}$, and the witness complexity is 
\begin{align}
\norm{\omega A}^2=\sum_{(u,v,\edgeL )\in \overrightarrow{E}(G_\phi)}(\omega(u)-\omega(v))^2=\sum_{(u',v',\edgeL )\in \overrightarrow{E}(G'_\phi)}\theta(u',v',\edgeL )^2=\|\ket{\theta}\|^2
\end{align}
where we create $\ket{\theta}$ from $\theta$ in the usual way. 
By an argument similar to the previous argument, if $\omega$ is an optimal approximate negative witness for $x$, then $\|\ket{\theta}\|^2$ is upper bounded by twice the length of the longest self-avoiding $s't'$-path in $G_\phi'=G_{\phi'}$. By Lemma \ref{claim:primeForm} and Claim \ref{claim:depths}, this is upper bounded by $2\subf^{\depth_\wedge(\phi')}=2\subf^{\depth_\vee(\phi)}$, where $\depth_\wedge(\phi')$ is the $\wedge$-depth of $\phi'$, and $\depth_\vee=\depth_\vee(\phi)$ is the $\vee$-depth of $\phi$. Thus 
$\tilde{w}_-(x,P_{G_\phi}) \leq 2 \subf^{\depth_\vee}$, and so $\widetilde{W}_-\leq 2\subf^{\depth_\vee}$.
\end{proof}

\subsection{Winning the NAND-tree}\label{app:nandWin}

\noindent We now analyze the algorithm for winning the game associated with a $\textsc{nand}$-tree, proving Lemma~\ref{lemma:select} and Theorem \ref{thm:winning}. 
\select*

\begin{proof}
Since at least one of $x^0$ and $x^1$ is a 1-instance, using the description of $\texttt{Select}$ in Section \ref{sec:nandWin}, at least one of the programs will terminate. Suppose without loss of generality that $\texttt{Est}(x^0)$ is the first to terminate, outputting $w_0$. Then there are two possibilities: $\texttt{Est}(x^1)$ does not terminate after $p(d)\sqrt{w_0}N^{1/4}$ steps, in which case, $R_{s,t}(G_{\textsc{nand}_d}(x^0))\leq 2R_{s,t}(G_{\textsc{nand}_d}(x^1))$, and \texttt{Select} outputs 0; or $\texttt{Est}(x^1)$ outputs $w_1$ before $p(d)\sqrt{w_0}N^{1/4}$ steps have passed and \texttt{Select} outputs $b$ such that $w_b\leq w_{\bar{b}}$. 

We will prove the first case by contradiction. Suppose
\begin{align}
2R_{s,t}(G_{\textsc{nand}_d}(x^1))<R_{s,t}(G_{\textsc{nand}_d}(x^0)).
\end{align}
Then $\texttt{Est}(x^1)$ must terminate after 
\begin{align}\label{eq:St1}
 p(d)\sqrt{R_{s,t}(G_{\textsc{nand}_d}(x^1))}N^{1/4}\leq \frac{1}{\sqrt{2}}p(d)\sqrt{R_{s,t}(G_{\textsc{nand}_d}(x^0))}N^{1/4}
 \end{align} steps. 
In $\mathtt{Select}$, we run $\mathtt{Est}$ to relative accuracy $\eps=1/3$, so we have 
\begin{align}
|w_0-R_{s,t}(G_{\textsc{nand}_d}(x^0))|\leq \frac{1}{3} R_{s,t}(G_{\textsc{nand}_d}(x^0)),
\end{align} 
and so 
\begin{align}\label{eq:St2}
 w_0\geq \frac{2}{3}R_{s,t}(G_{\textsc{nand}_d}(x^0)).
 \end{align} 
Plugging Eq.\ (\ref{eq:St2}) into Eq.\ (\ref{eq:St1}), we have $\texttt{Est}(x^1)$ must terminate after $\frac{1}{\sqrt{2}}p(d)\sqrt{\frac{3}{2}w_0}N^{1/4}<p(d)\sqrt{w_0}N^{1/4}$ steps, which is a contradiction. 

Thus, $R_{s,t}(G_{\textsc{nand}_d}(x^0))\leq 2R_{s,t}(G_{\textsc{nand}_d}(x^1))$, so outputting 0 is correct. Furthermore, since we terminate after $p(d)\sqrt{w_0}N^{1/4}=\widetilde{O}(\sqrt{R_{s,t}(G_{\textsc{nand}_d}(x^0))}N^{1/4})$ steps, and since $R_{s,t}(G_{\textsc{nand}_d}(x^0))=O(R_{s,t}(G_{\textsc{nand}_d}(x^1)))$, the running time is at most $\widetilde{O}\left(N^{1/4}\sqrt{w_{\min}}\right)$.

We now consider the second case, in which both programs output estimates $w_0$ and $w_1$, such that $|w_b-R_{s,t}(G_{\textsc{nand}_d}(x^b))|\leq \eps R_{s,t}(G_{\textsc{nand}_d}(x^b))$ for $b=0,1$. Suppose $w_b\leq w_{\bar b}$. We then have
\begin{equation}
\frac{R_{s,t}(G_{\textsc{nand}_d}(x^b))}{R_{s,t}(G_{\textsc{nand}_d}(x^{\bar b}))}\leq \frac{R_{s,t}(G_{\textsc{nand}_d}(x^b))}{w_b}\frac{w_{\bar b}}{R_{s,t}(G_{\textsc{nand}_d}(x^{\bar{b}}))}\leq \frac{1+\eps}{1-\eps} = \frac{4/3}{2/3}=2. 
\end{equation}
Thus $R_{s,t}(G_{\textsc{nand}_d}(x^b))\leq 2R_{s,t}(G_{\textsc{nand}_d}(x^{\bar b}))$, as required. Furthermore, the running time of the algorithm is bounded by the running time of $\texttt{Est}(x^1)$, the second to terminate. We know that $\texttt{Est}(x^1)$ has running time at most $\widetilde{O}\left(\sqrt{R_{s,t}(G_{\textsc{nand}_d}(x^1))}N^{1/4}\right)$ steps, and by assumption, $\texttt{Est}(x^1)$ terminated after less than $p(d)\sqrt{w_0}N^{1/4}=\widetilde{O}\left(\sqrt{R_{s,t}(G_{\textsc{nand}_d}(x^0))}N^{1/4}\right)$ steps, so the total running time is at most $\widetilde{O}\left(N^{1/4}\sqrt{w_{\min}}\right)$.
\end{proof}

\win*

\begin{proof}
First note that Player $A$ must make $O(d)$ choices
over the course of the game. We amplify Player $A$'s probability
of success by repeating $\mathtt{Select}$ at each decision node $O(\log d)$
times and taking the majority. Then the probability that Player
$A$ chooses the wrong direction at any node is $O(1/d)$, and we ensure that her 
probability of choosing the wrong direction over the course of the algorithm
is $<1/3$. From here on, we analyze the error free case.

Let $p(d)$ be a non-decreasing polynomial function in $d$ such that $\mathtt{Select}$, on inputs $x^0,x^1\in\{0,1\}^{2^d}$, terminates in at most $p(d)2^{d/4}\sqrt{\min\{R_{s,t}(G_{\textsc{nand}_d}(x^0)),R_{s,t}(G_{\textsc{nand}_d}(x^1))\}}$ queries. Then we will prove that for trees of odd depth $d$, the expected number of queries by Player $A$ over the course of the game is at most $p(d)2^{d/4+5}\sqrt{R_{s,t}(G_{\textsc{nand}_{d}}(x))}$,
while for even depth trees, it is at most 
$p(d)2^{d/4+11/2}\sqrt{R_{s,t}(G_{\textsc{nand}_{d}}(x))}$, thus proving the main result.

We prove the result by induction on the depth of the tree. For depth zero trees, there are no decisions, $N=R_{s,t}G_{\textsc{nand}_0}(x)=1$, so the result holds.

For the inductive case, we treat odd and even depth cases separately. First consider an instance of $\textsc{nand}_d$ with $d>0$, $d$ odd. Thus $\textsc{nand}_d(x)=\textsc{nand}_{d-1}(x^0)\wedge\textsc{nand}_{d-1}(x^1)$, where $x=(x^0,x^1)$. Because the root is at odd distance from the leaves, the root is a decision node for Player $B$. Because we are in an $A$-winnable tree, no matter which choice Player $B$ makes, we will end up at an $A$-winnable subtree of depth $d-1$, so the inductive assumption holds for those trees. That is, the expected number of queries for Player $A$ must make to win the subtree with input $x^b$ (for $b\in\{0,1\}$) averaged over Player $B$'s choices is at most 
\begin{align}
p(d-1)2^{(d-1)/4+11/2}\sqrt{R_{s,t}(G_{\textsc{nand}_{d-1}}(x^b))}.
\end{align}

We are assuming that Player $B$ chooses left and right with equal probability.  Thus, the expected number of queries that Player $A$ must make over Player $B$'s choices throughout the game is at most
\begin{eqnarray}
&& \frac{1}{2}\Big(p(d-1)2^{(d-1)/4+11/2}\sqrt{R_{s,t}(G_{\textsc{nand}_{d-1}}(x^0))}\notag\\
&&\qquad\qquad+p(d-1)2^{(d-1)/4+11/2}\sqrt{R_{s,t}(G_{\textsc{nand}_{d-1}}(x^1))}\Big)\notag\\
&\leq & p(d-1)2^{(d-1)/4+11/2}\sqrt{\frac{1}{2}\left(R_{s,t}(G_{\textsc{nand}_{d-1}}(x^0))+R_{s,t}(G_{\textsc{nand}_{d-1}}(x^1))\right)}\quad \mbox{by Jensen's ineq.,}\notag\\
&= & p({d-1})2^{(d-1)/4+11/2}\sqrt{\frac{1}{2}R_{s,t}(G_{\textsc{nand}_{d}}(x))}\quad\mbox{by Claim \ref{claim:parallel_series}},\notag\\
&\leq & p\left(d\right)2^{d/4-1/4+11/2-1/2}\sqrt{R_{s,t}(G_{\textsc{nand}_{d}}(x))}\notag\\
&\leq & p\left(d\right)2^{d/4+5}\sqrt{R_{s,t}(G_{\textsc{nand}_{d}}(x))},
\end{eqnarray}
proving the case for odd $d$.

Now consider an instance of $\textsc{nand}_d$ with $d>0$, $d$ even. Thus $\textsc{nand}_d(x)=\textsc{nand}_{d-1}(x^0)\vee\textsc{nand}_{d-1}(x^1)$, where $x=(x^0,x^1)$. Because the root is at even distance from the leaves, the root is a decision node for Player $A$. Player $A$ runs
 $\mathtt{Select}(x^{0},x^{1})$, which returns $b\in\{0,1\}$ such that (by Lemma~\ref{lemma:select})
\begin{align}\label{eq:select_cond}
R_{s,t}(G_{\textsc{nand}^{d-1}}(x^{ b}))\leq 2
R_{s,t}(G_{\textsc{nand}^{d-1}}(x^{ \bar b})),
\end{align}
which requires at most
\begin{align}\label{eq:stepPlay}
\min_{b^*\in\{0,1\}}p\left({d-1}\right)2^{(d-1)/4}\sqrt{R_{s,t}(G_{\textsc{nand}_{d-1}}(x^{b^*}))}
\end{align}
queries.

After making the choice to move to the subtree with input $x^b$, by the inductive assumption, the expected number of queries that Player $A$ need to make throughout the rest of the game (averaged over Player $B$'s choices) is \begin{align}\label{eq:indStep}
p\left({d-1}\right)2^{d/4+5}\sqrt{R_{s,t}(G_{\textsc{nand}_{d-1}}(x^b))}.
\end{align}

There are two cases to consider. If $R_{s,t}(G_{\textsc{nand}_{d-1}}(x^b))\leq R_{s,t}(G_{\textsc{nand}_{d-1}}(x^{\bar{b}}))$, then combining Eq.~\eqref{eq:stepPlay} and Eq.\ \eqref{eq:indStep}, we have that the total number of queries averaged over Player $B$'s choices~is 
\begin{align}
&p\left({d-1}\right)2^{(d-1)/4}\sqrt{R_{s,t}(G_{\textsc{nand}_{d-1}}(x^b))}+p\left({d-1}\right)2^{(d-1)/4+5}\sqrt{R_{s,t}(G_{\textsc{nand}_{d-1}}(x^b))}\nonumber\\
\leq &
p\left({d-1}\right)2^{(d-1)/4}\sqrt{R_{s,t}(G_{\textsc{nand}_{d-1}}(x^b))}(1+2^{5})\nonumber\\
\leq &
p\left({d-1}\right)2^{(d-1)/4+5+1/16}\sqrt{R_{s,t}(G_{\textsc{nand}_{d-1}}(x^b))}
\nonumber\\
\leq & p\left({d-1}\right)2^{(d-1)/4+5+1/16+1/2}\sqrt{R_{s,t}(G_{\textsc{nand}_{d}}(x))}
\nonumber\\
\leq & p\left(d\right)2^{d/4+11/2}\sqrt{R_{s,t}(G_{\textsc{nand}_{d}}(x))}
\end{align}
where we've used 
$R_{s,t}(G_{\textsc{nand}_d}(x))=\left({R_{s,t}(G_{\textsc{nand}_{d-1}}(x^0))}^{-1}+{R_{s,t}(G_{\textsc{nand}_{d-1}}(x^1))}^{-1}\right)^{-1}$
from Claim \ref{claim:parallel_series} and the fact that $R_{s,t}(G_{\textsc{nand}_{d-1}}(x^b))\leq R_{s,t}(G_{\textsc{nand}_{d-1}}(x^{\bar{b}}))$ to bound the value $R_{s,t}(G_{\textsc{nand}_{d-1}}(x^b))$ by $2R_{s,t}(G_{\textsc{nand}_{d}}(x))$. This proves the even induction step for this case.

The other case is if $R_{s,t}(G_{\textsc{nand}_{d-1}}(x^b))> R_{s,t}(G_{\textsc{nand}_{d-1}}(x^{\bar{b}}))$. In that case, 
again using the fact that $R_{s,t}(G_{\textsc{nand}_d}(x))=\left({R_{s,t}(G_{\textsc{nand}_{d-1}}(x^0))}^{-1}+{R_{s,t}(G_{\textsc{nand}_{d-1}}(x^1))}^{-1}\right)^{-1}$, we have
\begin{equation}
R_{s,t}(G_{\textsc{nand}_{d-1}}(x^{\bar{b}}))=R_{s,t}(G_{\textsc{nand}_d}(x))\left(1+\frac{R_{s,t}(G_{\textsc{nand}_{d-1}}(x^{\bar{b}}))}{R_{s,t}(G_{\textsc{nand}_{d-1}}(x^{{b}}))}\right)^{-1}\leq \frac{2}{3}R_{s,t}(G_{\textsc{nand}_d}(x)),
\end{equation}
where the inequality follows from Eq.\ \eqref{eq:select_cond}. Thus, the average total number of queries is
\begin{align}
&p\left({d-1}\right)2^{(d-1)/4}\sqrt{R_{s,t}(G_{\textsc{nand}_{d-1}}(x^{\bar{b}}))}+p\left({d-1}\right)2^{(d-1)/4+5}\sqrt{R_{s,t}(G_{\textsc{nand}_{d-1}}(x^b))}\nonumber\\
\leq & p(d-1)2^{(d-1)/4}\left(\sqrt{R_{s,t}(G_{\textsc{nand}_{d-1}}(x^{\bar{b}}))}+2^5\sqrt{2R_{s,t}(G_{\textsc{nand}_{d-1}}(x^{\bar{b}}))}\right)\nonumber\\
\leq & p(d-1)2^{(d-1)/4}(1+2^{5+1/2})\sqrt{\frac{2}{3}R_{s,t}(G_{\textsc{nand}_{d}}(x))}\nonumber\\
\leq & p(d)2^{d/4 -1/4+5}\sqrt{R_{s,t}(G_{\textsc{nand}_{d}}(x))}\nonumber\\
\leq & p(d)2^{d/4+5}\sqrt{R_{s,t}(G_{\textsc{nand}_{d}}(x))}.
\end{align}
This proves the induction step for the other case.
\end{proof}

\end{document}